\documentclass[sigconf,balance=false]{acmart}

\AtBeginDocument{%
  }

\setcopyright{acmlicensed}
\copyrightyear{2018}
\acmYear{2018}
\acmDOI{XXXXXXX.XXXXXXX}
\acmConference[Conference acronym 'XX]{Make sure to enter the correct
  conference title from your rights confirmation email}{June 03--05,
  2018}{Woodstock, NY}
\acmISBN{978-1-4503-XXXX-X/2018/06}

\usepackage{float}
\usepackage{graphicx}
\usepackage{enumerate}
\usepackage{enumitem}
\usepackage{subfigure}
\usepackage{multirow}
\usepackage{color}
\usepackage{mathtools}
\usepackage{epstopdf}
\usepackage{balance}
\usepackage{tabularx}

\usepackage[linesnumbered,vlined,ruled,noend]{algorithm2e}

\newcommand{\cut}[1]{}

\newcommand{\todo}[1]{\textcolor{red}{TODO: #1}}

\newcommand{\xy}[1]{\textcolor{cyan}{#1}}
\newcommand{\rv}[1]{\textcolor{black}{#1}}

\newcommand{\rone}[1]{\textcolor{black}{#1}}
\newcommand{\rtwo}[1]{\textcolor{black}{#1}}
\newcommand{\rthree}[1]{\textcolor{black}{#1}}

\newcommand{\shepherd}[1]{\textcolor{black}{#1}}

\newtheorem{theorem}{Theorem}
\newtheorem{definition}{Definition}

\newtheorem{example}{Example}
\newtheorem{lemma}{Lemma}

\newtheorem{assumption}{Assumption}

\newtheorem{procedure}{Procedure}

\DeclareMathOperator{\cost}{cost}

\DeclareMathOperator{\argmin}{argmin}

\DeclareMathOperator{\default}{def}
\DeclareMathOperator{\progressive}{pg}
\DeclareMathOperator{\roundrobin}{rr}

\begin{document}


\title{Query Performance Tuning with Optimal Exploration of Optimizer Cost Model Parameter Space}

\author{Wentao Wu, Xiaoying Wang, Vivek Narasayya, Surajit Chaudhuri}
\affiliation{
  \institution{Microsoft Research, Redmond, USA\\ 
\{wentao.wu, wangxiaoying, viveknar, surajitc\}@microsoft.com
}
  \city{}
  \country{}
}






\begin{abstract}
Modern query optimizers use analytical cost models to estimate the cost of a given query plan.
Such cost models are typically functions of a set of ``cost units'' that specify \emph{unit} CPU cost when processing a row or \emph{unit} IO cost when accessing a disk page. 
These cost units are traditionally viewed as \emph{platform-dependent constants}, that is, they require a one-shot calibration when a database is deployed on a hardware/software platform, but are fixed afterward regardless of the query being optimized for.
\rthree{Some very recent work has taken a different perspective by viewing these cost units as tunable parameters that we call ``cost model parameters (CMPs)'' in this paper.}
\shepherd{However, so far there is no approach that offers any optimality guarantee for the tuning results.}
We present a new approach to systematically explore the query plan space spanned by the CMPs and find the best plan in terms of execution time.
Compared to alternative exploration approaches that use random search (RS) or Bayesian optimization (BO), 
our new approach is \emph{deterministic} and, therefore, avoids the undesirable instability that is inevitable when applying RS or BO.
Moreover, it is guaranteed to find all candidate plans in the query plan space without suffering from the overhead of an exhaustive enumeration.
We also present a set of optimization \rthree{techniques} to reduce the overall evaluation time spent on executing the candidate plans found, a factor that is often overlooked by previous work but is critical from a practical point of view.
Experimental evaluation on top of PostgreSQL \rthree{and Microsoft SQL Server} 
demonstrates the efficacy of tuning the CMPs, which can find query plans that are orders of magnitude faster in execution time than the ones found by RS or BO.
\end{abstract}

\settopmatter{printfolios=true}
\maketitle



\section{Introduction}
\label{sec:intro}

Modern query optimizers use analytical cost models to estimate the cost of a given query plan~\cite{WuCZTHN13,WuW24,YangWZDLC23}.
Such cost models are typically functions of a set of ``cost units'' that specify \emph{unit} CPU cost when processing a row or \emph{unit} IO cost when accessing a disk page.
As a concrete example, Table~\ref{tab:pgsql-cost-units} presents the cost units that are currently used by PostgreSQL's query planner/optimizer~\cite{pgsql-cost-units}.
This paradigm is also followed by other mainstream query processing engines, such as Microsoft SQL Server~\cite{sql-server-cost-units}, IBM DB2~\cite{db2-cost-units}, and Apache Spark~\cite{spark-cost-units}.
Cost units are traditionally viewed as \emph{platform-dependent constants}, that is, they require a one-shot calibration after the database is created on top of a hardware/software platform, but are fixed afterward regardless of the query being optimized for~\cite{DuKS92,SororMASKK08,WuCHN13,WuWHN14,WuNS16,YangWZDLC23}.
\rthree{Recent work~\cite{WuW24} took a different view by treating cost units as \emph{query-dependent tunable variables} and proposed \emph{tuning} cost units to obtain better query plans in terms of execution time, by noticing that query optimizer can return different query plans based on different settings/values of the cost units.
We term this tuning problem as ``cost model parameter (CMP) tuning'' in the rest of this paper.}
Specifically, given a query and a set of CMPs (as well as a range of valid values for each CMP), the goal of CMP tuning is to find the settings/values of the CMPs that can result in the best query plan with minimum execution time.
A formal problem definition can be found in Section~\ref{sec:problem-formulation:definition}.

\begin{table}
\caption{Examples of cost units, i.e., cost model parameters (CMPs), used by PostgreSQL's query planner/optimizer~\cite{pgsql-cost-units}. 
\rtwo{Each of these CMPs can be set with a SQL statement (e.g., \texttt{SET random\_page\_cost TO 1.5}) per individual query session without disruptive operations such as server restart.}}
\vspace{-1em}
\label{tab:pgsql-cost-units}
\small
\centering
\begin{tabular}
{|l|r|r|}
\hline
\textbf{Cost Unit} & \textbf{Notation} & \textbf{Default Value} \\
\hline
\emph{seq\_page\_cost} & $c_s$ & 1.0\\
\emph{random\_page\_cost} & $c_r$ & 4.0\\
\emph{cpu\_tuple\_cost} & $c_t$ & 0.01\\
\emph{cpu\_index\_tuple\_cost} & $c_i$ & 0.005\\
\emph{cpu\_operator\_cost} & $c_o$ & 0.0025\\
\hline
\end{tabular}
\vspace{0.5em}
\end{table}

In previous work~\cite{WuW24}, the authors viewed CMP tuning as a hyperparameter optimization (HPO) problem and proposed using random search (RS) and Bayesian optimization (BO) as solutions, which are standard in the literature on HPO and AutoML.
They showed that better query plans can indeed be found.
This benefit of CMP tuning has also been demonstrated in Section~\ref{sec:eval:exec-time} of this paper.
One common shortcoming shared by these standard HPO approaches is their intrinsic \emph{indeterminism}, since they all rely on probabilistic or sampling mechanisms internally.
As a result, the tuning results are not reproducible unless the same seed is always used for the random number generator (RNG) underlying the probabilistic mechanisms.
Another shortcoming, which is a subsumed consequence of indeterminism, is that there is no guarantee of the optimality on the best CMP values found by these HPO approaches.

In this paper, we overcome these shortcomings with a \emph{deterministic} and \emph{optimal} approach to CMP tuning.
Not only does our approach improve the robustness of CMP tuning, it is also guaranteed to find all candidate plans with complete exploration of the plan space considered by the query optimizer when varying the CMP values, without suffering from the prohibitive overhead of an exhaustive enumeration.
Empirically, as shown in our experimental evaluation (Section~\ref{sec:eval:exec-time}), it also finds much faster query plans than randomized approaches such as RS and BO.


Specifically, the basis of our approach is the observation that the analytic cost models used by today's query optimizers are typically \emph{linear} functions of the CMPs~\cite{WuCZTHN13,WuCHN13,WuW24}.
\begin{example}[Cost Model of PostgreSQL]
\label{example:pg-cost-model}
We use PostgreSQL as an example to illustrate this important observation.
As shown in Table~\ref{tab:pgsql-cost-units}, PostgreSQL uses the following cost units or CMPs.
\begin{itemize}
    \item $c_s$: \emph{seq\_page\_cost}, which represents the unit IO cost of \emph{sequentially} accessing a disk page;
    \item $c_r$: \emph{random\_page\_cost}, which represents the unit IO cost of \emph{randomly} accessing a disk page;
    \item $c_t$: \emph{cpu\_tuple\_cost}, which represents the unit CPU cost of processing a tuple;
    \item $c_i$: \emph{cpu\_index\_tuple\_cost}, which represents the unit CPU cost of processing a tuple \emph{via index access};
    \item $c_o$: \emph{cpu\_operator\_cost}, which represents the unit CPU cost of performing an operation such as \emph{hash} or \emph{aggregation}.
\end{itemize}
The estimated cost of a relational operator $O$ in a query plan $P$ considered by PostgreSQL's query optimizer is then represented as a \emph{linear combination} of $c_s$, $c_r$, $c_t$, $c_i$, and $c_o$~\cite{WuCZTHN13}:
$$\cost(O)=\vec{n}_O\cdot\vec{c}=n_s\cdot c_s+n_r\cdot c_r+n_t\cdot c_t + n_i\cdot c_i+n_o\cdot c_o.$$
The vector $\vec{n}_O=(n_s,n_r,n_t,n_i,n_o)$ represents the corresponding estimated amount of IO or CPU work for each CMP.
For example, $n_s$ represents the number of disk pages that are sequentially accessed by $O$, $n_r$ represents the number of disk pages that are randomly accessed by $O$, and so on.
The estimated cost of the plan $P$ is simply the sum of the estimated costs of the individual relational operators:
$$\cost(P)=\sum\nolimits_{O\in P}\cost(O)=\sum\nolimits_{O\in P}\vec{n}_O\cdot\vec{c}=\vec{n}_P\cdot\vec{c},$$
where $\vec{n}_P$ represents the \emph{element-wise} sum of $\vec{n}_O$ for all the operators $O$ contained by the plan $P$.
Therefore, the estimated cost of the plan $P$ is again a linear function of the CMPs.
\end{example}

\rtwo{Note that the linearity of the cost functions with respect to the CMPs holds even for operators whose cost functions that are not linear functions of the input cardinality. For example, the CPU cost function of the \emph{sort} operator is $2\cdot N\cdot\log_2N\cdot c_o$, where $N$ represents the number of input tuples.
We refer the readers to Table 1 of Yang et al.~\cite{YangWZDLC23} for more concrete examples of the cost functions used by PostgreSQL's operators.}
This important property allows us to develop more advanced exploration strategies to navigate the plan space in a more systematic manner compared to RS or BO, as was done in previous work~\cite{WuW24}.
Specifically, we propose an exploration algorithm based on \emph{binary space decomposition} (BSP)~\cite{BSP-Naylor04}.
Unlike RS and BO, this exploration algorithm is \emph{deterministic} and is guaranteed to return \emph{all} possible plans that one could obtain from the query optimizer by navigating the space spanned by the ranges of the CMPs (see Section~\ref{sec:exploration-plan-space}).
Our evaluation in Section~\ref{sec:eval:exploration} shows that this new exploration strategy can discover many more candidate plans than randomized search strategies, which is essential for identifying better query plans with shorter execution time.
For example, the best plans found by our approach for the TPC-H queries running on PostgreSQL improve by 32\% in execution time compared to those identified by randomized approaches (Section~\ref{sec:eval:exec-time}).
We note that the linearity of the plan cost functions has also been exploited in previous work on the calibration of cost units~\cite{WuCZTHN13}.
However, the \rthree{techniques} developed there require changes to the query optimizer to output the $\vec{n}$ in Example~\ref{example:pg-cost-model}, which are not required by the BSP-based plan-space exploration strategy proposed in this paper.
Moreover, calibration of cost units assumes that the $\vec{n}$ outputted by the query optimizer is accurate, which is hardly true in practice given that the $\vec{n}$ depends on cardinality estimation (CE), a well-known difficult problem despite decades of research~\cite{WangQWWZ21}.
\rone{CMP tuning aims to tweak the cost model to obtain query plans \emph{without} fixing CE errors, as this is also the ultimate purpose of improving CE.
Although we expect that injecting true cardinality information into the query optimizer would also produce better query plans, it is not practical due to the prohibitive overhead of obtaining true cardinality for \emph{all} expressions explored during query optimization~\cite{LeeDNC23}.}
This is in a similar spirit of robust query optimization and processing~\cite{yin2015robust,Haritsa24} that aims to mitigate the uncertainty in plan quality due to inaccurate CE by retaining multiple access paths or plans at query compilation time and deferring the selection of access path or plan to query execution, although 
it requires probing the plan search space by varying CE, and thus again modifications to the query optimizer.
From this point of view, CMP tuning offers an alternative and non-intrusive way of obtaining more candidate plans from the query optimizer.



Following \emph{complete} navigation of the plan space, we then focus on reducing the total plan evaluation time required to find the best query plan (in terms of execution time) among the candidates, another important aspect that was somewhat overlooked by existing work until very recently~\cite{YiTIM25,GiannakourisT25}.
The specific challenge we face is that we may have hundreds of candidate plans after exploration.
As a result, simply letting all candidate plans run to completion may not be feasible.
We therefore further propose two classes of optimization \rthree{techniques}, one based on timeout mechanisms and the other based on similarity between candidate plans (see Section~\ref{sec:optimization-tuning-time}).
The timeout mechanisms have been used in a different context to find the best index configurations among the candidates~\cite{GiannakourisT25}, and we adapt them to the problem of tuning CMPs in this paper.
The similarity-based mechanisms, in particular the one based on identifying ``bottleneck operators'' in query plans, are our new proposals and we show that they can significantly reduce the plan evaluation time compared to timeout-only baselines (Section~\ref{sec:eval:tuning-time}).

\begin{figure}[t]
\centering
    \includegraphics[clip, trim=3.5cm 4cm 3.5cm 3.5cm, width=\columnwidth]{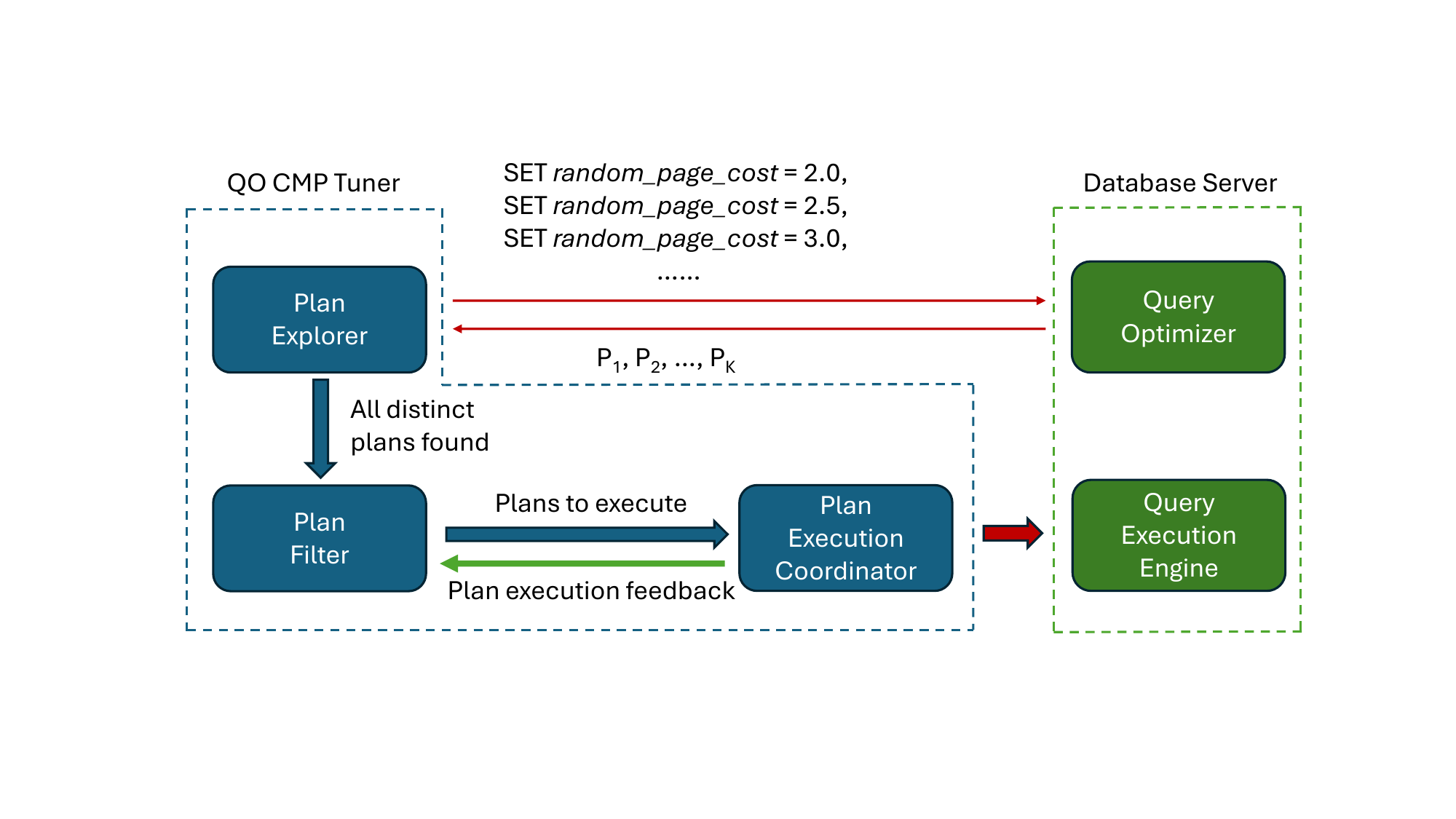}
\vspace{-2em}
\caption{The architecture of tuning query optimizer (QO) cost model parameters (CMP).}
\label{fig:architecture}
\end{figure}

Figure~\ref{fig:architecture} presents an overview of the architecture for CMP tuning, by putting together all the components that we have outlined above.
Specifically, the CMP tuner contains three modules: (1) \emph{plan explorer}, (2) \emph{plan filter}, and (3) \emph{plan execution coordinator}.
The \emph{plan explorer} takes the query and the ranges of the CMPs as input and searches for all possible plans within the space, using the \rthree{techniques} described in Section~\ref{sec:exploration-plan-space}.
It makes a query optimizer call for each explored setting of the CMPs to obtain the corresponding plan.
Figure~\ref{fig:architecture} highlights this interaction between the CMP tuner and the query optimizer with a simplified example of tuning one single CMP ``\emph{random\_page\_cost}'' of PostgreSQL.
The \emph{plan filter} then prunes candidates that do not require execution, based on the plan similarity mechanisms described in Section~\ref{sec:optimization:similarity}.
Finally, the plans that are necessary for execution are submitted to the \emph{plan execution coordinator}, which further leverages the timeout mechanisms described in Section~\ref{sec:optimization:timeout} to reduce the plan evaluation time.


In summary, this paper makes the following contributions:
\begin{itemize}
    \item We propose an optimal approach to the problem of CMP tuning, which systematically navigates the plan space and is guaranteed to find all possible candidate plans.
    \item We propose two classes of optimization mechanisms, one based on timeout and the other based on plan similarity, to reduce the overall candidate plan evaluation time.
    \item We perform extensive experiments using four publicly available benchmark database workloads and demonstrate the efficacy of the proposed \rthree{techniques}.
\end{itemize}

\vspace{-0.5em}
\paragraph*{Limitation and Future Work}
The current focus of CMP tuning is on improving individual queries independently.
Such fine-grained query performance tuning has found its way in the development of automated software~\cite{BelknapDDY09,HerodotouB10} that significantly reduces the manual tuning effort of a database administrator (DBA).
We can further extend the CMP tuning \rthree{techniques} developed in this work to the workload level.
\rone{That is, we can reuse information across queries rather than having a set of common CMP values for all queries in the workload.
For example, the similarity-based plan filtering mechanism can be extended to filter candidate plans from different queries that share similar bottleneck operators.}
Moreover, instead of always tuning the CMPs from scratch for each new query, one may want to directly \emph{predict} a better setting of the CMPs.
This \emph{online} setup is intriguing when there are time and resource constraints that prevent a full-fledged CMP tuning process from finishing.
We leave these interesting but unexplored areas for future work.

\vspace{-0.5em}
\paragraph*{Paper Organization}
The remainder of the paper is organized as follows.
Section~\ref{sec:problem-formulation} formulates the problem of CMP tuning.
Section~\ref{sec:exploration-plan-space} presents optimal exploration algorithms to obtain all possible (candidate) plans by tuning CMPs.
Section~\ref{sec:optimization-tuning-time} further presents optimization \rthree{techniques} to reduce the evaluation time of candidate plans.
We present experimental setups and results in Section~\ref{sec:eval}, discuss related work in Section~\ref{sec:related-work}, and conclude the paper in Section~\ref{sec:conclusion}.


\vspace{-0.5em}
\section{Problem Formulation}
\label{sec:problem-formulation}

\rthree{We start with a view of the cost units by treating them as query-dependent variables rather than platform-dependent constants.
We then present a formal problem definition of ``cost model parameter (CMP) tuning'' under this view of the cost units.}

\vspace{-0.5em}
\subsection{Cost Model Parameters (CMPs)}

We use $\vec{c}$ to represent the vector of all cost units.
Given a query plan $P$, its estimated cost can be viewed as some function $f_P(\vec{c})$ of $\vec{c}$.
Based on previous work~\cite{WuCZTHN13,WuCHN13}, $f_P$ is typically a \emph{linear} function of $\vec{c}$ and can therefore be written as
$$f_P(\vec{c})=\vec{n}_P\cdot\vec{c},$$
where $\vec{n}_P$ represents a vector of plan-dependent constants.
As a result, $\vec{n}_P$ serves as the \emph{coefficients} and $\vec{c}$ serves as the \emph{variables} in the linear function (ref. Example~\ref{example:pg-cost-model}).

This view is very different from the traditional view of the cost function $f_P$, where the roles of $\vec{n}_P$ and $\vec{c}$ are \emph{reversed}.
In the traditional view, $\vec{c}$ represents the (constant) coefficients, while $\vec{n}_P$ are variables that depend on the estimated cardinality.
This traditional view is in particular useful when considering the inaccuracy or uncertainty of cardinality estimation, which has led to extensive research on \emph{robust query optimization and processing}~\cite{yin2015robust,Haritsa24}.

\rthree{We call the cost units under this alternative view of the cost function the ``cost model parameters (CMPs)'' to emphasize that they are no longer platform-dependent constants but tunable parameters for \emph{each query}.} We model a CMP $c$ as a \emph{continuous variable} with \emph{range} $R_c=[c_l, c_h]$ and the default value $c_{\default}$.

\begin{figure}
\centering
    \includegraphics[clip, trim=9.5cm 6.5cm 9.5cm 6cm, width=0.7\columnwidth]{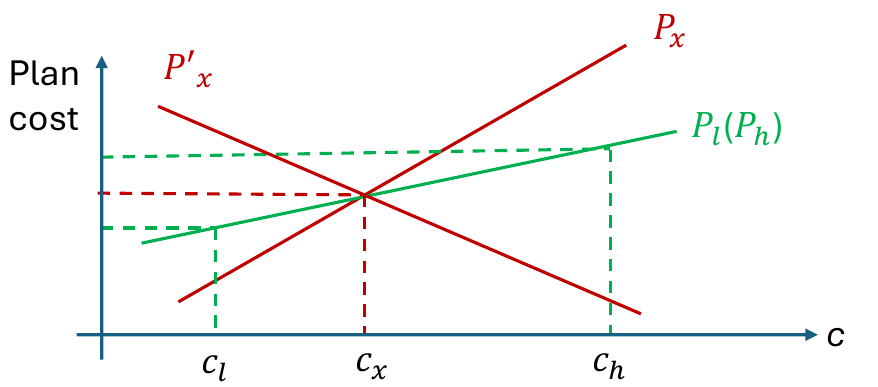}
\vspace{-1em}
\caption{Illustration of the case with only one CMP.}
\label{fig:plan-cost-one-dim}
\end{figure}

\vspace{-0.5em}
\subsection{Problem Definition}
\label{sec:problem-formulation:definition}

Given a query $q$, setting the CMPs $\vec{c}$ to different values can result in different query plans of $q$.
Formally, we define the \emph{plan space} by exploring the possible values of $\vec{c}$ as follows.

\begin{definition}[Plan Space]
\label{def:plan-space}
Let $\vec{c}=(c_1, ..., c_K)$ be the vector of CMPs. Let $R_i$ be the range of the CMP $c_i$.
We define \emph{plan (search) space} as the $K$-dimensional search space $\mathcal{R}=\prod_{i=1}^K R_i$ spanned by the ranges of $\vec{c}$.
We call $K$ the \emph{dimension} of $\mathcal{R}$.
\end{definition}

\begin{definition}[Candidate Plans]
\label{def:candidate-plans}
Consider a given SQL query $q$. For any $\vec{v}\in\mathcal{R}=\prod_{i=1}^K R_i$, let $P(q, \vec{v})$ be the query plan returned by the query optimizer for $q$ by setting $\vec{c}=\vec{v}$.
The set of \emph{candidate plans} for $q$ is defined as the set of all unique plans $\mathcal{P}_q=\{P(q, \vec{v})|\vec{v}\in\mathcal{R}\}$.
\end{definition}

For any query plan $P$, let $t(P)$ be its execution time. We define the problem of ``CMP tuning'' as follows.

\begin{definition}[CMP Tuning]
\label{def:QO-knob-tuning}
Let $q$ be a given SQL query and let $\mathcal{R}$ be the plan space defined by Definition~\ref{def:plan-space}.
The goal of CMP tuning is to find the plan $P^*$ from the set of candidate plans $\mathcal{P}_q$ defined by Definition~\ref{def:candidate-plans} that minimizes the execution time, namely,
$$P^*=\argmin_{P\in\mathcal{P}_q}\{t(P)\}.$$
\end{definition}

Although we focus on minimizing plan execution time in this work, the plan-space exploration techniques proposed in Section~\ref{sec:exploration-plan-space} are independent of the performance metric and therefore can be applied to other metrics (e.g., resource consumption) as well.

\section{Exploration of Plan Space}
\label{sec:exploration-plan-space}

Previous work~\cite{WuW24} proposed using either random search (RS) or Bayesian optimization (BO) to explore the plan space.
One common and undesirable shortcoming of RS or BO is their inherent indeterminism, given that they are probabilistic algorithms.
Moreover, there is no guarantee that the exploration done by RS or BO is complete.
As a result, the true optimal plan can be missed and excluded from the candidate plans found.
In this section, we propose a new \emph{deterministic} and \emph{optimal} approach to explore the plan space for CMP tuning, based on the critical observation that the plan cost functions are \emph{linear} with respect to the CMPs.
Optimality here means the exploration is complete, that is, the true optimal plan is guaranteed to be included in the candidate plans found by exploration.
Throughout this section, we make the following assumptions to simplify our analysis.

\begin{assumption}
\label{assumption:no-same-cost}
Given any nontrivial values of the CMPs $\vec{c}$ (i.e., $\vec{c}\neq\vec{\mathbf{0}}$), different plans have different estimated costs.
\end{assumption}

\begin{assumption}
\label{assumption:well-behaved}
Given the same setting $\vec{v}$ of the CMPs $\vec{c}$, the plan $P(q, \vec{v})$ returned by the query optimizer and its estimated cost remain the same if the optimizer is invoked repeatedly.
\end{assumption}
Assumption~\ref{assumption:well-behaved} excludes nondeterministic behaviors of the query optimizer that may arise in practice.
For example, PostgreSQL will switch from dynamic programming to genetic algorithms when optimizing queries with a large number of joins~\cite{pgsql-qeqo}.
\rtwo{However, we note that this is a specific weakness of PostgreSQL, and there are alternative techniques that can produce deterministic plans for queries with many joined tables~\cite{NeumannR18}.}


\vspace{-0.5em}
\subsection{One-Dimensional Plan Space}
\label{sec:explore:one-dim}

For ease of exposition, we start with a special case by considering tuning only one CMP $c\in R_c=[c_l, c_h]$, where $c_l$ and $c_h$ represent the \emph{low value} and \emph{high value} of the range $R_c$.
Let $P_l$ and $P_h$ be the two plans (of a given query $q$) when setting $c=c_l$ and $c=c_h$.

\begin{theorem}
\label{theorem:one-dim}
Let $c_l<c_x<c_h$ and let $P_x$ be the plan returned by the query optimizer by setting $c=c_x$.
If $P_l=P_h$, then $P_l=P_x=P_h$ under Assumption~\ref{assumption:no-same-cost}.
\end{theorem}

The proof can be found in Appendix~\ref{appendix:proofs:theorem:one-dim}.
\shepherd{Assumption~\ref{assumption:no-same-cost} is required by Theorem~\ref{theorem:one-dim}.
Without it, multiple query plans may have the same estimated cost, and the query optimizer may return any one of them depending on its tie-breaking strategy.
Consequently, $P_x$ in Theorem~\ref{theorem:one-dim} may not be well defined.
Assumption~\ref{assumption:no-same-cost} eliminates this ambiguity by ensuring a unique plan for each CMP configuration $\vec{c}$.}
Figure~\ref{fig:plan-cost-one-dim} illustrates this one-dimensional case in an intuitive way.
Since $P_l=P_h$, the two lines that represent the (linear) cost functions of $P_l$ and $P_h$ coincide.
For any $c_l<c_x<c_h$, the cost function of the plan by setting $c=c_x$ also represents a line given its linearity.
There are two possibilities of this line, as shown in Figure~\ref{fig:plan-cost-one-dim} by the two lines denoted by $P_x$ and $P'_x$ (we have abused the notation by using the plans to also represent their linear cost functions).
Neither is possible, because (1) $P_x$ indicates that the query optimizer would have returned $P_x$ rather than $P_l$ when setting $c=c_l$ and (2) $P'_x$ indicates that the query optimizer would have returned $P'_x$ rather than $P_h$ ($=P_l$) when setting $c=c_h$.

Theorem~\ref{theorem:one-dim} suggests a \emph{bisection} algorithm to find all plans in the plan space given by $R_c=[c_l, c_h]$. Algorithm~\ref{alg:exploration:one-dim} presents the details.
Specifically, given a range $R_c=[c_l, c_h]$, we first check whether the two plans $P_l$ and $P_h$ at the end points of $R_c$ (namely, $c_l$ and $c_h$) are the same.
If so, the search within $R_c$ can be terminated based on Theorem~\ref{theorem:one-dim}.
If not, we then choose the ``middle point'' $c_m$ of $[c_l, c_h]$ and recursively explore the two ranges $[c_l, c_m]$ and $[c_m, c_h]$.

\begin{algorithm}[t]
\small
  \SetAlgoLined
  \KwIn{$q$, the query; $R_c=[c_l, c_h]$, the range of the CMP $c$.}
  \KwOut{$\mathcal{P}_q$, the candidate plans of the query $q$.}
  \SetAlgoLined
  $P_l\leftarrow P(q, c_l)$, $P_h\leftarrow P(q, c_h)$\;
  $\mathcal{P}_q\leftarrow\{P_l, P_h\}$\;
  \uIf{$P_l\neq P_h$}{
    $c_m\leftarrow\frac{c_l+c_h}{2}$\;
    $\mathcal{P}_q^l\leftarrow\texttt{ExplorePlanSpaceBS}(q, [c_l, c_m])$\;
    $\mathcal{P}_q^h\leftarrow\texttt{ExplorePlanSpaceBS}(q, [c_m, c_h])$\;
    $\mathcal{P}_q\leftarrow\mathcal{P}_q\cup\mathcal{P}_q^l\cup\mathcal{P}_q^h$\;
  }
  \Return{$\mathcal{P}_q$\;}
  \caption{\texttt{ExplorePlanSpaceBS}($q$, $R_c=[c_l, c_h]$)}
\label{alg:exploration:one-dim}
\end{algorithm}

\vspace{-0.5em}
\subsection{Multi-Dimensional Plan Space}
\label{sec:explore:multi-dim}

\begin{figure}
\centering
    \includegraphics[clip, trim=9.5cm 7cm 9.5cm 5.5cm, width=0.8\columnwidth]{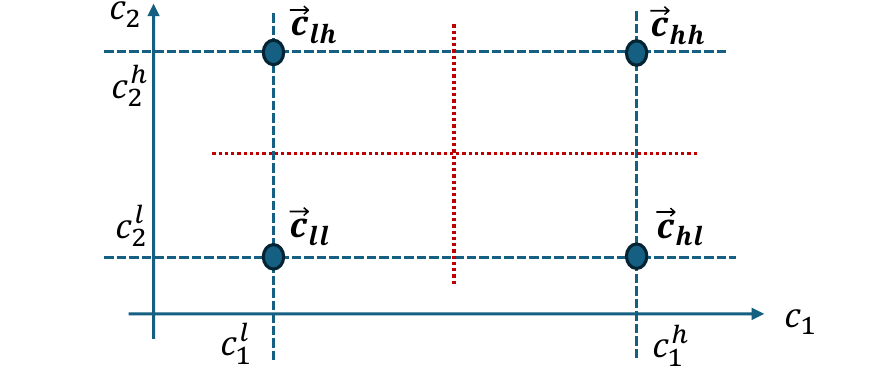}
\vspace{-1.5em}
\caption{Illustration of the case with multiple CMPs.}
\label{fig:plan-cost-multi-dim}
\end{figure}

We now extend our discussion to the general case with multiple CMPs.
Figure~\ref{fig:plan-cost-multi-dim} illustrates the idea with an example of two CMPs $c_1$ and $c_2$.
Suppose that the range of $c_1$ is $R_1=[c_1^l, c_1^h]$ and the range of $c_2$ is $R_2=[c_2^l, c_2^h]$.
We first check the four vertices $\vec{c}_{ll}=(c_1^l, c_2^l)$, $\vec{c}_{lh}=(c_1^l, c_2^h)$, $\vec{c}_{hl}=(c_1^h, c_2^l)$, and $\vec{c}_{hh}=(c_1^h, c_2^h)$ of the rectangle that covers $R_1$ and $R_2$ in the two-dimensional plan space.
If the four plans obtained at these four vertices are the same, we stop exploring the plan space.
Otherwise, we perform a \emph{binary space partitioning} (BSP)~\cite{BSP-Naylor04} of the rectangle, which ends up with four smaller rectangles that represent four subspaces.
We then recursively explore these subspaces to find more plans.

\begin{theorem}
\label{theorem:multi-dim}
Let $\mathcal{R}$ be the \emph{hyperrectangle} that represents the plan space spanned by the CMPs $\vec{c}$, that is, $\mathcal{R}=\prod_{i=1}^K R_i$ where $R_i=[c_i^l, c_i^h]$ for $1\leq i\leq K$.
Let the set of all \emph{vertices} of $\mathcal{R}$ be $\mathcal{V}$.
If the query plan $P_q=P(q, \vec{v})$ remains the same for \emph{all} the vertices $\vec{v}\in\mathcal{V}$, then for \emph{any} settings of the CMPs $\vec{c}=\vec{u}\in\mathcal{R}$ we have $P(q, \vec{u})=P_q$.
\end{theorem}

The proof can be found in Appendix~\ref{appendix:proofs:theorem:multi-dim}.
Theorem~\ref{theorem:multi-dim} guarantees the correctness of extending the idea outlined above from two-dimensional plan space to general multi-dimensional plan space.
Algorithm~\ref{alg:exploration:multi-dim} presents in detail this BSP-based extension of Algorithm~\ref{alg:exploration:one-dim} to explore multidimensional plan space.
We first obtain all the vertices of the $K$-dimensional hyperrectangle that represent the plan space $\mathcal{R}$.
We then check the query plans obtained at the vertices.
If \emph{all} these plans are the same, then we can stop further exploration. 
Otherwise, we invoke the BSP procedure presented by Algorithm~\ref{alg:bsp} to partition the plan space $\mathcal{R}$ into subspaces.
We further recursively explore each subspace.

The BSP procedure illustrated in Algorithm~\ref{alg:bsp} is also recursive.
We start from the first dimension or range $R_1=[c_1^l, c_1^h]$ of the plan space $\mathcal{R}$.
We split the plan space by half with respect to $R_1$.
The result is two subspaces, and we then split each subspace with respect to the second dimension or range $R_2=[c_2^l, c_2^h]$.
We continue doing this in a recursive way until we reach the last dimension or range $R_K=[c_K^l, c_K^h]$.
All the final subspaces are the results of BSP.

\begin{algorithm}[t]
\small
  \SetAlgoLined
  \KwIn{$\mathcal{R}$, the plan space; $K$, the number of CMPs.}
  \KwOut{$\mathbb{S}$, the partitioned subspaces.}
  \SetAlgoLined
    \underline{\textbf{BinarySplit}}($\mathcal{R}=\prod_{i=1}^K R_i$ where $R_i=[c_i^l, c_i^h]$, $k$, $K$, $\mathbb{S}$):\\
    \uIf{$k > K$}{
        $\mathbb{S}\leftarrow\mathbb{S}\cup\{\mathcal{R}\}$; // No more partitioning required.\\
    }\Else{
        $\mathcal{R}_l\leftarrow\texttt{Copy}(\mathcal{R})$, $\mathcal{R}_r\leftarrow\texttt{Copy}(\mathcal{R})$\;
        $c_k^m\leftarrow\frac{c_k^l+c_k^h}{2}$; // Split the $k$-th range $R_k=[c_k^l, c_k^h]$ by half.\\
        $\mathcal{R}_{l,k}\leftarrow[c_k^l, c_k^m]$; // Replace $k$-th range in $\mathcal{R}_l$ by $[c_k^l, c_k^m]$.\\
        $\mathcal{R}_{r,k}\leftarrow[c_k^m, c_k^h]$;  // Replace $k$-th range in $\mathcal{R}_r$ by $[c_k^m, c_k^h]$.\\
        \texttt{BinarySplit}($\mathcal{R}_l$, $k+1$, $K$, $\mathbb{S}$); // Recursion with $\mathcal{R}_l$.\\
        \texttt{BinarySplit}($\mathcal{R}_r$, $k+1$, $K$, $\mathbb{S}$); // Recursion with $\mathcal{R}_r$.\\           
    }
    \quad\\
    \underline{\textbf{Main}}:\\
    $\mathbb{S}\leftarrow\emptyset$\;
    \texttt{BinarySplit}($\mathcal{R}$, 1, $K$, $\mathbb{S}$)\;
    \Return{$\mathbb{S}$\;}
  \caption{\texttt{BinarySpacePartition}($\mathcal{R}$, $K$)}
\label{alg:bsp}
\end{algorithm}

\begin{algorithm}[t]
\small
  \SetAlgoLined
  \KwIn{$q$, the query; $\mathcal{R}$, the plan space with $K$ CMPs.}
  \KwOut{$\mathcal{P}_q$, the candidate plans of $q$.}
  \SetAlgoLined
  $\mathcal{V}\leftarrow$\texttt{GetVertices}($\mathcal{R}$)\;
  $\mathcal{P}_v\leftarrow\{P(q,\vec{v})|\vec{v}\in\mathcal{V}\}$; // Get all plans at vertices of $\mathcal{R}$.\\
  $\mathcal{P}_q\leftarrow\mathcal{P}_v$\;
  \uIf{$|\mathcal{P}_v|>1$}{
    $\mathbb{S}\leftarrow$\texttt{BinarySpacePartition}($\mathcal{R}$, $K$); // Use Algorithm~\ref{alg:bsp}.\\
    \ForEach{$\mathcal{S}\in\mathbb{S}$}{
        $\mathcal{P}_S\leftarrow$\texttt{ExplorePlanSpaceBSP}($q$, $\mathcal{S}$); // Recursion with $\mathcal{S}$.\\
        $\mathcal{P}_q\leftarrow\mathcal{P}_q\cup\mathcal{P}_S$\;
    }
  }
  \Return{$\mathcal{P}_q$\;}
  \caption{\texttt{ExplorePlanSpaceBSP}($q$, $\mathcal{R}$)}
\label{alg:exploration:multi-dim}
\end{algorithm}

\subsubsection{Complexity Analysis}
For a plan space $\mathcal{R}$ with $K$ CMPs, the invocation of BSP would result in $2^K$ subspaces.
Moreover, since each subspace represents a $K$-dimensional hyperrectangle, it has $2^K$ vertices.
As a result, if we invoke the BSP $M$ times, we have 
$$2^K\cdot (1+2^K + 2^{2K}+\cdots+2^{MK})=2^K\cdot\frac{2^{(M+1)K}-1}{2^K-1}\approx 2^{(M+1)K}$$
vertices with as many query optimizer calls to obtain the corresponding query plans.
This number grows exponentially with respect to \emph{both} $M$ and $K$.
Since $K$ is given (e.g., $K=5$ for PostgreSQL), from a practical point of view we need to control $M$ so that the total number of query optimizer calls can be affordable.
In our experimental evaluation, we found that setting $M=2$ is sufficient for exploring the plan spaces of all the queries evaluated.

\subsubsection{Further Optimization}
\label{sec:explore:multi-dim:optimality}

Although Algorithm~\ref{alg:exploration:multi-dim} can find \emph{all} candidate plans in the plan space, it is not the most efficient or optimal in the sense that it makes more query optimizer calls than necessary.
Instead of using a fixed space partitioning strategy like the BSP in Algorithm~\ref{alg:bsp}, we can use a more adaptive partitioning strategy similar to the one proposed in previous work~\cite{HulgeriS02}.
The idea is to only check the \emph{intersection} points of two (linear) plan cost functions, and partition the plan space accordingly based on the plans found at those points.
As an example, consider the one-dimensional plan space in Figure~\ref{fig:intersection-point}.
If the two plans obtained at $c_l$ and $c_h$ are different, then instead of checking the middle point $c_m$ of $[c_l, c_h]$, we can find the intersection point $c_z$ of the two lines that represent the (linear) cost functions of the two plans $P_l$ and $P_h$.
There are two possible cases: \textbf{(C1)} the query optimizer returns either $P_l$ or $P_h$ at $c_z$; or \textbf{(C2)} the query optimizer returns a different plan $P_z$ at $c_z$.
In the first case \textbf{(C1)}, we can skip further checking the entire range $[c_l, c_h]$ by applying Theorem~\ref{theorem:one-dim}:
if $P_l$ is returned at $c_z$, $P_l$ is the unique plan in $[c_l, c_z]$, while $P_h$ is the unique plan in $(c_z, c_h]$;
on the other hand, if $P_h$ is returned at $c_z$, $P_l$ is the unique plan in $[c_l, c_z)$, while $P_h$ is the unique plan in $[c_z, c_h]$.
In the second case \textbf{(C2)}, we can further find the intersection points of the (linear) cost function of $P_z$ with the (linear) cost functions of $P_l$ \emph{and} $P_h$ and repeat the above process.


One problem of implementing this intersection-point based idea is that we need the \emph{explicit expressions} of the linear cost functions, which requires nontrivial changes to the query optimizer~\cite{WuCZTHN13}.
We therefore leave the implementation of this intriguing idea to reduce the number of query optimizer calls for future work.


\vspace{-0.5em}
\section{Reducing Evaluation Time of Candidate Plans}
\label{sec:optimization-tuning-time}

The plans found during the exploration of the plan space, as described in Section~\ref{sec:exploration-plan-space}, consist of the \emph{candidate set} to find the query plan with minimum execution time.
A naive solution is to evaluate all candidate plans, which would be inefficient in the presence of long-running candidates.
In this section, we propose two sets of optimization \rthree{techniques} that significantly reduce the evaluation time spent on running the candidate plans.
The first set of optimizations is based on \emph{timeout}, while the second set of optimizations is based on skipping \emph{similar} candidate plans.
The problem of reducing plan evaluation time is critical from a practical point of view.
However, it was somewhat overlooked until recently~\cite{YiTIM25,GiannakourisT25}.

\vspace{-0.5em}
\subsection{Timeout-based Pruning}
\label{sec:optimization:timeout}

The first set of optimization \rthree{techniques} is based on \emph{timeout}. We present two such \rthree{techniques} in the following.

\begin{figure}
\centering
    \includegraphics[clip, trim=9.5cm 6.5cm 9.5cm 6cm, width=0.7\columnwidth]{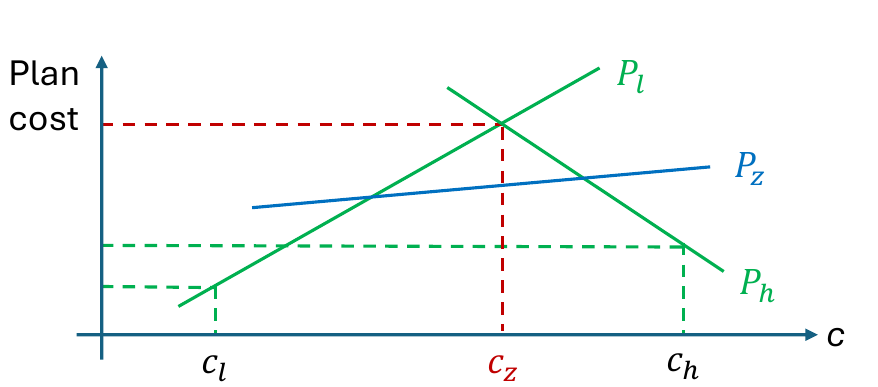}
\vspace{-1em}
\caption{Illustration of checking the intersection point.}
\label{fig:intersection-point}
\end{figure}

\subsubsection{Progressive Timeout}
\label{sec:optimization:timeout:progressive}

\begin{algorithm}[t]
\small
  \SetAlgoLined
  \KwIn{$q$, the query; $\mathcal{P}_q$, the candidate plans of $q$; $P_q^{\default}$, the default plan of $q$ with known execution time.}
  \KwOut{$P_q^*$, the candidate plan with the minimum execution time.}
  \SetAlgoLined
  $P_q^*\leftarrow P_q^{\default}$, $\delta\leftarrow t(P_q^{\default})$\;
  \ForEach{$P_q\in\mathcal{P}_q$}{
    Execute $P_q$ and monitor its execution time $t(P_q)$\;
    \uIf{$P_q$ finishes before hitting the timeout $\delta$}{
        $P_q^*\leftarrow P_q$\;
        $\delta\leftarrow t(P_q)$\;
    }
    \Else{
        Terminate the execution of $P_q$ when $t(P_q)\geq\delta$\;
    }
  }
  \Return{$P_q^*$\;}
  \caption{\texttt{ProgressiveTimeout}($q$, $\mathcal{P}_q$, $P_q^{\default}$)}
\label{alg:optimization:progressive-timeout}
\end{algorithm}

Algorithm~\ref{alg:optimization:progressive-timeout} illustrates a simple idea by progressively resetting the timeout $\delta$ based on the best plan observed so far.
It starts from a timeout that is equal to the execution time of the default plan.
If a candidate plan finishes before the current timeout, then we record its execution time and use that as the new timeout.
Otherwise, the plan runs out of time and we immediately terminate its execution.

\subsubsection{Slice-based Round-robin Timeout}
\label{sec:optimization:timeout:rr}

One primary drawback of progressive timeout is that the total evaluation time depends on the \emph{order} of the plans being executed.
To address this order-dependent issue, recent work~\cite{GiannakourisT25} proposed a slice-based round-robin timeout mechanism.
Although it was proposed to evaluate candidate index configurations, we can adapt it to our problem of CMP tuning for the new purpose of evaluating candidate plans.

\begin{algorithm}[t]
\small
  \SetAlgoLined
  \KwIn{$q$, the query; $\mathcal{P}_q$, the candidate plans of $q$; $P_q^{\default}$, the default plan of $q$ with known execution time; $\delta_0$, the initial timeout.}
  \KwOut{$P_q^*$, the candidate plan with the minimum execution time.}
  \SetAlgoLined
  $P_q^*\leftarrow P_q^{\default}$, $t_q^*\leftarrow t(P_q^{\default})$, $\delta\leftarrow \delta_0$\;
  \While{\texttt{true}}{
      \ForEach{$P_q\in\mathcal{P}_q$}{
        Execute $P_q$ and monitor its execution time $t(P_q)$\;
        \uIf{$P_q$ finishes before hitting the timeout $\delta$}{
            \uIf{$t(P_q)<t_q^*$}{
                $P_q^*\leftarrow P_q$, $t_q^*\leftarrow t(P_q)$, $\delta\leftarrow t_q^*$\;
            }
        }
        \Else{
            Terminate the execution of $P_q$ when $t(P_q)\geq\delta$\;
        }
     }
     \uIf{There exists some plan that finishes in round $r$}{
        \textbf{break}\;
     }
     $\delta\leftarrow \min\{2\cdot\delta, t_q^*\}$\;
  }
  \Return{$P_q^*$\;}
  \caption{\texttt{RoundRobinTimeout}($q$, $\mathcal{P}_q$, $P_q^{\default}$, $\delta_0$)}
\label{alg:optimization:roundrobin-timeout}
\end{algorithm}

Algorithm~\ref{alg:optimization:roundrobin-timeout} presents the details of this round-robin timeout mechanism.
In each round, we double the timeout $\delta$ and restart the execution of the remaining plans that cannot be finished in previous rounds before hitting the timeout.
We keep running this round-robin process until some plan has been finished in the current round.
We simply return the best plan with the minimum execution time among the ones that have been finished.

\vspace{-0.5em}
\subsubsection{Comparison and Analysis}

The advantages of Algorithm~\ref{alg:optimization:roundrobin-timeout} over Algorithm~\ref{alg:optimization:progressive-timeout} are twofold.
First, it alleviates the order-dependent issue of Algorithm~\ref{alg:optimization:progressive-timeout}, since a query plan $P_q$ can be finished in the round with $\delta>t(P_q)$.
Clearly, if multiple plans finish in the same round, then the order among these plans still matters, though the impact is milder.
Second, if the distribution of plan execution time is skewed (i.e., some plans are much faster than the others), then we can find the best plan quickly.
On the other hand, the drawback of Algorithm~\ref{alg:optimization:roundrobin-timeout} compared to Algorithm~\ref{alg:optimization:progressive-timeout} is also clear---if all plans finish in similar time then Algorithm~\ref{alg:optimization:roundrobin-timeout} can waste more time in the early rounds by repeatedly terminating and restarting execution of the same plans. In the following, we present a formal analysis.

Suppose that there are $N$ unique candidate plans in $\mathcal{P}_q$.
Let the minimum execution time of the candidate plans be $t_{\min}$ and let the maximum execution time be $t_{\max}$.
\begin{theorem}
\label{theorem:progressive-timeout}
The total plan evaluation time $T^{\progressive}$ by running the progressive timeout mechanism in Algorithm~\ref{alg:optimization:progressive-timeout} satisfies
\begin{equation}\label{eq:T-pg-bounds}
  N\cdot t_{\min}\leq T^{\progressive}\leq N\cdot t_{\max}.  
\end{equation}
\end{theorem}

\begin{theorem}
\label{theorem:rr-timeout}
The total plan evaluation time $T^{\roundrobin}$ by running the round-robin timeout mechanism in Algorithm~\ref{alg:optimization:roundrobin-timeout} satisfies
\begin{equation}\label{eq:T-rr-bounds}
    N\cdot(\delta_m-\delta_0)\leq T^{\roundrobin}\leq N\cdot(2\cdot\delta_m-\delta_0),
\end{equation}
or equivalently,
$$N\cdot\delta_0\cdot(2^{m}-1)\leq T^{\roundrobin}\leq N\cdot\delta_0\cdot(2^{m+1}-1),$$
where $\delta_m$ is the timeout of the first round $m$ with $t_{\min}\leq\delta_m$.
\end{theorem}

The proofs can be found in Appendices~\ref{appendix:proofs:theorem:progressive-timeout} and~\ref{appendix:proofs:theorem:rr-timeout}.
To simplify our analysis, assume $\delta_0\ll\delta_m$ (e.g., we set $\delta_0$ to 1 second in our evaluation).
We can then simplify Equation~\ref{eq:T-rr-bounds} as
\begin{equation}\label{eq:T-rr-bounds-simplified}
    N\cdot\delta_m\leq T^{\roundrobin}\leq 2\cdot N\cdot\delta_m.
\end{equation}
Since $t_{\min}\leq\delta_m$, by Equation~\ref{eq:T-rr-bounds-simplified} we have $T^{\roundrobin}\geq N\cdot\delta_m\geq N\cdot t_{\min}$.
Therefore, the lower bound (i.e., the best case) for $T^{\roundrobin}$ is no better than that for $T^{\progressive}$ as shown in Equation~\ref{eq:T-pg-bounds}.
Notice that $\delta_m$ can be as large as twice of $t_{\min}$, indicating that in practice we should expect a worse best-case situation of $T^{\roundrobin}$ than that of $T^{\progressive}$.
On the other hand, the upper bound (i.e., the worst case) for $T^{\roundrobin}$ is better than that for $T^{\progressive}$ if $2\cdot\delta_m\leq t_{\max}$, or equivalently, $\delta_m\leq\frac{t_{\max}}{2}$.
Since $t_{\min}\leq\delta_m$, this implies $t_{\min}\leq \frac{t_{\max}}{2}$.
Therefore, we would expect a better worst-case situation of $T^{\roundrobin}$ than that of $T^{\progressive}$ if $t_{\max}\geq 2\cdot t_{\min}$.

\vspace{-0.5em}
\paragraph*{Remark} We can further generalize Algorithm~\ref{alg:optimization:roundrobin-timeout} by setting the timeout $\delta_r$ in round $r$ as $\delta_r=\theta\cdot\delta_{r-1}$ with any $\theta > 1$.
In this generalization, Equation~\ref{eq:T-rr-bounds} can be revised as
\begin{equation}\label{eq:eq:T-rr-bounds-general}
    N\cdot\delta_0\cdot\frac{\theta^{m}-1}{\theta - 1}\leq T^{\roundrobin}\leq N\cdot\delta_0\cdot\frac{\theta^{m+1}-1}{\theta-1},
\end{equation}
by noticing $\sum\nolimits_{r=0}^m\delta_r=\delta_0\cdot\frac{\theta^{m+1}-1}{\theta-1}$.
It is an interesting question to select the $\theta$ that minimizes $T^{\roundrobin}$, which we leave for future work.

\vspace{-0.5em}
\subsection{Similarity-based Pruning}
\label{sec:optimization:similarity}

One common shortcoming of the above two timeout-based pruning techniques is that they view query plans as \emph{black boxes}.
In the context of CMP tuning, many candidate plans are structurally similar --- this is not surprising given that they are all execution plans of the \emph{same} query.
Intuitively, if two plans of the same query are similar, then we expect their execution time be similar as well.
We therefore propose the second set of optimizations based on exploiting the similarity between candidate plans.

\vspace{-0.5em}
\subsubsection{Similarity by Plan Embedding}
\label{sec:optimization:similarity:embedding}

A straightforward idea is to convert a plan into some vectorized feature representation or \emph{embedding}.
The goal is to map structurally similar plans into embedding vectors that are close to each other.

Specifically, let $h:\mathcal{P}_q\to\mathbb{R}^n$ be the embedding function that maps from the plan space to the feature representation vector space with $n$ dimensions.
For any candidate plan $P\in \mathcal{P}_q$, $h(P)$ represents the corresponding plan embedding.
Let $s$ be the \emph{cosine similarity function} and $0<\tau<1$ be some threshold.
Two plans $P_1, P_2\in\mathcal{P}_q$ are similar if 
$s(h(P_1), h(P_2))\geq\tau.$

\begin{algorithm}[t]
\small
  \SetAlgoLined
  \KwIn{$q$, the query; $\mathcal{P}_q$, the candidate plans of $q$ that have not been executed; $P$, a plan of $q$ with known execution time $t(P)$; $h$, the plan embedding function; $s$, the (cosine) similarity function; $\tau$, the plan similarity threshold.}
  \KwOut{$\tilde{\mathcal{P}}_q$, the remaining candidate plans of $q$.}
  \SetAlgoLined
  $\tilde{\mathcal{P}}_q\leftarrow\emptyset$\;
  \ForEach{$P'\in \mathcal{P}_q$}{
    \uIf{$s(h(P'), h(P))<\tau$}{
        $\tilde{\mathcal{P}}_q\leftarrow\tilde{\mathcal{P}}_q\cup\{P'\}$\;
    }
  }
  \Return{$\tilde{\mathcal{P}}_q$\;}
  \caption{\texttt{PrunePlansEmbedding}($q$, $\mathcal{P}_q$, $P$, $h$, $s$, $\tau$)}
\label{alg:optimization:embedding-sim}
\end{algorithm}

Algorithm~\ref{alg:optimization:embedding-sim} presents the details of excluding similar candidate plans from execution based on plan embeddings.
\rone{In our current implementation, we use the plan embedding techniques presented in ``AI meets AI (AMA)''~\cite{DingDM0CN19}. Compared to other plan embedding techniques, AMA does not require parametrized learning with labeled execution data, which fits the context of CMP tuning better as the candidate plans do not carry execution information after plan space exploration.
Specifically, AMA considers multiple ``feature channels'' for each type of physical operators, such as \emph{table scan}, \emph{index scan}, and \emph{hash join}, that appears in a query plan. Each feature channel represents a numeric metric, such as the estimated number of rows or the estimated cost of an operator, aggregated for all operators with the same type. For example, if a query plan contains three hash joins, then the feature channel ``estimated rows'' for the operator type ``hash join'' represents the \emph{sum} of the estimated rows outputted by all the three hash joins.
A ``channel vector'' is the vector that represents the values of the corresponding feature channel for all operator types.
The final plan vector/embedding is simply a \emph{concatenation} of all the channel vectors.
We note that the area of learned query/plan representations has been advanced significantly in recent years~\cite{ZhaoLC23}, with more detailed analysis regarding the invariance and uniqueness properties that should be satisfied by such representations~\cite{ReinerG23}.
These more advanced plan embedding techniques can be further incorporated into Algorithm~\ref{alg:optimization:embedding-sim} to enhance its performance, although, unlike AMA, training overhead becomes an issue that requires further consideration.
}


\subsubsection{Similarity by Bottleneck Matching}
\label{sec:optimization:similarity:bottleneck}

The embedding-based plan similarity only utilizes information of query plan structures without considering any query runtime information.
Two plans can be structurally similar but have very different execution time.
It is unclear how to extend the plan embeddings to also include execution time information, as this information is only available for plans that have been executed.

To better exploit the execution time information in the plans that have been executed, we propose a new similarity metric based on matching ``bottleneck operator,'' defined as follows.
\begin{definition}[Bottleneck Operator]
\label{def:bottleneck}
Let $P$ be an executed plan with execution time $t(P)$.
Let $O$ be an operator contained by $P$ with execution time $t(O)$.
$O$ is a bottleneck operator if $\frac{t(O)}{t(P)}\geq\alpha$, where $0<\alpha<1$ is some threshold.
\end{definition}
Once the execution of a plan $P$ is completed before hitting the timeout, we extract the top $M$ of its bottleneck operators by Definition~\ref{def:bottleneck}.
For each remaining plan $P'$ that has not yet been executed, we check whether $P'$ contains (matches of) all these top $M$ bottleneck operators from $P$.
If so, we will mark $P'$ as similar to $P$ and skip the execution of $P'$.
Algorithm~\ref{alg:optimization:bottleneck-sim} illustrates this process.

\rtwo{If a plan contains a new bottleneck operator that has not been observed before, then it cannot be skipped for execution. With our current framework, it has to be executed and we will potentially pay the execution cost of the new bottleneck operator.
However, its execution cost can be reduced by using the timeout mechanisms.}

\begin{algorithm}[t]
\small
  \SetAlgoLined
  \KwIn{$q$, the query; $\mathcal{P}_q$, the candidate plans of $q$ that have not been executed; $P$, a plan of $q$ with known execution time $t(P)$; $\alpha$, threshold for determining bottleneck operator; $M$, threshold for selecting the top bottleneck operators.}
  \KwOut{$\tilde{\mathcal{P}}_q$, the remaining candidate plans of $q$.}
  \SetAlgoLined
  $\mathcal{O}\leftarrow$\texttt{GetTopBottleneckOperators}($P$, $\alpha$, $M$)\;
  $\tilde{\mathcal{P}}_q\leftarrow\emptyset$\;
  \ForEach{$P'\in \mathcal{P}_q$}{
    $\mathcal{O}'\leftarrow$\texttt{MatchBottleneckOperators}($P'$, $\mathcal{O}$)\;
    \uIf{$|\mathcal{O}'|\neq |\mathcal{O}|$}{
        $\tilde{\mathcal{P}}_q\leftarrow\tilde{\mathcal{P}}_q\cup\{P'\}$\;
    }
  }
  \Return{$\tilde{\mathcal{P}}_q$\;}
  \caption{\texttt{PrunePlansBottleneck}($q$, $\mathcal{P}_q$, $P$, $\alpha$, $M$)}
\label{alg:optimization:bottleneck-sim}
\end{algorithm}



\vspace{-0.5em}
\subsection{Putting It Together}
\label{sec:optimization:combination}

Similarity-based pruning can be combined with timeout-based pruning. 
In the following, we present such an example by combining the progressive timeout mechanism (Algorithm~\ref{alg:optimization:progressive-timeout}) with the bottleneck-based similarity mechanism (Algorithm~\ref{alg:optimization:bottleneck-sim}).
It also serves as our \emph{main proposal} for reducing the evaluation time of candidate plans.

\begin{algorithm}
\small
  \SetAlgoLined
  \KwIn{$q$, the query; $\mathcal{P}_q$, the candidate plans of $q$; $P_q^{\default}$, the default plan of $q$ with known execution time;
  $\alpha$, threshold for determining bottleneck operator; $M$, threshold for selecting the top bottleneck operators.}
  \KwOut{$P_q^*$, the candidate plan with the minimum execution time.}
  \SetAlgoLined
  $P_q^*\leftarrow P_q^{\default}$, $\delta\leftarrow t(P_q^{\default})$\;
  $\mathcal{P}_q\leftarrow$\texttt{PrunePlansBottleneck}($q$, $\mathcal{P}_q$, $P_q^{\default}$, $\alpha$, $M$)\; \label{alg:call-bottleneck-match0}
  \While{$\mathcal{P}_q \neq \emptyset$}
  {
      $P\leftarrow$\texttt{GetAndRemoveFirstPlan}($\mathcal{P}_q$)\;
      Execute $P$ and monitor its execution time $t(P)$\;
      \uIf{$P$ finishes before hitting the timeout $\delta$}{
            $P_q^*\leftarrow P$, $\delta\leftarrow t(P)$\;
            $\mathcal{P}_q\leftarrow$\texttt{PrunePlansBottleneck}($q$, $\mathcal{P}_q$, $P$, $\alpha$, $M$)\; \label{alg:call-bottleneck-match1}
      }
      \Else{
        Terminate the execution of $P_q$ when $t(P_q)\geq\delta$\;
      }
  }
  \Return{$P_q^*$\;}
  \caption{\texttt{TimeoutAndBottleneck}($q$, $\mathcal{P}_q$, $P_q^{\default}$, $\alpha$, $M$)}
\label{alg:optimization:timeout-sim-combine}
\end{algorithm}

Algorithm~\ref{alg:optimization:timeout-sim-combine} presents the details, where we retain the skeleton of the progressive timeout mechanism in Algorithm~\ref{alg:optimization:progressive-timeout} and call Algorithm~\ref{alg:optimization:bottleneck-sim} to prune candidate plans 
(lines~\ref{alg:call-bottleneck-match0} and~\ref{alg:call-bottleneck-match1}).

Similar combinations of timeout-based and similarity-based optimization mechanisms can be composed by following the paradigm exemplified in Algorithm~\ref{alg:optimization:timeout-sim-combine}.
For example, we can combine the progressive (Algorithm~\ref{alg:optimization:progressive-timeout}) or the round-robin (Algorithm~\ref{alg:optimization:roundrobin-timeout}) timeout mechanism with the embedding-based similarity mechanism (Algorithm~\ref{alg:optimization:embedding-sim}) by filtering the remaining candidate plans that are similar to the ones executed.
However, the combination of the round-robin timeout mechanism and the similarity based on bottleneck operators is perhaps not effective in practice, as most query execution engines currently do not return operator-level execution information if a query plan cannot be completed before hitting the timeout.
This is different from the application scenario of the round-robin timeout mechanism \rone{when it was proposed by Giannakouris and Trummer~\cite{GiannakourisT25}}, where it was used to evaluate (index) configurations suggested by LLMs and all queries that can be completed in the current round do not need to be run again in the next round.
The round-robin timeout mechanism would have worked better in the context of CMP tuning, if a ``pause and resume'' query (plan) execution mechanism could be supported to avoid restarting query execution from scratch.



\vspace{-0.5em}
\section{Experimental Evaluation}
\label{sec:eval}

We conduct extensive experimental evaluations of the proposed CMP tuning algorithms using publicly available benchmark workloads and report the evaluation results in this section.

\begin{table}[t] 
\caption{Summary of database and workload statistics.}
\vspace{-1em}
\small
\centering
\begin{tabular}{l|r|r|r|r|r}
\hline
\textbf{Name} & \textbf{DB Size} & \textbf{\#Queries} & \textbf{\#Tables} & \textbf{\#Joins} 
& \textbf{\#Scans}\\
\hline
\textbf{TPC-H} & \emph{sf}=1, 10 & 22 & 8 & 2.8 
& 3.7 \\
\textbf{TPC-DS} & \emph{sf}=1, 10 & 99 & 24 & 7.7 
& 8.8 \\
\textbf{JOB} & 9.2GB & 113 & 21 & 7.7 & 8.7\\ 
\textbf{STATS} & 223MB & 146 & 8 & 3.3 & 4.3\\
\hline
\end{tabular}
\label{tab:databases}
\end{table}

\vspace{-0.5em}
\subsection{Experimental Settings}

We run our experiments with PostgreSQL 17.4 and Ubuntu 22.04 LTS running on a workstation configured with 40 Intel CPUs at 2.6 GHz and 384 GB memory.
We clean the PostgreSQL buffer pool and the file system cache before executing each query plan.
\rone{That is, the queries were run from a ``cold start'' for better reproducibility.
To ensure this, we restart the PostgreSQL server and clean the Linux OS page cache by executing \textsf{/proc/sys/vm/drop\_caches}.}
\rtwo{The workstation is also equipped with a 24 TB HDD and a 400 GB SSD manufactured by HP.}
\shepherd{The databases are located on the HDD array, which consists of HPE EG1200FDJYT 1.2 TB 10.5K RPM SAS drives managed by an HPE Smart Array Gen9 RAID controller.}

\vspace{-0.5em}
\paragraph*{Databases and Workloads}

We use four publicly available benchmark workloads in our experiments: (1) \textbf{TPC-H}, (2) \textbf{TPC-DS}, (3) the ``Join Order Benchmark'' (\textbf{JOB})~\cite{LeisGMBK015}, and (4) the ``Cardinality Estimation Benchmark'' (\textbf{STATS})~\cite{HanWWZYTZCQPQZL21}.
Table~\ref{tab:databases} summarizes some basic properties of the workloads, in terms of schema complexity (e.g., the number of tables), query complexity (e.g., the average number of joins and table scans contained by a query), and database/workload size.
We also use different scaling factors (sf), namely sf = 1 and sf = 10, for \textbf{TPC-H} and \textbf{TPC-DS} to test the efficacy of our proposed \rthree{techniques} for CMP tuning at different database scales.

\vspace{-0.5em}
\paragraph*{Cost Model Parameters}
We focus on tuning the CMPs for PostgreSQL that are listed in Table~\ref{tab:pgsql-cost-units}.
From the viewpoint of the query optimizer, only the \emph{relative} costs of the plans matter because the query optimizer only needs to \emph{compare} plan costs when selecting the best plan.
Therefore, we omit tuning the ``\emph{seq\_page\_cost}'' CMP by fixing its value to 1.
This simplifies the tuning problem from a five-dimensional plan space to a four-dimensional plan space, which is meaningful in practice as it saves significant time for plan exploration.
Each remaining CMP $c$ is associated with a range $[0.1\times c_d, 10\times c_d]$ where $c_d$ is the default value of $c$ (as shown in Table~\ref{tab:pgsql-cost-units}).
The same setup has been used by previous work~\cite{WuW24}.

\vspace{-0.5em}
\paragraph*{Baselines}
We compare our deterministic and optimal tuning (\textbf{DOT}) approach with the following baselines: (1) \textbf{Default}, which simply uses the default query plan without tuning the CMPs; (2) \textbf{RS}, which uses \emph{random search}~\cite{RandomSearch} to explore the plan space; and (3) \textbf{BO}, which uses \emph{SMAC}~\cite{SMAC,SMAC3}, a state-of-the-art \emph{Bayesian optimization} approach.
Both \textbf{RS} and \textbf{BO} were used in previous work~\cite{WuW24} with proved efficacy of tuning CMPs.
We run \textbf{RS} and \textbf{BO} 100 times for each query and collect all unique plans found.

\vspace{-0.5em}
\paragraph*{Performance Metrics}
We use the following metrics to evaluate the performance of \textbf{DOT}: (1) the (percentage) improvement of the best plan found compared to the default plan, in terms of execution time, and (2) the overall time spent on tuning the CMPs.


\begin{table}[t]
\caption{Aggregated workload execution time (in minutes) with the best plan found for each query. 
\rone{We excluded queries with default plans that cannot finish in 20 minutes. }}
\vspace{-1em}
\small
    \centering
    \begin{tabular}{l|rrrrr}
    \toprule
        {\bf Workload} & {\bf \#Queries} & {\bf Default} & {\bf RS} & {\bf BO} & {\bf DOT}\\
    \midrule
         {\bf TPC-H} (sf = 1) & 22 & 7.28 & 4.82	& 4.57 & 4.57 \\
         {\bf TPC-H} (sf = 10) & 19 & 17.22 & 14.29 & 14.42 & 9.92 \\
         {\bf TPC-DS} (sf = 1) & 95 & 7.42 & 3.96 & 4.23 & 3.80 \\
         {\bf TPC-DS} (sf = 10) & 95 & 171.68 & 81.30 & 75.96 & 72.79 \\
         {\bf JOB} & 113 & 138.75 & 73.81 & 74.00 & 73.08 \\
         {\bf STATS} & 143 & 13.99 & 10.51 & 10.72 & 10.13 \\
    \bottomrule
    \end{tabular}
    \label{tab:improve:workload}
\end{table}

\vspace{-0.5em}
\subsection{Improvement over Baselines}
\label{sec:eval:exec-time}

Table~\ref{tab:improve:workload} summarizes the aggregated execution time of the workload using the best plan found for each query.
\rone{Here we have excluded queries with default plans taking more than 20 minutes to complete. 
For example, three \textbf{TPC-H} (sf = 10) queries 8, 9, and 20 finish with time longer than 20 minutes, and they cannot be further improved after CMP tuning.}
\shepherd{For these queries, PostgreSQL selected plans dominated by index-based nested-loop joins. Under our cold-cache configuration, the resulting random disk accesses became the primary execution bottleneck.}

We observe that \textbf{RS}, \textbf{BO}, and \textbf{DOT} all significantly outperform \textbf{Default}.
Moreover, \textbf{DOT} further significantly outperforms \textbf{RS} and \textbf{BO} for several workloads, such as \textbf{TPC-H} and \textbf{TPC-DS} with scaling factor 10.
For example, \textbf{DOT} improves \textbf{RS} and \textbf{BO} by 30\% for \textbf{TPC-H} (sf = 10).

The improvement of \textbf{DOT} over \textbf{RS} and \textbf{BO} for certain workloads, such as \textbf{JOB} and $\textbf{STATS}$, may not seem impressive.
However, this is due to a bias introduced by summing up query execution time at workload level: long-running queries dominate the aggregated workload execution time.
To remove this bias, Figure~\ref{fig:improve:query} presents the distribution of the (percentage) improvement at the query-level in terms of the individual query execution time for each workload.
We observe a visible increase for both \textbf{JOB} and \textbf{STATS} by comparing the \emph{median} and other percentiles of the query-level improvement distribution of \textbf{DOT} with those of \textbf{RS} and \textbf{BO}, as shown by Figures~\ref{fig:improve:query:job} and~\ref{fig:improve:query:stats}.
\rthree{Figure~\ref{fig:improve:singlequery:tpch10} further presents the improvement observed in each individual query of \textbf{TPC-H} (sf = 10). 
We observe significantly increased query-level improvements by using \textbf{DOT} compared to using \textbf{RS} or \textbf{BO}.
For example, 
the (percentage) improvement given by \textbf{DOT} over \textbf{Default} is 81\% for query 11, whereas both \textbf{RS} and \textbf{BO} show only 21\% improvement over \textbf{Default}.
We also observe similar advantages/benefits of \textbf{DOT} for queries from the other workloads (see Appendix~\ref{appendix:results:query-level-impr}).}


\begin{figure}[t]
\centering
\subfigure[\textbf{TPC-H}, sf = 1]{\label{fig:improve:query:tpch1}
    \includegraphics[width=0.31\columnwidth]{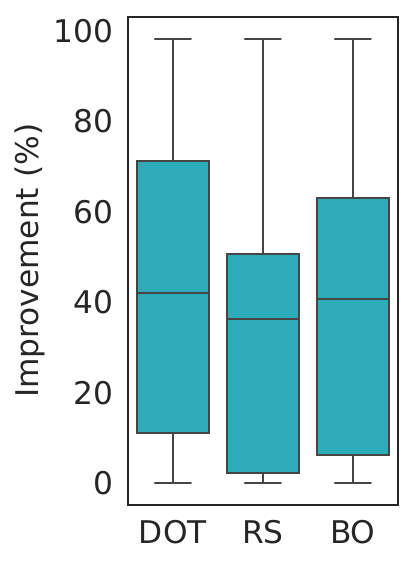}}
\subfigure[\textbf{TPC-H}, sf = 10]{\label{fig:improve:query:tpch10}
    \includegraphics[width=0.31\columnwidth]{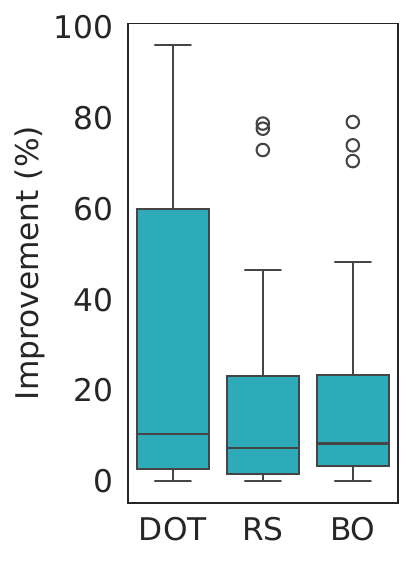}}
\subfigure[\textbf{TPC-DS}, sf = 1]{\label{fig:improve:query:tpcds1}
    \includegraphics[width=0.31\columnwidth]{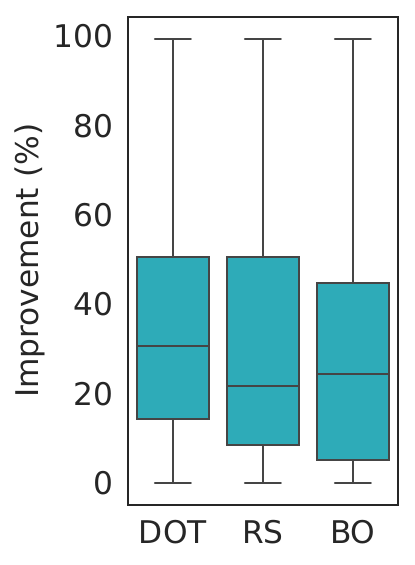}}
\subfigure[\textbf{TPC-DS}, sf = 10]{\label{fig:improve:query:tpcds10}
    \includegraphics[width=0.31\columnwidth]{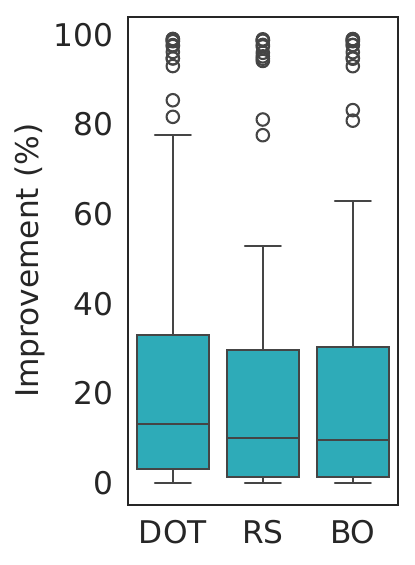}}
\subfigure[\textbf{JOB}]{\label{fig:improve:query:job}
    \includegraphics[width=0.31\columnwidth]{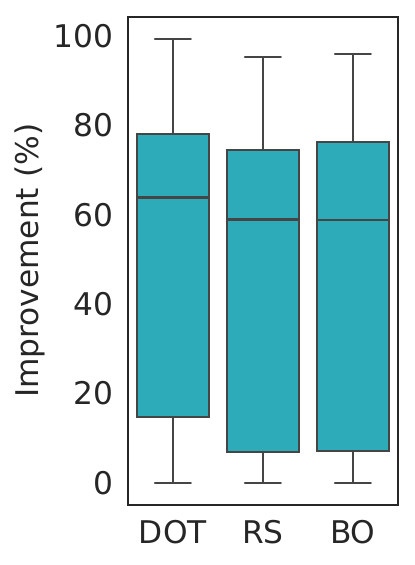}}
\subfigure[\textbf{STATS}]{\label{fig:improve:query:stats}
    \includegraphics[width=0.31\columnwidth]{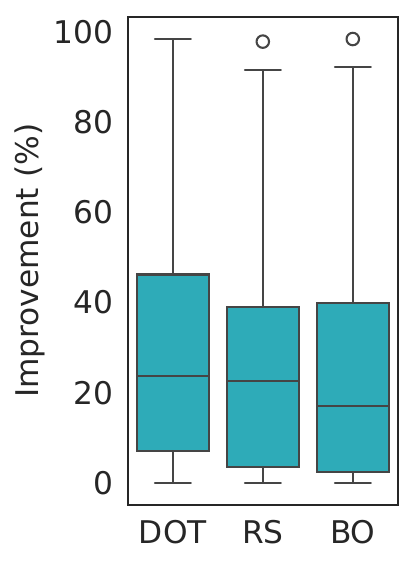}}
\vspace{-1.5em}
\caption{Distribution of (percentage) improvement of query-level execution time by \textbf{RS}, \textbf{BO}, and \textbf{DOT} over \textbf{Default}. Each box represents the \emph{interquartile range} (IQR) that spans from the 25th percentile (Q1) to the 75th percentile (Q3). The horizontal line inside each box represents the \emph{median}. Whiskers extend to the most extreme data points that are not outliers. Specifically, the lower whisker represents the smallest value $\geq$ Q1 $-$ 1.5 $\times$ IQR, whereas the upper whisker represent the largest value $\leq$ Q3 $+$ 1.5 $\times$ IQR. Points plotted individually (as dots) outside the whiskers are \emph{outliers}.}
\label{fig:improve:query}
\vspace{0.5em}
\end{figure}





\begin{figure}
\centering
\includegraphics[width=1.0\columnwidth]{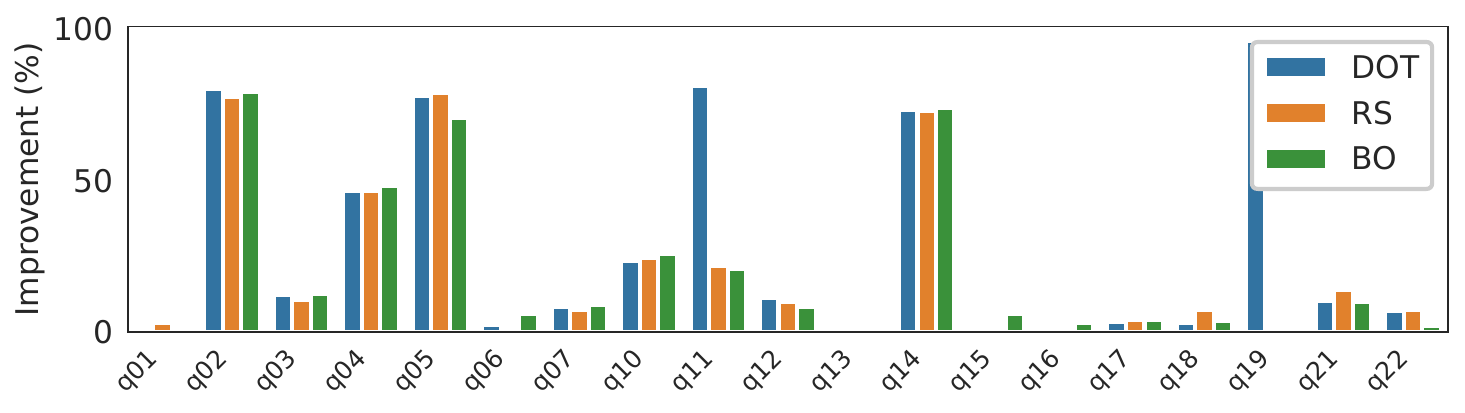}
\vspace{-1.5em}
\caption{\rthree{Improvements over \textbf{Default} by \textbf{RS}, \textbf{BO}, and \textbf{DOT} for individual queries from TPC-H (sf = 10).}}
\label{fig:improve:singlequery:tpch10}
\end{figure}

\subsubsection{\rthree{Query Groups with Different Number of Joins}}

\begin{figure*}
\centering
\subfigure[\textbf{TPC-H}, sf = 1]{ \label{fig:join-group:tpch1}
    \includegraphics[width=0.3\textwidth]{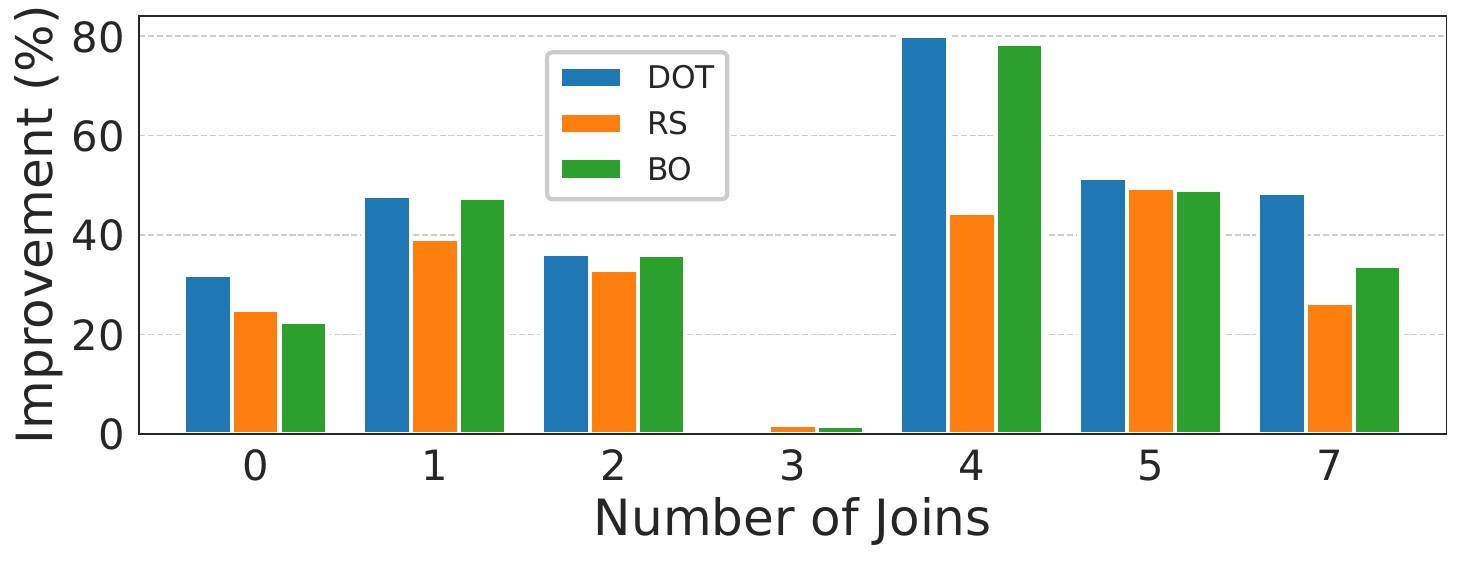}}
\subfigure[\textbf{TPC-H}, sf = 10]{ \label{fig:join-group:tpch10}
    \includegraphics[width=0.3\textwidth]{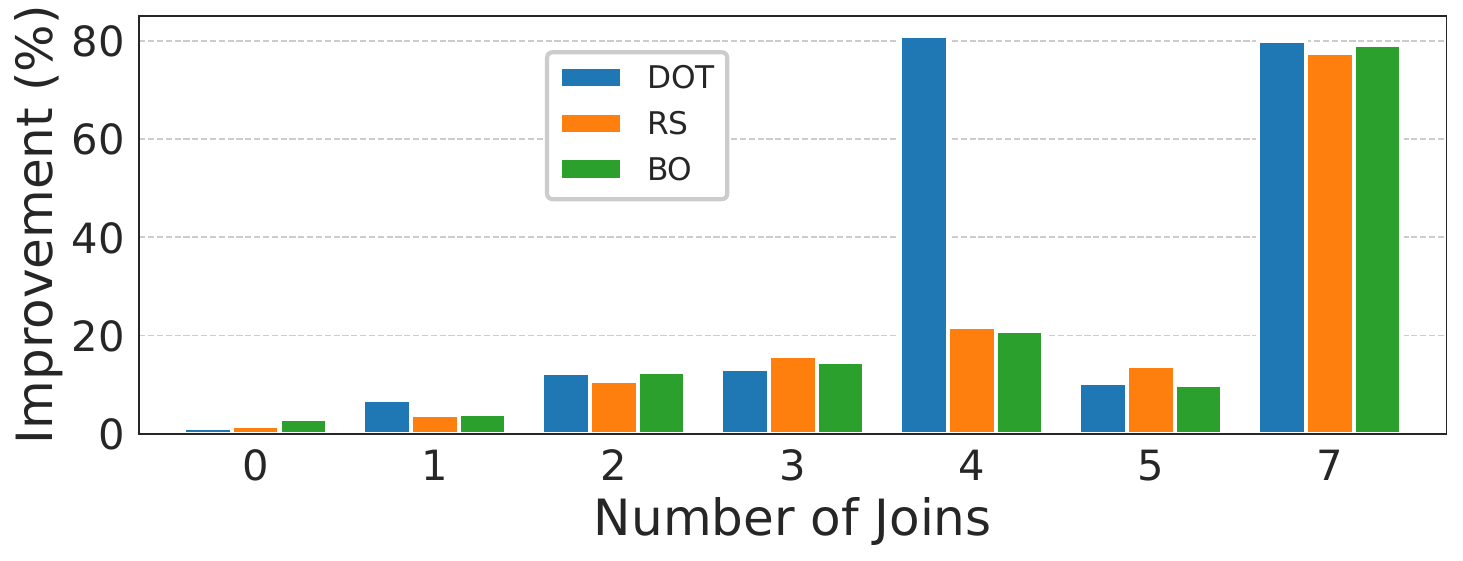}}
\subfigure[\textbf{TPC-DS}, sf = 1]{ \label{fig:join-group:tpcds1}
    \includegraphics[width=0.3\textwidth]{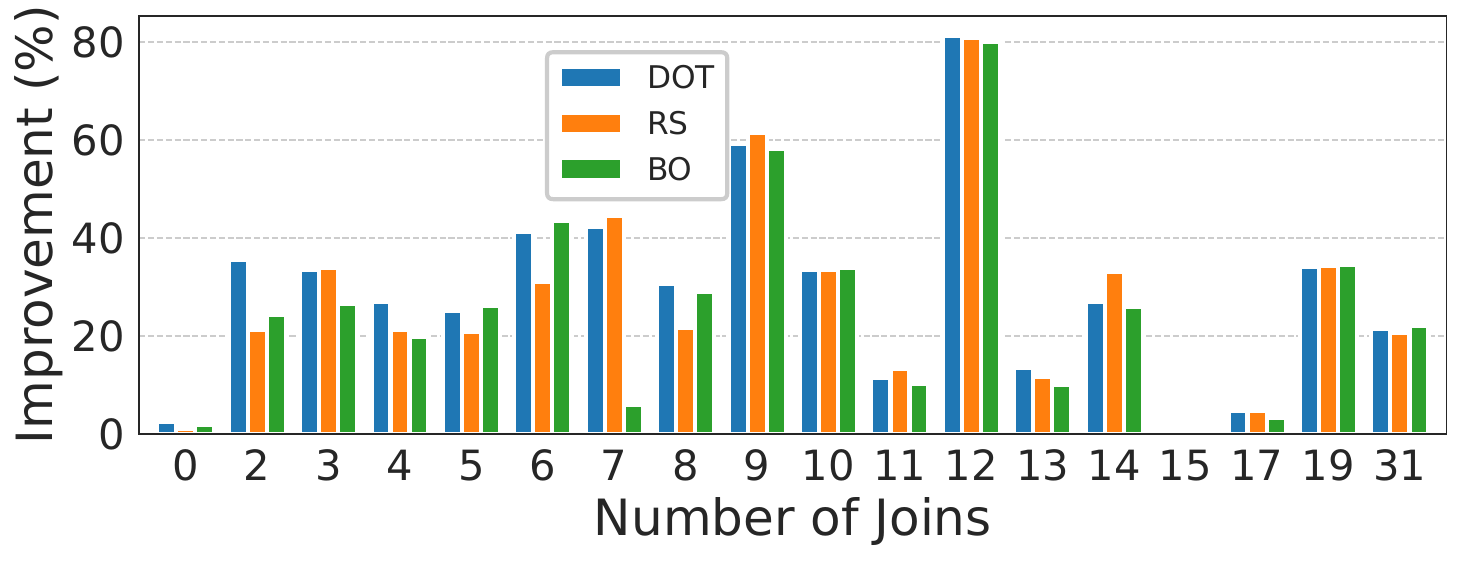}}
\subfigure[\textbf{TPC-DS}, sf = 10]{ \label{fig:join-group:tpcds10}
    \includegraphics[width=0.3\textwidth]{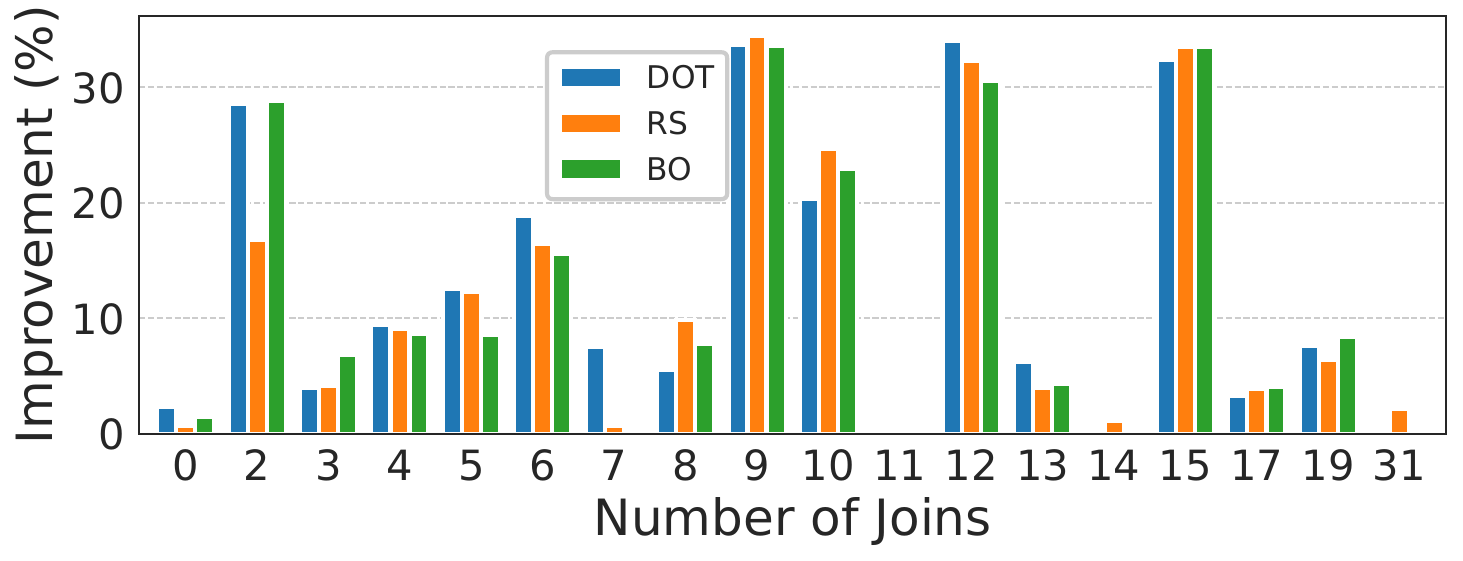}}
\subfigure[\textbf{JOB}]{ \label{fig:join-group:job}
    \includegraphics[width=0.3\textwidth]{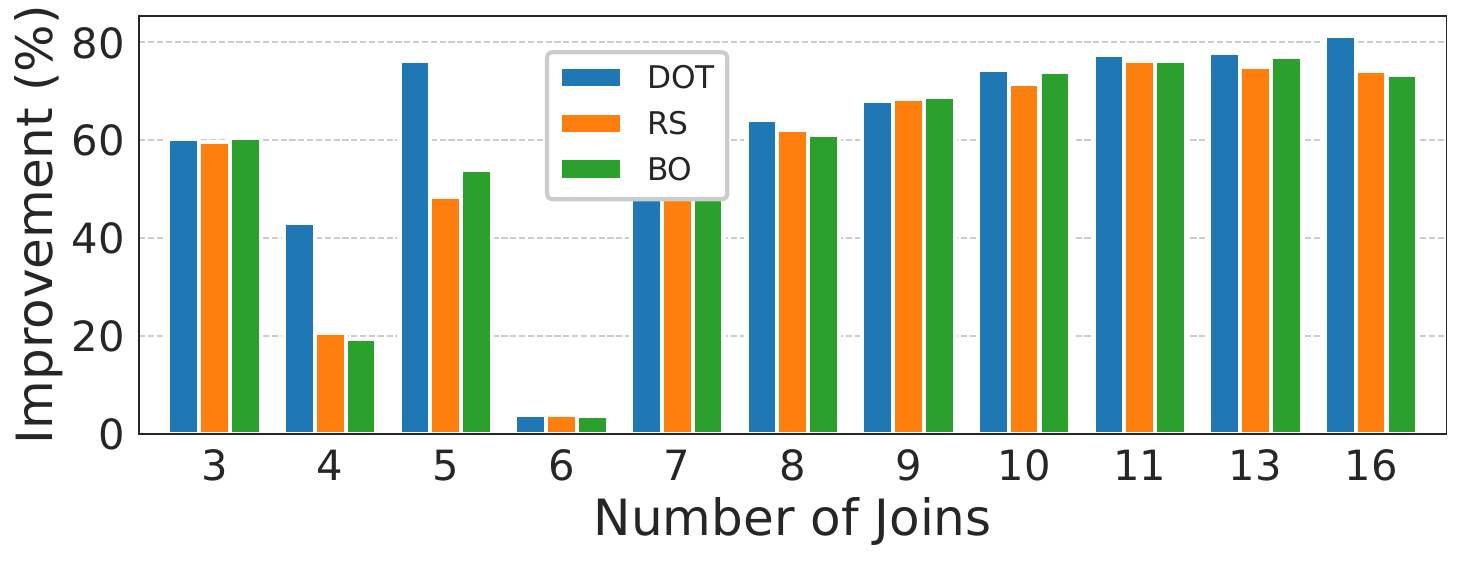}}
\subfigure[\textbf{STATS}]{ \label{fig:join-group:stats}
    \includegraphics[width=0.3\textwidth]{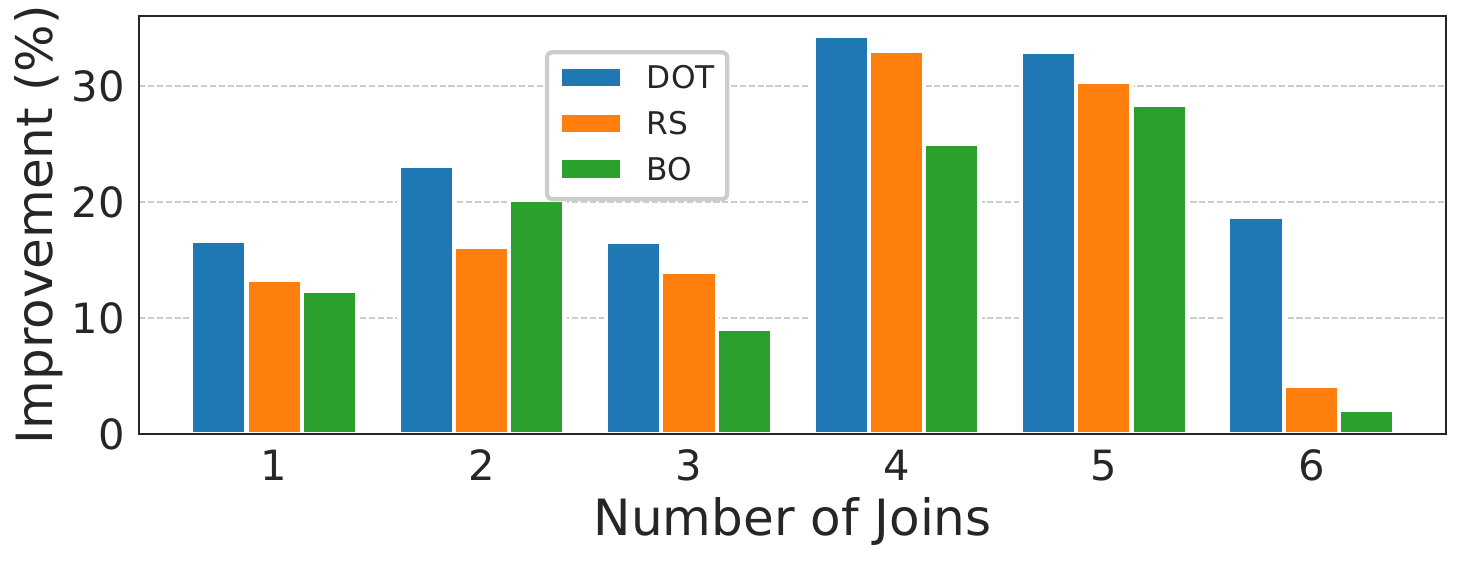}}
\vspace{-1.5em}
\caption{\rthree{Median improvements over \textbf{Default} by \textbf{RS}, \textbf{BO}, and \textbf{DOT} for individual queries grouped by the number of joins.}}
\label{fig:join-group}
\vspace{-1em}
\end{figure*}

\rthree{
We further analyze the benefits of \textbf{DOT} over the other baselines for queries with different complexity.
Specifically, we group queries by their number of joins and report the median improvement in each query group.
Figure~\ref{fig:join-group} presents the results.
We observe some interesting patterns that are workload dependent.
For example, the main benefit of the proposed approach for \textbf{TPC-H} (sf = 10) comes from queries with four joins.
The benefits for the other workloads are more evenly distributed, ranging from queries with only two joins (e.g., \textbf{TPC-DS}) to queries with six joins (e.g., \textbf{JOB}).
However, it is not true that queries with more joins can benefit more from \textbf{DOT}, even though the search spaces of candidate plans become larger. Since tuning CMPs comes with nontrivial overheads (even with the optimization techniques presented in Section~\ref{sec:optimization-tuning-time}), it is an interesting direction to develop prediction techniques to forecast whether a query is worth CMP tuning, for the sake of avoiding unnecessary tuning. 
We leave this for future work.
}

\subsubsection{\rtwo{Hard Cases}}

\begin{figure}
\centering
\subfigure[\textbf{TPC-H} (sf = 10), Query 18]{ \label{fig:analysis:hard-cases:tpch10}
    \includegraphics[width=0.48\columnwidth]{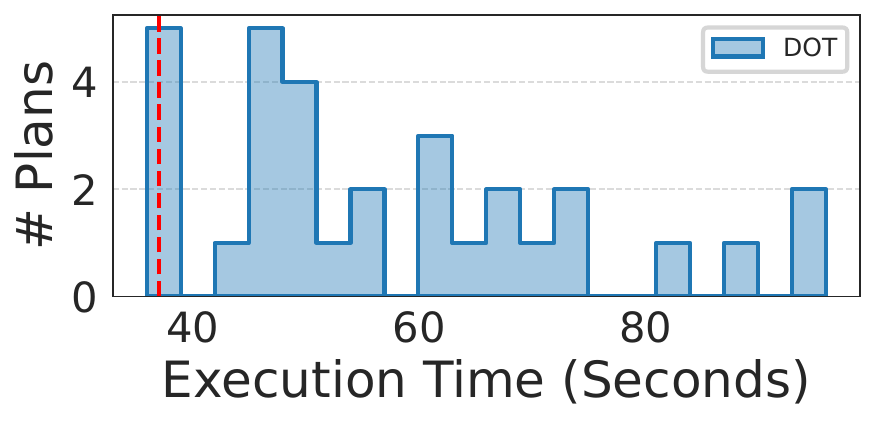}}
\subfigure[\textbf{JOB}, Query 23c]{ \label{fig:analysis:hard-cases:stat}
    \includegraphics[width=0.48\columnwidth]{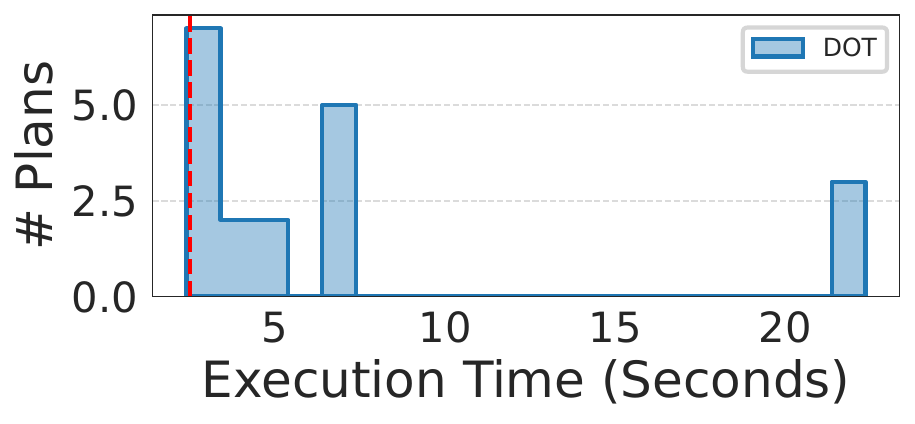}}
\vspace{-1.5em}
\caption{\rtwo{Execution time distribution of the candidate plans found by DOT for example ``hard'' queries. The red dashed line represents the execution time of the \textbf{Default} plan.}}
\label{fig:analysis:hard-cases}
\end{figure}

\rtwo{Although many queries can be significantly improved by \textbf{DOT}, there are queries that cannot be improved. We now take a closer look at such ``hard cases.'' In general, the hard cases fall into two classes: (1) lack of query plan diversity, and (2) the default plan already being the best.
Examples of the first class include the \textbf{TPC-H} query 1 and the \textbf{TPC-DS} query 41, which requires only scans over a single table and some aggregates on top of the scans.
It is clear that \textbf{DOT} (and other CMP tuning techniques) cannot help in such situations.
The second class typically includes queries where the query optimizer already does a good job.
For example, the degree of cardinality estimation (CE) errors is small, or CE errors have similar impacts on all candidate plans in the search space so that the ranking of the plans is not affected.
Figure~\ref{fig:analysis:hard-cases} presents two such examples: although \textbf{DOT} find many candidate plans, none of them can outperform the default plan.}

\vspace{-0.5em}
\subsection{Plan Space Exploration}
\label{sec:eval:exploration}

To understand the efficacy of the BSP-based plan space exploration strategy employed by \textbf{DOT}, we present an in-depth study of the query plans obtained.
We use two metrics for this study: (1) the number of distinct plans found in the exploration of the plan space, and (2) the number of query optimizer calls made for the exploration, which reflects the time spent on exploring the plan space.


\begin{figure*}
\centering
\subfigure[\textbf{TPC-H}, sf = 10]{ \label{fig:space:plans:tpch10}
    \includegraphics[width=0.28\textwidth]{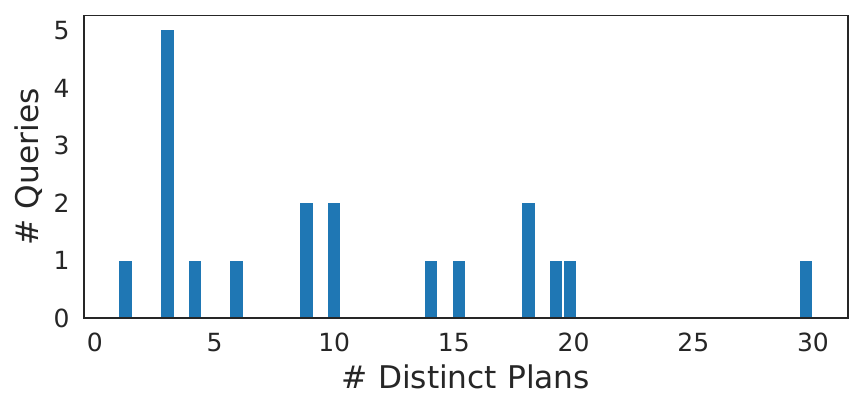}}
\subfigure[\textbf{TPC-DS}, sf = 10]{ \label{fig:space:plans:tpcds10}
    \includegraphics[width=0.28\textwidth]{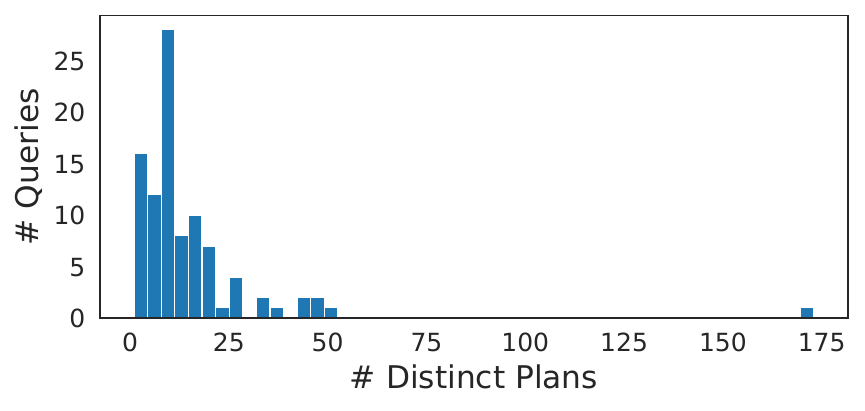}}
\subfigure[\textbf{JOB}]{ \label{fig:space:plans:job}
    \includegraphics[width=0.28\textwidth]{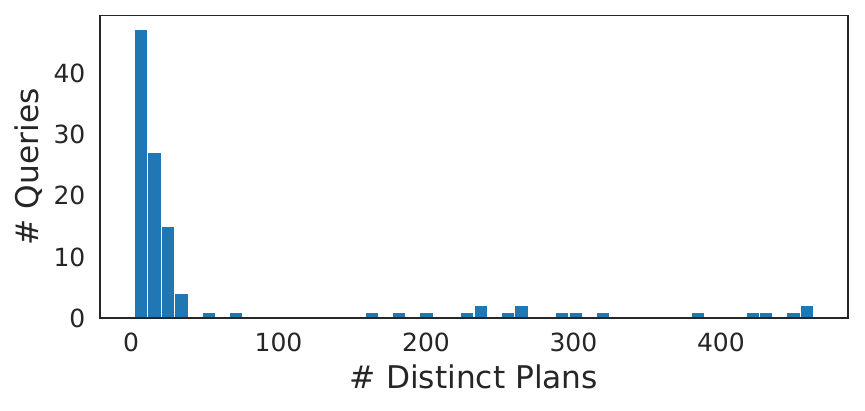}}
\vspace{-1.5em}
\caption{\rthree{Distribution of the number of distinct query plans found during plan space exploration by \textbf{DOT}.}}
\label{fig:space:plans}
\vspace{-1em}
\end{figure*}

\vspace{-0.5em}
\subsubsection{\rthree{Plans and Configurations Explored}}
\label{sec:eval:exploration:plan-config}

Figure~\ref{fig:space:plans} presents the distribution of queries based on the number of distinct plans found by \textbf{DOT} during plan space exploration.
Interestingly, we observe very different distributions for different workloads.
\rthree{For example, the distribution observed for \textbf{TPC-H} (sf = 10) is almost uniform, with the number of distinct plans ranging from 1 to 30.
Meanwhile, the distributions observed for \textbf{TPC-DS} and \textbf{JOB} are skewed and
have long tails --- most queries have only a dozen of plans, whereas a few outliers have hundreds of plans.}
This highlights the opportunities missed by \textbf{RS} and \textbf{BO} with 100 iterations, which can identify at most 100 plans.
It is possible that \textbf{RS} and \textbf{BO} can perform better with more iterations; \emph{however, it is difficult to tell how many iterations are required by \textbf{RS} or \textbf{BO} for a given query.}
\textbf{DOT} addresses this dilemma faced by \textbf{RS} and \textbf{BO} \emph{in an automated adaptive manner} by exploiting the linearity properties of the plan cost functions used by the query optimizer.
\rthree{The results from the other workloads are included in Appendix~\ref{appendix:results:analysis-search-space}.}

\begin{figure*}
\centering
\subfigure[\textbf{TPC-H}, sf = 10]{ \label{fig:space:configs:tpch10}
    \includegraphics[width=0.28\textwidth]{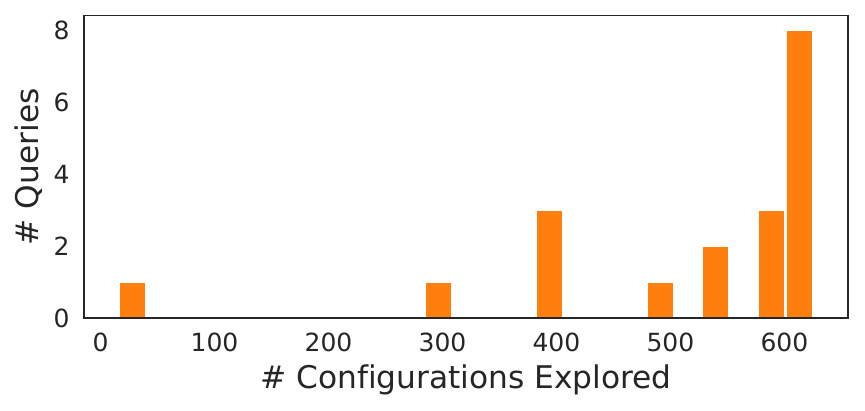}}
\subfigure[\textbf{TPC-DS}, sf = 10]{ \label{fig:space:configs:tpcds10}
    \includegraphics[width=0.28\textwidth]{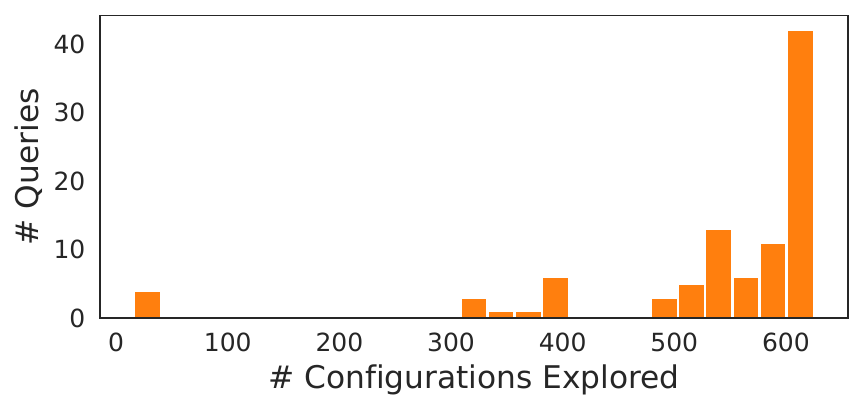}}
\subfigure[\textbf{JOB}]{ \label{fig:space:configs:job}
    \includegraphics[width=0.28\textwidth]{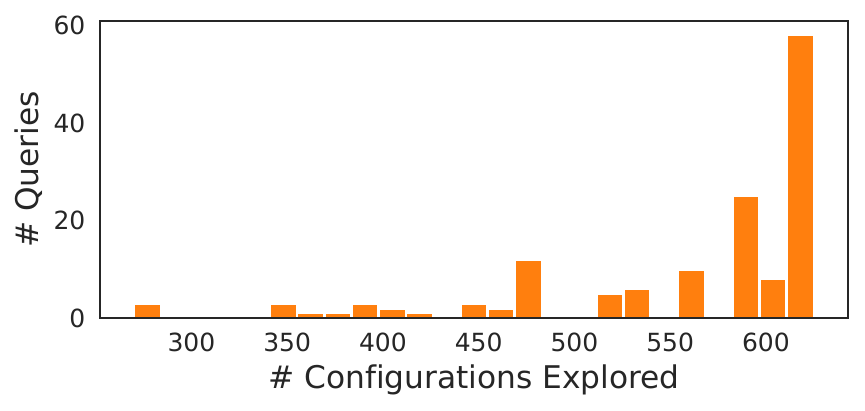}}
\vspace{-1.5em}
\caption{\rthree{Distribution of the number of CMP configurations explored (i.e., the number of query optimizer calls) by \textbf{DOT}.}}
\label{fig:space:configs}
\vspace{-1em}
\end{figure*}

\begin{figure*}
\centering
\subfigure[\textbf{TPC-H}, sf = 10]{ \label{fig:plan-config-per-query:tpch10}
    \includegraphics[width=0.28\textwidth]{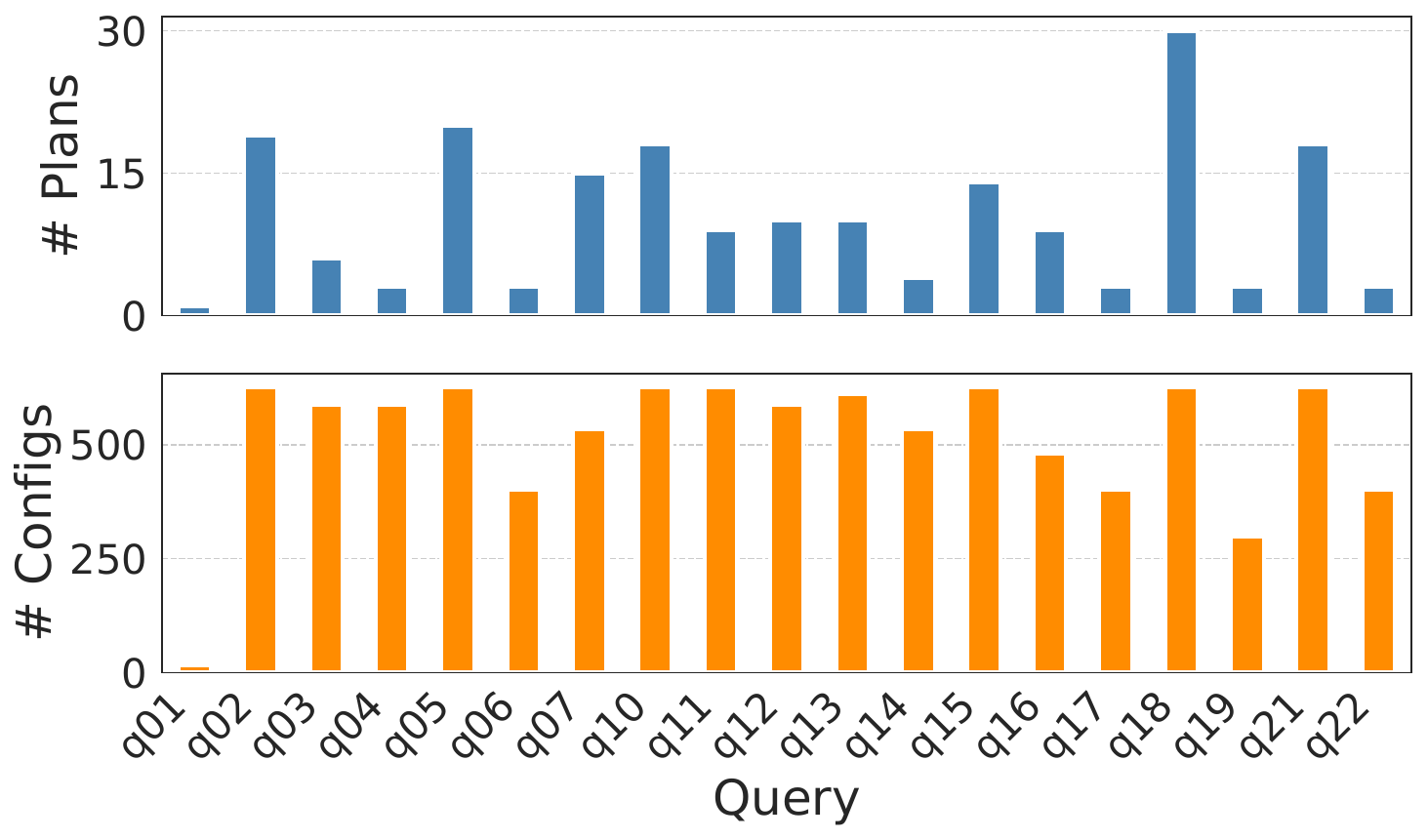}}
\subfigure[\textbf{TPC-DS}, sf = 10, Part I]{ \label{fig:plan-config-per-query:tpcds10:1}
    \includegraphics[width=0.28\textwidth]{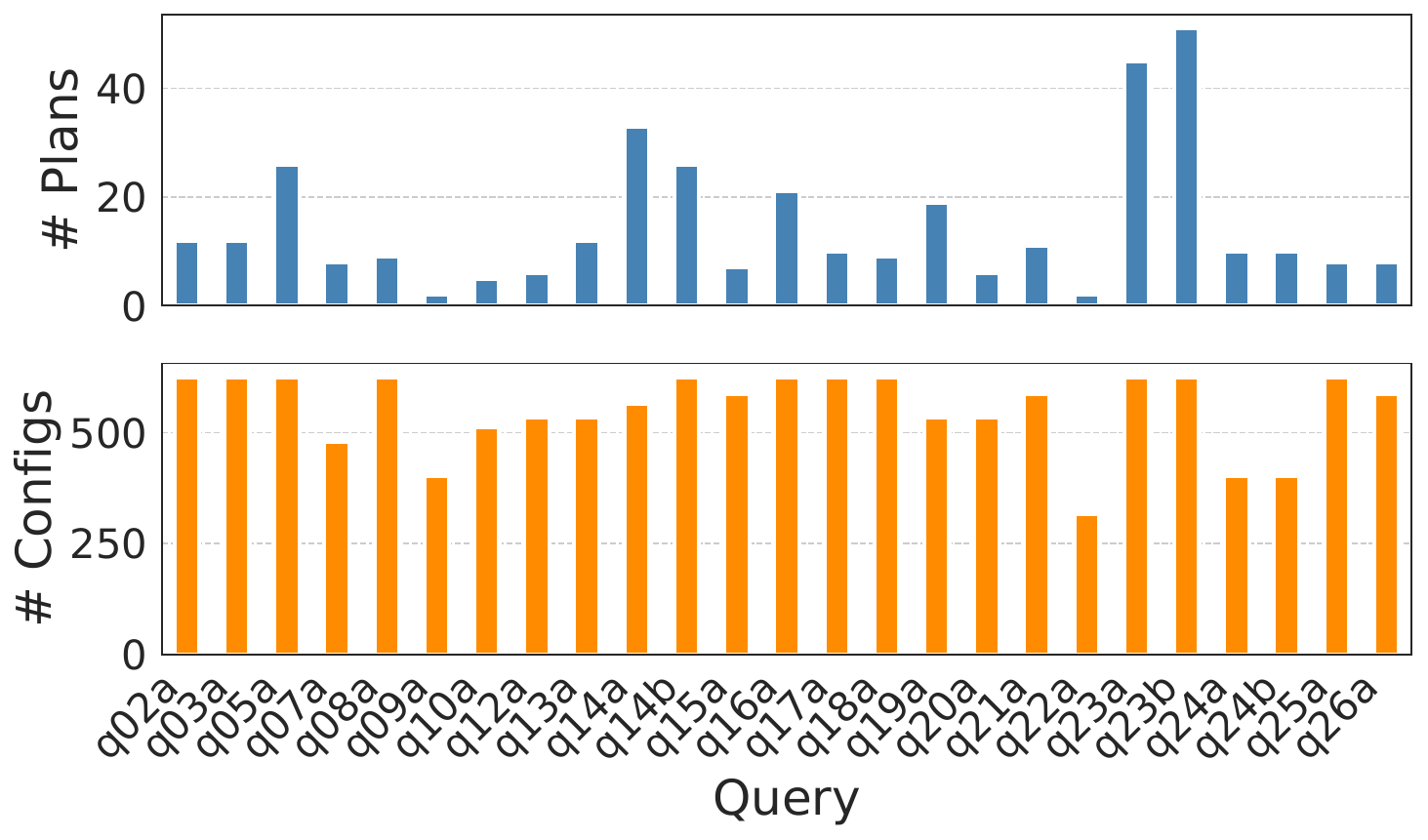}}
\subfigure[\textbf{JOB}, Part I]{ \label{fig:plan-config-per-query:job:1}
    \includegraphics[width=0.28\textwidth]{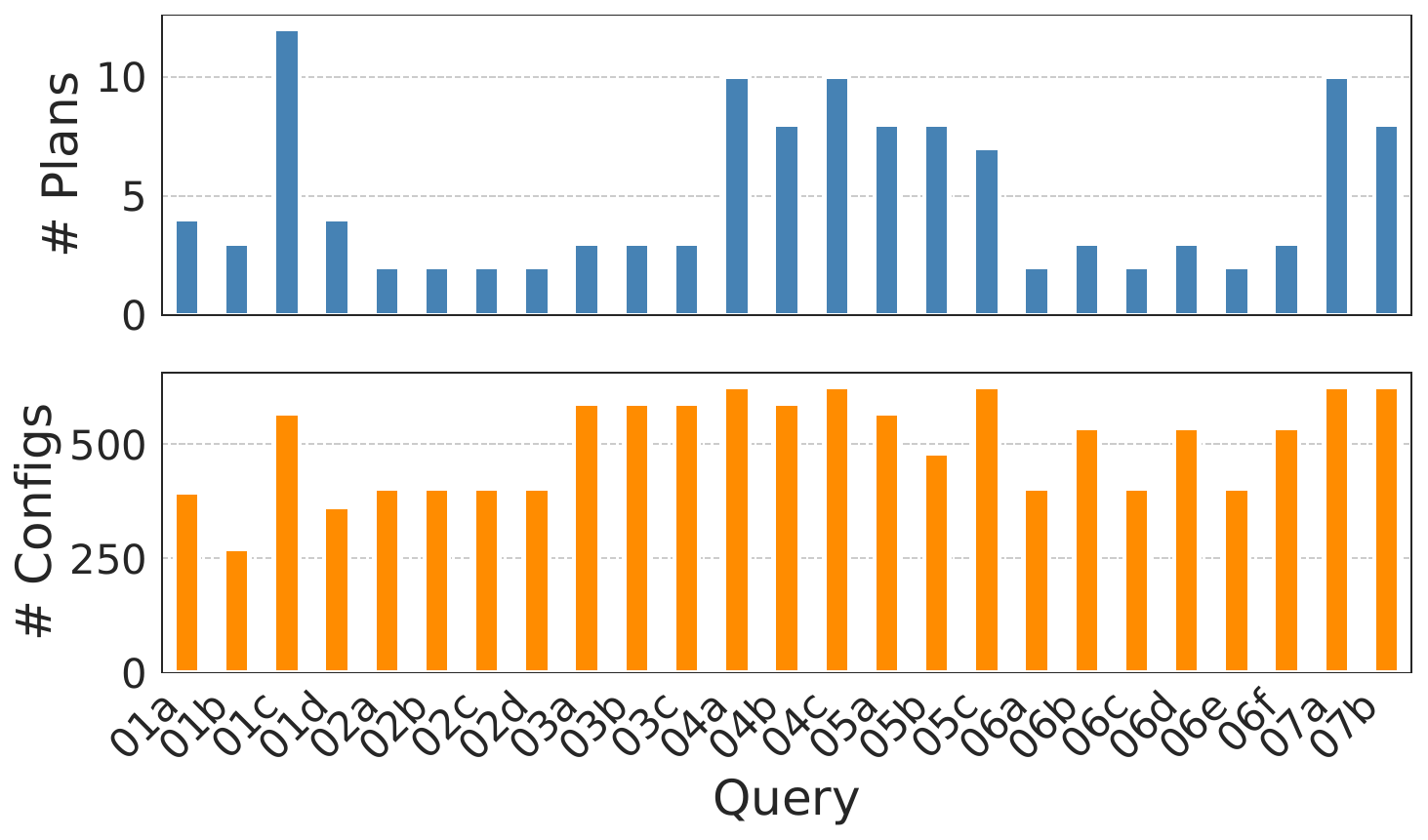}}
\vspace{-1.5em}
\caption{\rthree{The number of distinct query plans found and the number of CMP configurations explored (i.e., the number of query optimizer calls) during plan space exploration by \textbf{DOT} for individual queries.}}
\label{fig:plan-config-per-query}
\vspace{-1em}
\end{figure*}

Figure~\ref{fig:space:configs} presents the distribution for the number of CMP configurations explored by \textbf{DOT}, which is the same as the number of query optimizer calls made.
We observe that the majority of the queries have a rich space for exploration with hundreds of query optimizer calls.
Comparison of Figure~\ref{fig:space:configs} with Figure~\ref{fig:space:plans} reiterates the point we made in Section~\ref{sec:explore:multi-dim:optimality}: the exploration strategy employed by \textbf{DOT} is not optimal in terms of the number of query optimizer calls, which can result in more query optimizer calls than the number of distinct query plans.
We leave further optimization of the plan space exploration strategy for future work.

\rthree{Figure~\ref{fig:plan-config-per-query} further presents the number of plans and configurations explored for individual queries. Due to space constraints, only the first few queries of \textbf{TPC-DS} and \textbf{JOB} are shown and the rest are in Appendix~\ref{appendix:results:analysis-search-space}.
We can observe that, although \textbf{DOT} explored hundreds of configurations for most queries, the number of distinct plans found is much less and differs dramatically across queries.}

\subsubsection{\rtwo{Diversity of CMPs}}
\label{sec:eval:exploration:cmp-diversity}

\begin{figure*}
\centering
\subfigure[\textbf{TPC-H}, sf = 10]{ \label{fig:cmp-diversity:tpch10}
    \includegraphics[width=0.28\textwidth]{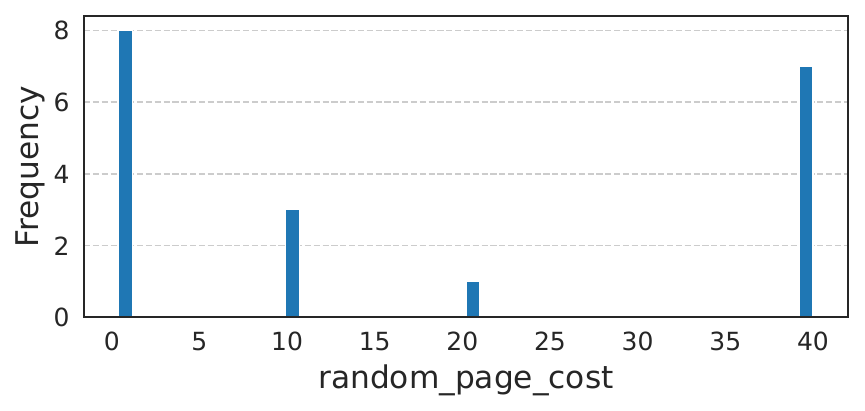}}
\subfigure[\textbf{TPC-DS}, sf = 10]{ \label{fig:cmp-diversity:tpcds10}
    \includegraphics[width=0.28\textwidth]{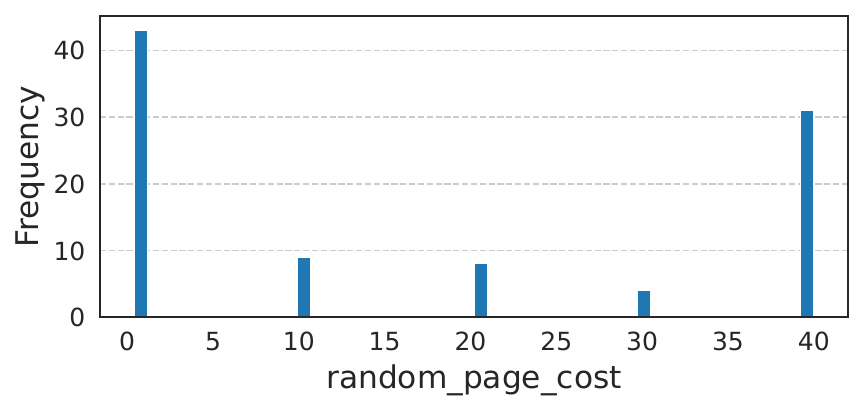}}
\subfigure[\textbf{JOB}]{ \label{fig:cmp-diversity:job}
    \includegraphics[width=0.28\textwidth]{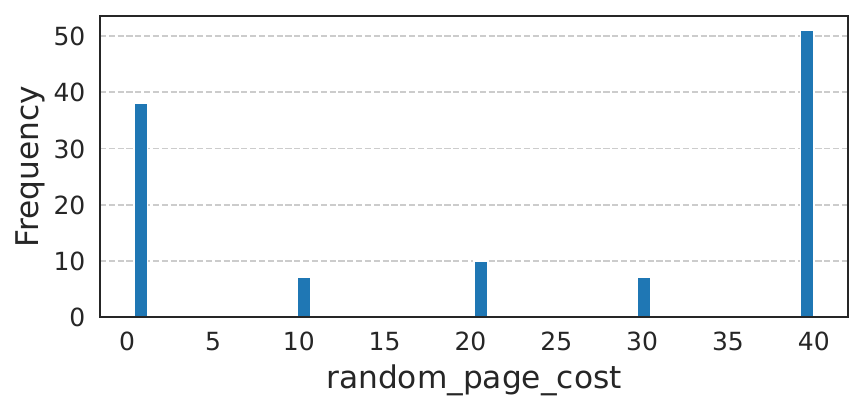}}
\vspace{-1.5em}
\caption{\rtwo{Distribution of values of the CMP \texttt{random\_page\_cost} used by the best query plans found by \textbf{DOT}.}}
\label{fig:cmp-diversity}
\vspace{-1em}
\end{figure*}

\rtwo{
It is an interesting question regarding the \emph{diversity} of the CMPs used by the best plans found by \textbf{DOT}.
To understand this better, Figure~\ref{fig:cmp-diversity} presents the distributions of the values of the CMP \texttt{random\_page\_cost} for different workloads.
This CMP has arguably the highest impact on the estimated plan cost under the ``cold start'' setup.
We observe that the distribution of \texttt{random\_page\_cost} is quite heterogeneous.
That is, the best plans found for different queries use very different values for this CMP.
This implies that a simple one-time calibration of the CMPs is not sufficient to find these best query plans.
The results from the other workloads are included in Appendix~\ref{appendix:results:analysis-cmpt-diversity}.}

\vspace{-0.5em}
\subsubsection{Case Study}
\label{sec:eval:exploration:case-study}

\begin{figure}
\centering
\subfigure[\textbf{TPC-H} (sf = 10), Query 11]{ \label{fig:space:casestudy:tpch10}
    \includegraphics[width=0.48\columnwidth]{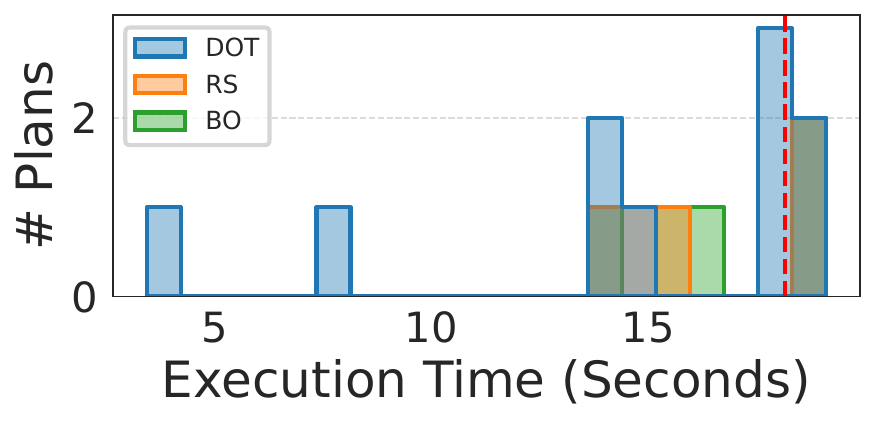}}
\subfigure[\textbf{STATS}, Query 47b]{ \label{fig:space:casestudy:stats}
    \includegraphics[width=0.48\columnwidth]{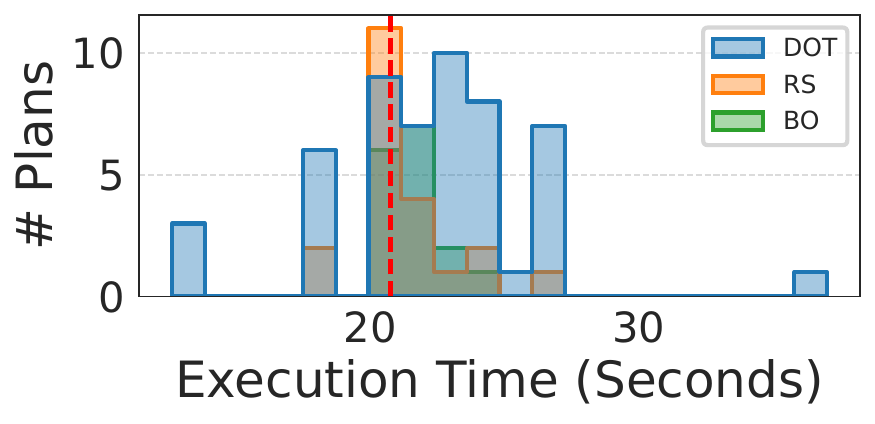}}
\vspace{-1.5em}
\caption{Execution time distribution of the candidate plans found by RS, BO and DOT. 
The red dashed line represents the execution time of the \textbf{Default} plan.}
\vspace{0.5em}
\label{fig:space:casestudy}
\end{figure}



\begin{figure}[t]
\centering
\subfigure[\textbf{TPC-H} (sf = 10), Query 11]{ \label{fig:casestudy:tpch10:q11}
    \includegraphics[clip, trim=1cm 6cm 1cm 5cm, width=1.0\columnwidth]{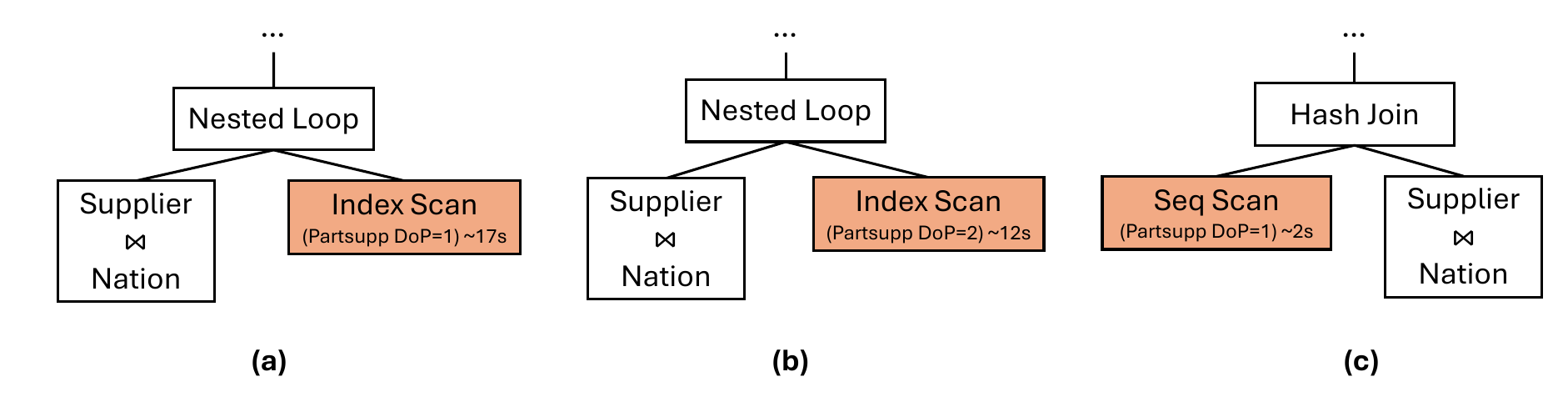}}
\subfigure[\textbf{STATS}, Query 47b]{ \label{fig:casestudy:stats:47b}
    \includegraphics[clip, trim=3cm 4cm 2cm 5cm, width=1.0\columnwidth]{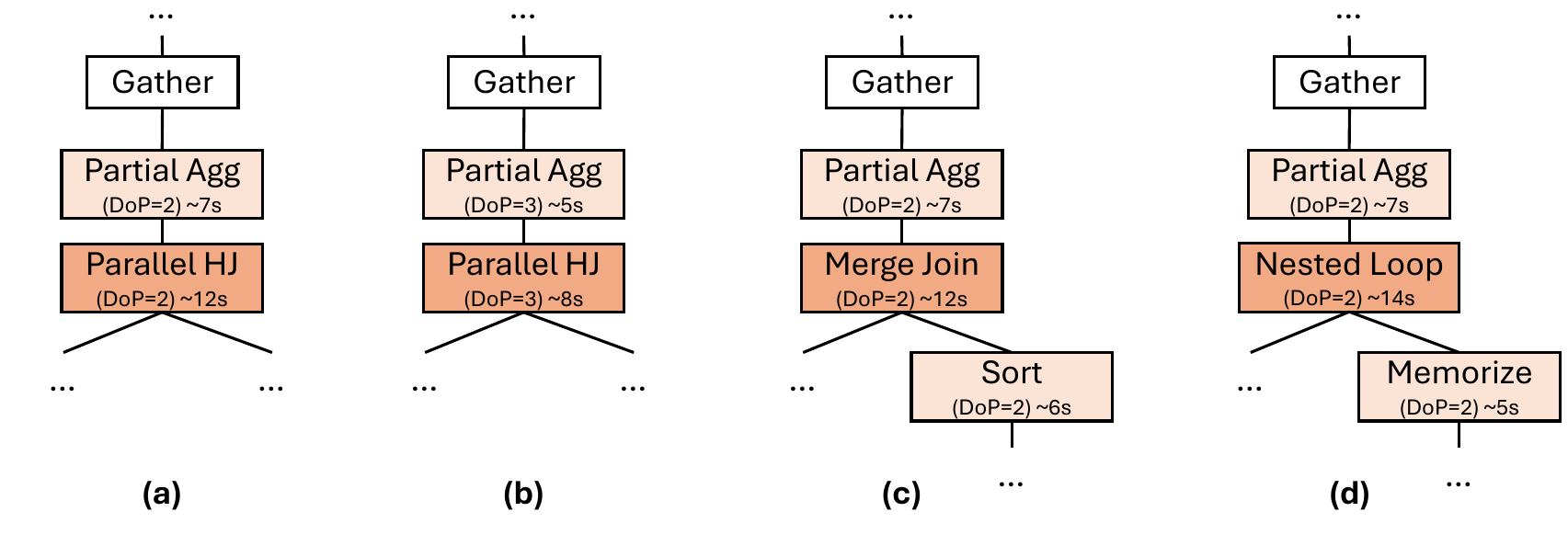}}
\vspace{-1.5em}
\caption{Variations of explored plans by altering CMPs in terms of the bottleneck operators and access paths.}
\label{fig:casestudy:queries}
\end{figure}

Given the large number of candidate plans, another interesting question is how \emph{homogeneous} or \emph{heterogeneous} they are in terms of execution time.
Clearly, the sheer number of plans does not always imply the potential to find much better plans (compared to the default one) if the candidate plans have similar execution time.
Our investigation of this question reveals that the outcome is query-dependent.
In the following, we present a case study with two queries (see Appendix~\ref{appendix:results:case-study-queries}) that exhibit very different structures and patterns in their candidate plans.

\vspace{-0.5em}
\paragraph*{\textbf{TPC-H} (sf = 10) Query 11}
Figure~\ref{fig:space:casestudy:tpch10} presents the distribution of the execution time of the candidate plans found by all the approaches evaluated. 
The blue bars represent the plans found by \textbf{DOT}. 
As shown in Figure~\ref{fig:casestudy:tpch10:q11}, these plans can be classified into three groups w.r.t. the scan operator used for the table \textsc{Partsupp}, which is also the execution bottleneck:
(a) \emph{index scan} without parallelism, which includes five plans with execution time of around 19 seconds;
(b) \emph{index scan} with degree of parallelism (DoP) = 2, which includes three plans with execution time of around 14 seconds;
(c) \emph{sequential scan}, which includes two plans with execution time of 4 and 7 seconds, respectively.
The \textbf{Default} plan falls into group (a). 
Both \textbf{RS} and \textbf{BO} found five plans, all of which fall into the categories of (a) and (b). 
For this query, the candidate plans are \emph{heterogeneous}: most plans have execution time greater than 14 seconds, whereas the best plan completes in around 4 seconds, a reduction of more than 3.5$\times$. 
Only \textbf{DOT} can find this best plan.


\vspace{-0.5em}
\paragraph*{\textbf{STATS} Query 47b}
Figure~\ref{fig:space:casestudy:stats} presents the distribution of the execution time of all discovered candidate plans. 
The plans covered by \textbf{DOT} can be classified into four groups according to the \emph{final} join operator used, as shown in Figure~\ref{fig:casestudy:stats:47b}:
(a) \emph{parallel hash join} with degree of parallelism (DoP) = 2, which contains 44 plans with execution time of around 21 seconds;
(b) \emph{parallel hash join} with DoP = 3, which contains three plans with execution time of around 13 seconds;
(c) \emph{merge join} with DoP = 2, which contains a single plan with execution time of around 37 seconds;
and (d) \emph{nested-loop join} with DoP = 2, which contains four plans with execution time of around 27 seconds. 
The \textbf{Default} plan belongs to group (a). 
Neither \textbf{RS} nor \textbf{BO} discovers any plans from the best group (b).


\vspace{-0.5em}
\paragraph*{Remark}
The fact that grouping candidate plans by their access paths is well correlated with the corresponding regions in the distribution of plan execution time is our main motivation for the similarity-based plan filtering mechanism that relies on matching bottleneck operators (Section~\ref{sec:optimization:similarity:bottleneck}).

\vspace{-0.5em}
\subsection{Reduction of Tuning Time}
\label{sec:eval:tuning-time}

\begin{figure*}
\centering
\subfigure[\textbf{TPC-H}, sf = 1]{ \label{fig:tune:reduction:tpch1}
    \includegraphics[width=0.24\textwidth]{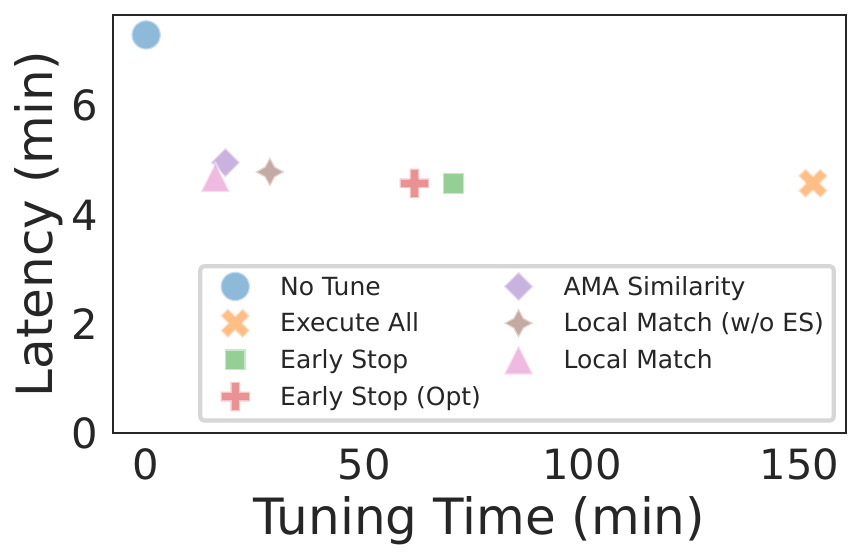}}
\subfigure[\textbf{TPC-DS}, sf = 1]{ \label{fig:tune:reduction:tpcds1}
    \includegraphics[width=0.24\textwidth]{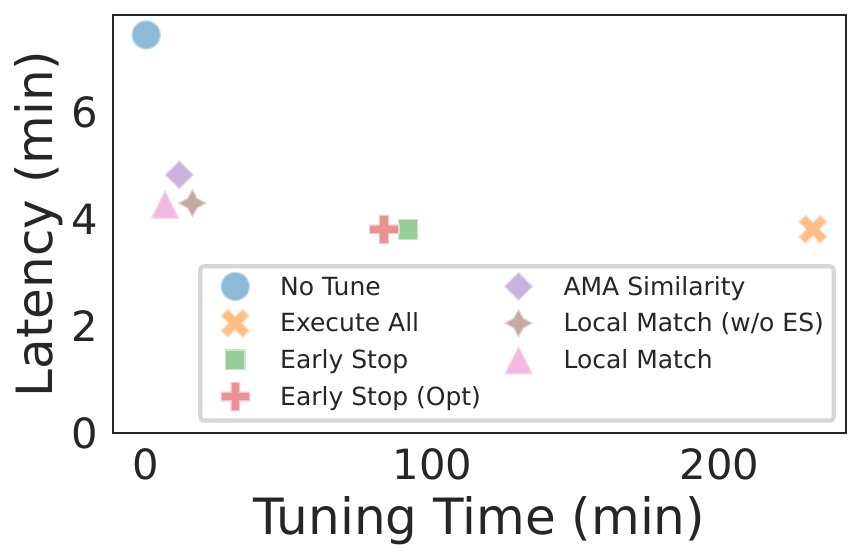}}
\subfigure[\textbf{JOB}]{ \label{fig:tune:reduction:job}
    \includegraphics[width=0.24\textwidth]{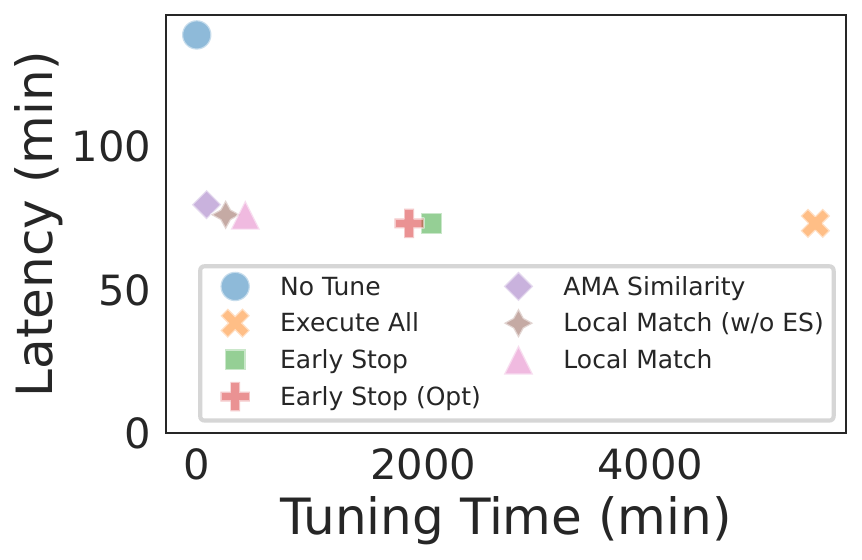}}
\subfigure[\textbf{STATS}]{ \label{fig:tune:reduction:stats}
    \includegraphics[width=0.24\textwidth]{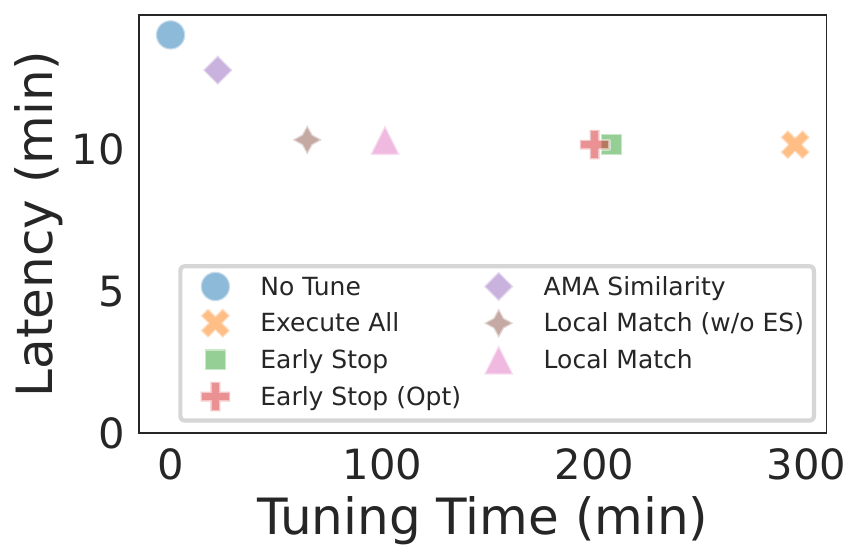}}
\vspace{-1.5em}
\caption{Comparison of variant approaches for reducing tuning time as well as baseline approaches. In particular, the baseline ``No Tune'' by simply using the default query plan results in the shortest tuning time but the longest execution time. \shepherd{Since the tuning overhead is incurred only once, CMP tuning is most appropriate for recurring queries that are executed multiple times.}}
\vspace{-1em}
\label{fig:tune:reduction}
\end{figure*}

We now evaluate the efficacy of the timeout-based and similarity-based mechanisms.
We compare the following variants in terms of how the two mechanisms are combined.
\begin{itemize}
    \item \emph{Early Stop}: This includes only the progressive timeout mechanism (Section~\ref{sec:optimization:timeout:progressive}) without combining with any similarity-based plan filtering mechanisms.
    \item \emph{Early Stop (Opt)}: This represents an ideal case for \emph{Early Stop} where candidate plans are evaluated in increasing order with respect to their execution time.
    \item \emph{AMA Similarity}: This is an implementation of combining the progressive timeout mechanism (i.e., \emph{Early Stop}) with the embedding-based similarity mechanism (Section~\ref{sec:optimization:similarity:embedding}), where we use the plan embedding techniques presented in ``AI meets AI (AMA)''~\cite{DingDM0CN19}. 
    We use $\tau=0.9$ as the similarity threshold of Algorithm~\ref{alg:optimization:embedding-sim}.
    \item \emph{Local Match (w/o ES)}: This implements the bottleneck-based similarity mechanism (Section~\ref{sec:optimization:similarity:bottleneck}).
    \item \emph{Local Match}: This further combines the bottleneck-based similarity mechanism with \emph{Early Stop} (Algorithm~\ref{alg:optimization:timeout-sim-combine}). We set the two input parameters of Algorithm~\ref{alg:optimization:bottleneck-sim} as $\alpha=0.2$ (threshold for determining bottleneck operator) and $M=2$ (threshold for selecting the top bottleneck operators).
\end{itemize}

We compare the mechanisms used for reducing tuning time with the following two simple baseline approaches:
\begin{itemize}
    \item \emph{No Tune}: This corresponds to the \textbf{Default} approach by always sticking to the default plan.
    \item \emph{Execute All}: This corresponds to another straightforward approach by executing all candidate plans without timeout.
\end{itemize}

Figure~\ref{fig:tune:reduction} presents the results of all the variant combinations of timeout-based and similarity-based mechanisms, as well as those of the two baselines, in terms of the trade-off between workload execution time with the best query plans found and the tuning time.
Clearly, the two baselines represent two extremes in this trade-off spectrum: \emph{No Tune} results in the minimum tuning time with the maximum workload execution time, while \emph{Execute All} results in the maximum tuning time but with the minimum execution time.
The results with \textbf{TPC-H} and \textbf{TPC-DS} (sf = 10) are in Appendix~\ref{appendix:results:reduction-tuning-time}.

\vspace{-0.5em}
\subsubsection{Comparison of the Variants}
Compared to \emph{Execute All}, \emph{Early Stop} significantly reduces the tuning time while achieving the same minimum execution time of the workload.
The order of plan execution does make a difference in \emph{Early Stop}, as there is an observable gap between the tuning time of \emph{Early Stop} and \emph{Early Stop (Opt)}.
For example, the tuning time for \textbf{TPC-H} (sf = 1) would have been reduced by 10 minutes if the optimal order of plan execution could be known beforehand, as shown in Figure~\ref{fig:tune:reduction:tpch1}.
\emph{AMA Similarity} further reduces the tuning time by combining embedding-based similarity with \emph{Early Stop}, though it may increase the execution time of the workload.
For example, as shown in Figure~\ref{fig:tune:reduction:stats}, the execution time of the workload \textbf{STATS} using \emph{AMA Similarity} is closer to that of \emph{No Tune} than that of \emph{Execute All}, indicating that embedding-based plan similarity does not always subsume the closeness of plan execution time.
\emph{Local Match} overcomes this challenge faced by \emph{AMA Similarity}: It can achieve the minimum execution time of the workload by retaining the best plans found for individual queries, with
slightly increased tuning time.
One may think that \emph{Local Match} should always outperform \emph{Local Match (w/o ES)}; however, this is interestingly not true, as shown in Figure~\ref{fig:tune:reduction:stats}.
Note that the plan filtering mechanism by checking bottleneck operators can be made more effective if more plans are allowed to complete their executions.
That is, more \emph{remaining} candidate plans can be filtered without execution, which compensates for the increased plan execution time due to eliminating timeout.


\vspace{-0.5em}
\subsubsection{Progressive vs. Round-robin Timeout}
\label{sec:eval:tuning-time:timeout-compare}

\begin{table}
\caption{Comparison of tuning time (in hours) with the progressive and the round-robin timeout mechanisms.}
\vspace{-1em}
    \centering
    \begin{tabular}{l|rrrr}
    \toprule
       {\bf Workload} & $T^{\progressive}$ & $T^{\roundrobin}$ & {\bf Slowdown} $\frac{T^{\roundrobin}}{T^{\progressive}}$\\
    \midrule
         {\bf TPC-H} (sf = 1) & 1.2 & 2.2 & 1.89 $\times$ \\
         {\bf TPC-H} (sf = 10) & 2.7 & 4.8 & 1.76 $\times$ \\ 
         {\bf TPC-DS} (sf = 1) & 1.5 & 3.1 & 2.06 $\times$ \\
         {\bf TPC-DS} (sf = 10) & 22.0 & 46.7 & 2.12 $\times$ \\
         {\bf JOB} & 34.5 & 76.7 & 2.22 $\times$ \\
         {\bf STATS} & 3.5 & 8.2 & 2.36 $\times$ \\
    \bottomrule
    \end{tabular}
    \label{tab:timeout:compare}
\end{table}

As we noted in Section~\ref{sec:optimization:combination}, the round-robin timeout mechanism introduced in Section~\ref{sec:optimization:timeout:rr}, while having an advantage with worst-case guarantees, is difficult to combine with the plan filtering mechanism based on the check of bottleneck operators (Section~\ref{sec:optimization:similarity:bottleneck}).
As a result, we have omitted its evaluation thus far because it is not fair to compare it with any of the compound mechanisms.
Nevertheless, for completeness, we compare it with the progressive timeout mechanism alone.
Table~\ref{tab:timeout:compare} summarizes the results.
We observe that the tuning time with the round-robin timeout mechanism ($T^{\roundrobin}$) is between $1.8\times$ and $2.4\times$ of the tuning time with the progressive timeout mechanism ($T^{\progressive}$), as shown in the last column of Table~\ref{tab:timeout:compare}.


\vspace{-0.5em}
\subsection{Extended Evaluation on PostgreSQL}
\label{sec:eval:pg-extended}

\subsubsection{\rtwo{Tuned Databases}}
\rtwo{It is common practice to tune the configuration knobs of a database system, such as the \emph{shared\_buffers} of PostgreSQL (PG), to improve query execution performance.
As a result, it is interesting to see whether \textbf{DOT} remains beneficial for a tuned database system.
For this purpose, we run the well-known PGTune tool~\cite{pgtune} to recommend the appropriate settings for the knobs and compare \textbf{DOT} with the other baselines on the PG system with these recommended settings.}

\begin{figure}
\centering
\subfigure[\textbf{TPC-H}, sf = 10]{\label{fig:improve:query:tuned:tpch10}
    \includegraphics[width=0.27\columnwidth]{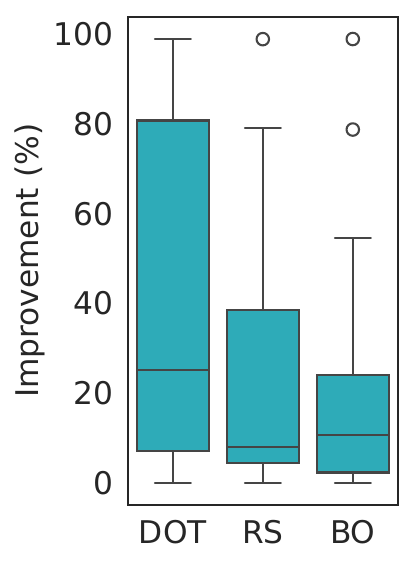}}
\subfigure[\textbf{TPC-DS}, sf = 10]{\label{fig:improve:query:tuned:tpcds10}
    \includegraphics[width=0.27\columnwidth]{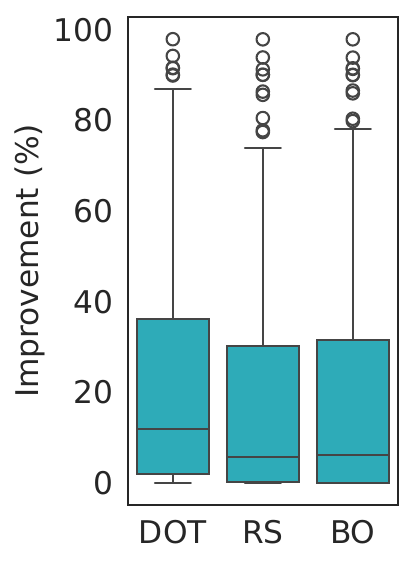}}
\subfigure[\textbf{STATS}]{\label{fig:improve:query:tuned:stats}
    \includegraphics[width=0.27\columnwidth]{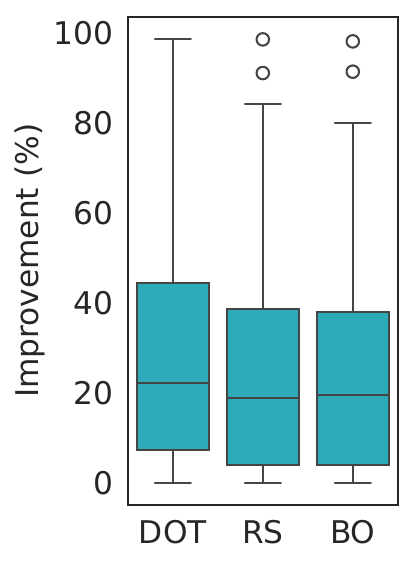}}
\vspace{-1.5em}
\caption{\rv{Distribution of query-level improvement by \textbf{RS}, \textbf{BO}, and \textbf{DOT} over \textbf{Default} on a tuned PG system.}}
\label{fig:improve:query:tuned}
\vspace{0.5em}
\end{figure}

\rtwo{Figure~\ref{fig:improve:query:tuned} summarizes the distributions of query-level execution time improvement.
We observe similar advantages of \textbf{DOT} over the other baselines, compared to the results shown in Figure~\ref{fig:improve:query}.
We leave to the details of individual query improvement to Appendix~\ref{appendix:results:reduction-tuning-time}, due to space constraints (same below).}

\vspace{-0.5em}
\subsubsection{\rv{Warm Cache}}
\rtwo{Although we have been focusing on using a ``cold start'' to execute queries for better reproducibility of the results, in practice it is common with data resident in the memory.
Figure~\ref{fig:improve:query:warm-cache} presents the evaluation results regarding the benefit of \textbf{DOT} in this ``warm cache'' scenario, where we observe similar advantages of \textbf{DOT} compared to the other baselines.}

\begin{figure}[t]
\centering
\subfigure[\textbf{TPC-H}, sf = 10]{\label{fig:improve:query:warm-cache:tpch10}
    \includegraphics[width=0.27\columnwidth]{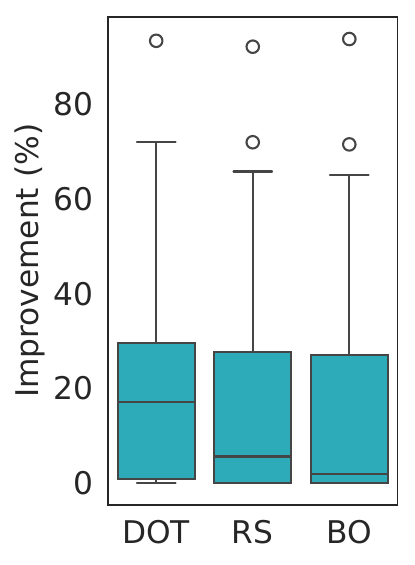}}
\subfigure[\textbf{TPC-DS}, sf = 10]{\label{fig:improve:query:warm-cache:tpcds10}
    \includegraphics[width=0.27\columnwidth]{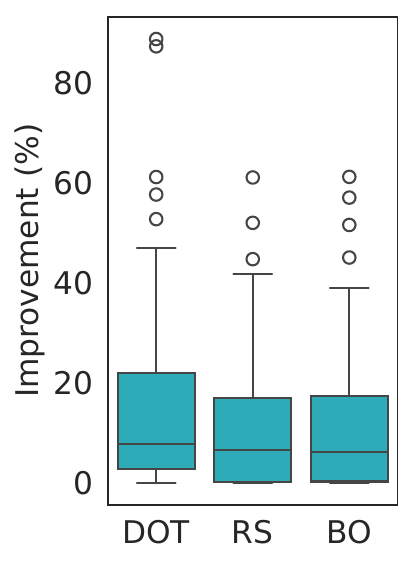}}
\subfigure[\textbf{STATS}]{\label{fig:improve:query:warm-cache:stats}
    \includegraphics[width=0.27\columnwidth]{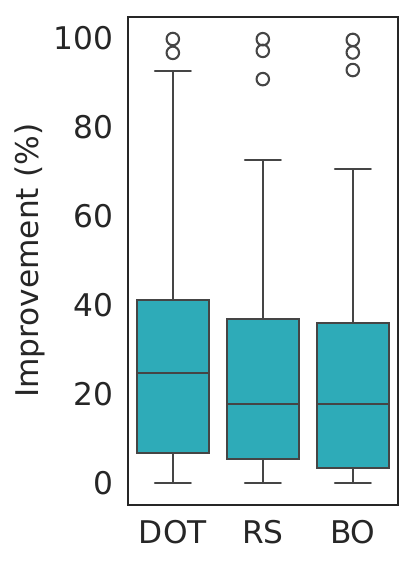}}
\vspace{-1.5em}
\caption{\rv{Distribution of query-level improvement by \textbf{RS}, \textbf{BO}, and \textbf{DOT} over \textbf{Default} with the ``warm cache'' scenario.}}
\label{fig:improve:query:warm-cache}
\vspace{0.5em}
\end{figure}

\vspace{-0.5em}
\subsubsection{\rthree{Databases with Larger Scales}}

\begin{figure}[t]
\centering
\subfigure[\textbf{TPC-H}, sf = 25]{\label{fig:improve:query:scale:tpch25}
    \includegraphics[width=0.26\columnwidth]{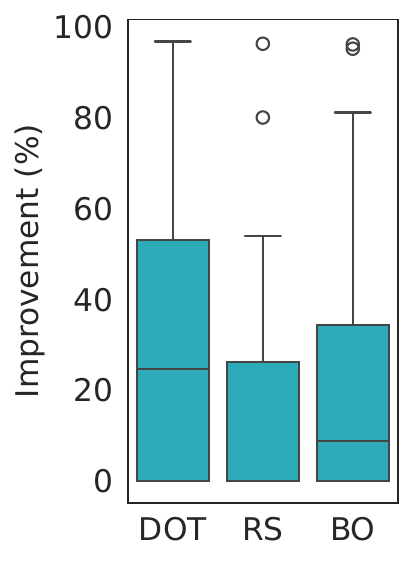}}
\subfigure[\textbf{TPC-DS}, sf = 25]{\label{fig:improve:query:scale:tpcds25}
    \includegraphics[width=0.26\columnwidth]{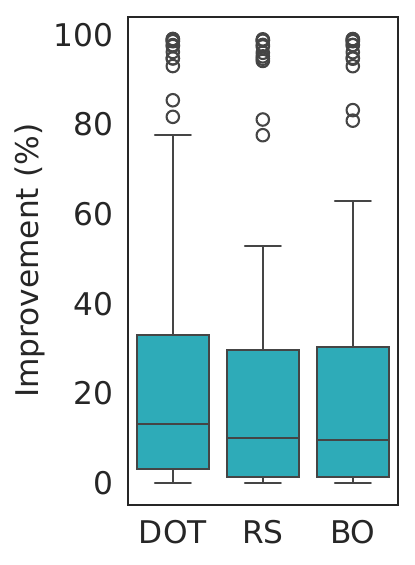}}
\subfigure[\textbf{TPC-H}, sf = 100]{\label{fig:improve:query:scale:tpch100}
    \includegraphics[width=0.26\columnwidth]{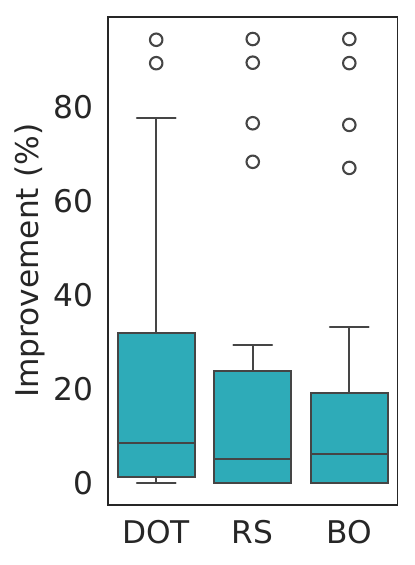}}
\vspace{-1.5em}
\caption{\rthree{Distribution of query-level improvement by \textbf{RS}, \textbf{BO}, and \textbf{DOT} over \textbf{Default} with larger scale databases.}}
\label{fig:improve:query:scale}
\vspace{0.5em}
\end{figure}

\rthree{
We further evaluate \textbf{DOT} on databases with larger scales.
Specifically, we increase the scaling factors of \textbf{TPC-H} and \textbf{TPC-DS} databases to 25 and 100.
Figure~\ref{fig:improve:query:scale} presents the evaluation results. 
We again observe that \textbf{DOT} can significantly outperform \textbf{RS} and \textbf{BO} on these larger scale databases.
}

\vspace{-1.5em}
\subsection{\rthree{Results with Microsoft SQL Server}}

\begin{figure}[t]
\centering
\subfigure[\textbf{TPC-H}, sf = 10]{\label{fig:improve:query:mssql:tpch10}
    \includegraphics[width=0.26\columnwidth]{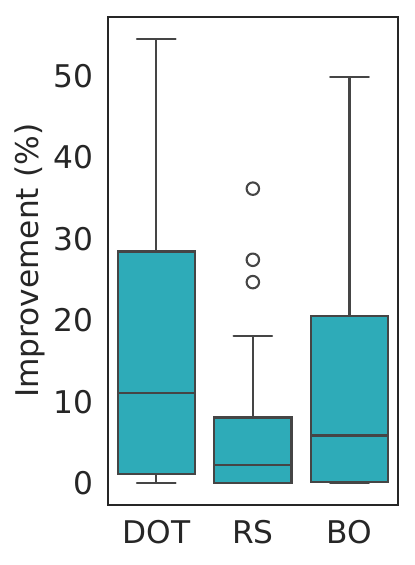}}
\subfigure[\textbf{TPC-DS}, sf = 10]{\label{fig:improve:query:mssql:tpcds10}
    \includegraphics[width=0.26\columnwidth]{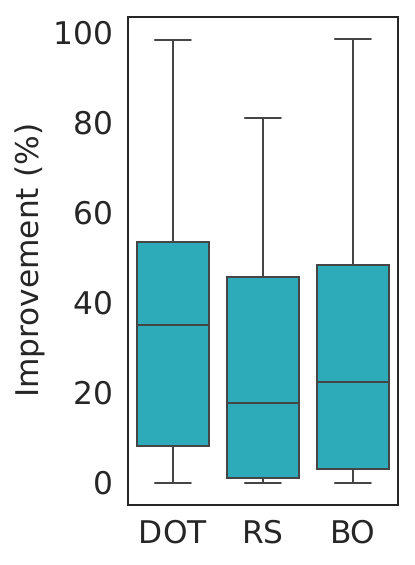}}
\subfigure[\textbf{JOB}]{\label{fig:improve:query:mssql:job}
    \includegraphics[width=0.26\columnwidth]{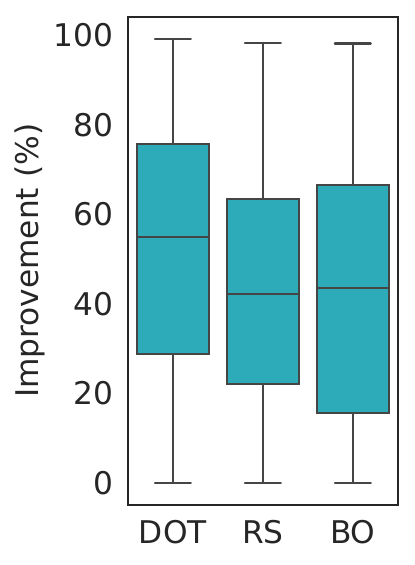}}
\vspace{-1.5em}
\caption{\rthree{Distribution of query-level improvement by \textbf{RS}, \textbf{BO}, and \textbf{DOT} over \textbf{Default} with Microsoft SQL Server.}}
\label{fig:improve:query:mssql}
\end{figure}

\rthree{To test the general applicability of \textbf{DOT}, we further evaluate it using Microsoft SQL Server. Following previous work~\cite{WuW24}, we focus on tuning CMPs that are critical to the costs of workhorse operators such as \emph{table scan}, \emph{index seek}, \emph{hash join}, and \emph{nested-loop join}.
Figure~\ref{fig:improve:query:mssql} presents the evaluation results.
We observe that \textbf{DOT} remains advantageous over \textbf{RS} and \textbf{BO} on the three workloads tested, namely, \textbf{TPC-H} (sf = 10), \textbf{TPC-DS} (sf = 10), and \textbf{JOB}.}

\vspace{-0.5em}
\section{Related Work}
\label{sec:related-work}

\paragraph*{CMP Tuning}
The idea of tuning CMPs to find better query execution plans has been sketched in the literature.
Most previous work views CMPs as platform-dependent constants that require only a one-shot calibration~\cite{DuKS92,SororMASKK08,WuCHN13,WuWHN14,WuNS16,YangWZDLC23}.
The alternative view of treating CMPs as query-dependent variables was first introduced by recent work~\cite{WuW24}, where the authors proposed to view the CMP tuning problem from an AutoML perspective.
However, the focus of~\cite{WuW24} remains on using random search or Bayesian optimization to address the CMP tuning problem.

\vspace{-0.5em}
\paragraph*{Database Knob Tuning}
In recent years, there has been intensive and extensive research on the problem of database knob tuning. 
In a nutshell, modern database systems typically have configurable knobs that require careful tuning to find their best set-ups to accelerate
query execution. For example, setting the appropriate size for the \emph{shared\_buffers} of PostgreSQL is critical to database system performance. 
Existing solutions to database knob tuning are based on Bayesian optimization~\cite{DuanTB09,AkenPGZ17,CeredaVCD21,KanellisDKMCV22}, reinforcement learning~\cite{LiZLG19,ZhangLZLXCXWCLR19,AkenYBFZBP21,WangTB21}, and most recently large language models~\cite{Trummer24,LaoWLWZCCTW25,GiannakourisT25}.
Moreover, some recent approaches~\cite{WangTB21,GiannakourisT25} further include \emph{index tuning} as part of the database knob tuning problem, although index tuning is generally viewed as a different problem that itself has been studied extensively~\cite{ChaudhuriN97,ValentinZZLS00,DashPA11,SchlosserK019,KossmannHJS20,WuWSWNCB22,KossmannKS22,ml-index-tuning-overview,Wii,Esc}.
Database knob tuning aims to find the globally optimal knob settings that maximize the performance of a multi-query workload, whereas CMP tuning is performed at a per-query basis with the goal of maximizing the performance of an individual query.
This setup of per-query tuning cannot be applied to database knob tuning because many knobs, such as \emph{shared\_buffers}, can only be set for the entire database system rather than for each query session.


\vspace{-0.5em}
\paragraph*{Query Hinting}
There has also been active research along the line of using query hints to steer the query optimizer~\cite{MarcusNMTAK21, WoltmannTHHL23, AnneserTCXPLM23} in the hope of finding better execution plans. 
Although both query hinting and CMP tuning share the goal of improving the query plan, they are orthogonal \rthree{techniques} that can be combined to further improve their efficacy.
On the one hand, query hints can be thought of as \emph{binary} switches that constrain the search space of the query optimizer;
on the other hand, the \emph{continuous} CMPs studied in this paper alter the plan costs in the entire plan space to allow the query optimizer to make different decisions.
Under the objective of finding the best query plan with minimum execution time, both approaches can be viewed as different ways of obtaining a diverse set of candidate plans, and, therefore, are complementary to each other.
Moreover, recent work~\cite{YiTIM25} also considered the problem of reducing the plan evaluation time in the context of query hinting.

\vspace{-0.5em}
\paragraph*{Robust Query Optimization and Processing}
The techniques for plan space exploration developed in this paper are closely related to the decades of research on robust query optimization (QO) and query processing (QP)~\cite{yin2015robust,Haritsa24}.
This line of work aims to overcome the potential mistakes made by the query optimizer when generating the query execution plans, due to reasons such as uncertainty in cardinality estimation and inaccuracy in cost modeling~\cite{IoannidisC91,ChuHG02,BabcockC05,WuCZTHN13,WuWHN14,LeisGMBK015}.
One common idea is to retain \emph{multiple} access paths at certain points of the query plan (or multiple query plans) and to defer the decision on the selection of the access path (or plan) to the query execution runtime~\cite{GraefeW89,MarklRSLP04,ReddyH05,DuttH14,WuNS16}.
In particular, our definition of the plan space in Section~\ref{sec:exploration-plan-space} is inspired by the notion of \emph{plan diagram} developed in a series of previous work on robust QO~\cite{ReddyH05,DDH08,AbhiramaBDSH10},  
although the plan diagram is based on varying the selectivity estimates 
rather than the CMPs.


\vspace{-0.5em}
\section{Conclusion}
\label{sec:conclusion}

We presented a deterministic and optimal approach to the problem of CMP tuning.
It contains two major components: (1) a complete exploration strategy to harvest all possible candidate plans in the search space without the overhead of an exhaustive enumeration, and (2) a combination of timeout-based and similarity-based plan filtering mechanisms to minimize evaluation time of candidate plans.
Unlike existing approaches that rely on probabilistic mechanisms such as random search or Bayesian optimization, our approach avoids indeterminism and uncertainty in the outcome of tuning while offering guarantees on finding the best plan in the search space.
Experimental results on top of widely used public database benchmarks with PostgreSQL \rthree{and Microsoft SQL Server} demonstrated the efficacy of our approach, showing that it can find query
plans that are orders of magnitude faster than both the default query plan returned by the query optimizer and the best plan found by randomized approaches.

\clearpage

\bibliographystyle{ACM-Reference-Format}
\bibliography{./paper}

\clearpage
\appendix

\section{Proofs}

\subsection{Proof of Theorem~\ref{theorem:one-dim}}
\label{appendix:proofs:theorem:one-dim}

\begin{proof}
Suppose that the theorem is not true. Then there exists some $c_l<c_y<c_h$ s.t. $P_y\neq P_l$, where $P_y$ is the plan when setting $c=c_y$.
Let the cost function of $P_y$ be $f_y$. 
We have 
\begin{equation}\label{eq:linearity:hypothesis}
f_y(c_y)<f_l(c_y)    
\end{equation}
by Assumption~\ref{assumption:no-same-cost}, since otherwise the query optimizer would not have returned $P_y$ when setting $c=c_y$.
Since $f_y$ and $f_l$ are both linear functions of $c$, we have
$f_y(c)=a_y\cdot c + b_y$ and $f_l(c)=a_l\cdot c + b_l$. Now define
$$g(c)=f_y(c)-f_l(c)=(a_y-a_l)\cdot c+(b_y-b_l).$$
It is clear that $g$ is also a linear function of $c$. Since $c_l<c_y<c_h$, we can write
$$c_y=\rho\cdot c_l+(1-\rho)\cdot c_h,$$
where $0<\rho< 1$. Due to the linearity of $g$, we have
$$g(c_y)=\rho\cdot g(c_l)+(1-\rho)\cdot g(c_h).$$
Since $g(c_l)>0$ and $g(c_h)>0$, we have $g(c_y)>0$, which is contradictory with Equation~\ref{eq:linearity:hypothesis} that implies $g(c_y)<0$.
This completes the proof of the theorem.
\end{proof}

\subsection{Proof of Theorem~\ref{theorem:multi-dim}}
\label{appendix:proofs:theorem:multi-dim}

To prove this theorem, we need the following lemma.
\begin{lemma}\label{lemma:convex-combination}
Each $\vec{u}\in\mathcal{R}$ can be represented as a \emph{convex combination} of the vertices $\vec{v}\in\mathcal{V}$.
That is, for each vertex $\vec{v}\in\mathcal{V}$ there exists some $\rho_v$ s.t. (1) $\rho_v\geq 0$, (2) $\sum_{\vec{v}\in\mathcal{V}}\rho_v=1$, and (3) $\vec{u}=\sum_{\vec{v}\in\mathcal{V}}\rho_v\cdot\vec{v}$.
\end{lemma}

\begin{proof}
We prove the lemma by construction.
Given any $\vec{u}\in\mathcal{R}$ with $\vec{u}=(u_1, ..., u_K)$, we define
$\theta_i=\frac{u_i-c_i^l}{c_i^h-c_i^l}$
for $1\leq i\leq K$.
Clearly, $\theta_i\in[0,1]$.
Consider a vertex $\vec{v}\in\mathcal{V}$ with $\vec{v}=(v_1, ..., v_K)$.
Define 
$$
\lambda_i = 
    \begin{cases}
    1-\theta_i, & \text{if } v_i=c_i^l; \\
    \theta_i, & \text{otherwise (i.e, } v_i=c_i^h\text{)}.
    \end{cases}  
$$
We can now set $\rho_v=\prod_{i=1}^K \lambda_i$. This completes the proof.
\end{proof}

\begin{example}
It is easy to verify the correctness of $\rho_v$ constructed in the proof of Lemma~\ref{lemma:convex-combination}.
We take the two-dimensional plan space in Figure~\ref{fig:plan-cost-multi-dim} as an example.
For any $\vec{u}=(u_1, u_2)$, define
$$\theta_1=\frac{u_1-c_1^l}{c_1^h-c_1^l},\quad \theta_2=\frac{u_2-c_2^l}{c_2^h-c_2^l}.$$
We then set the $\lambda_v$'s as follows:
$$\lambda_{(c_1^l, c_2^l)}=(1-\theta_1)(1-\theta_2),$$
$$\lambda_{(c_1^l, c_2^h)}=(1-\theta_1)\theta_2,$$
$$\lambda_{(c_1^h, c_2^l)}=\theta_1(1-\theta_2),$$
$$\lambda_{(c_1^h, c_2^h)}=\theta_1\theta_2.$$
One can check the three conditions (1), (2), and (3) for the $\lambda_v$'s.
\end{example}

\begin{proof}[Proof of Theorem~\ref{theorem:multi-dim}]
Suppose that there exists some $\vec{u}\in\mathcal{R}$ s.t. the query plan $P(q,\vec{u})\neq P_q$.
Let the (linear) cost functions of the two plans $P(q,\vec{u})$ and $P_q=P(q,\vec{v})$ be
$f_u(\vec{c})$ and $f_v(\vec{c})$.

Define $g(\vec{c})=f_u(\vec{c})-f_v(\vec{c})$.
Since both $f_u(\vec{c})$ and $f_v(\vec{c})$ are linear, $g(\vec{c})$ is a liner function of $\vec{c}$ as well.
We have $g(\vec{u})<0$ because the query optimizer returned $P(q,\vec{u})$ instead of $P(q,\vec{v})$ when setting $\vec{c}=\vec{u}$.
By Lemma~\ref{lemma:convex-combination}, $\vec{u}$ can be expressed as a convex (and therefore linear) combination of the vertices, namely $\vec{u}=\sum_{\vec{v}\in\mathcal{V}}\rho_v\cdot\vec{v}$.
Applying the linearity of $g(\vec{u})$ yields
$$g(\vec{u})=g(\sum\nolimits_{\vec{v}\in\mathcal{V}}\rho_v\cdot\vec{v})=\sum\nolimits_{\vec{v}\in\mathcal{V}}\rho_v\cdot g(\vec{v})=g(\vec{v})\cdot\sum\nolimits_{\vec{v}\in\mathcal{V}}\rho_v.$$
Again by Lemma~\ref{lemma:convex-combination}, $\sum_{\vec{v}\in\mathcal{V}}\rho_v=1$.
As a result, we obtain $g(\vec{u})=g(\vec{v})$, which implies $g(\vec{v})<0$ for any vertex $\vec{v}\in\mathcal{V}$.
That is, $f_u(\vec{v})<f_v(\vec{v})$ for any vertex $\vec{v}$, which implies that the query optimizer should have returned $P(q,\vec{u})$ rather than $P(q,\vec{v})$ when setting $\vec{c}=\vec{v}$.
This contradiction completes the proof of the theorem.
\end{proof}

\subsection{Proof of Theorem~\ref{theorem:progressive-timeout}}
\label{appendix:proofs:theorem:progressive-timeout}

\begin{proof}
Clearly, $T^{\progressive}$ achieves its minimum when we process the plans in the order of increasing execution time.
In this case, all plans except the best one are terminated after hitting the timeout $\delta=t_{\min}$.
On the other hand, $T^{\progressive}$ achieves its maximum when we process the plans in the order of decreasing execution time.
In this case, no plan can be terminated early by the timeout before its completion. As a result, 
$T^{\progressive}=\sum\nolimits_{P\in\mathcal{P}_q}t(P)\leq N\cdot t_{\max}.$
\end{proof}

\subsection{Proof of Theorem~\ref{theorem:rr-timeout}}
\label{appendix:proofs:theorem:rr-timeout}

\begin{proof}
Let $\delta_r$ be the timeout used in the round $r$ of Algorithm~\ref{alg:optimization:roundrobin-timeout}.
Let $m$ be the round in which the plan with execution time $t_{\min}$ is completed.
We have $\delta_{m-1}<t_{\min}\leq\delta_{m}$.
Since no plan can finish in less than $t_{\min}$, $m$ is the earliest round in which some plan can finish execution.
Therefore, Algorithm~\ref{alg:optimization:roundrobin-timeout} can terminate after $m$ rounds.
The total tuning time $T^{\roundrobin}$ can therefore be bounded as
$$N\cdot\sum\nolimits_{r=1}^{m-1}\delta_r\leq T^{\roundrobin}\leq N\cdot\sum\nolimits_{r=1}^m\delta_r$$
Moreover, we have
$$\sum\nolimits_{r=1}^m\delta_r=\delta_0\cdot\sum\nolimits_{r=1}^m 2^{r-1}=\delta_0\cdot(2^m-1).$$
As a result, it follows that
$$N\cdot\delta_0\cdot (2^{m-1}-1)\leq T^{\roundrobin}\leq N\cdot\delta_0\cdot (2^m-1),$$
namely, $N\cdot(\delta_m-\delta_0)\leq T^{\roundrobin}\leq N\cdot(2\cdot\delta_m-\delta_0)$ by noticing $\delta_m=\delta_0\cdot 2^{m-1}$.
This completes the proof.
\end{proof}

\section{More Details of Experimental Results}

\subsection{\rthree{Query-level Improvement in Section~\ref{sec:eval:exec-time}}}
\label{appendix:results:query-level-impr}

\begin{figure}
\centering
\subfigure[\textbf{TPC-H}, sf = 1]{ \label{fig:improve:singlequery:tpch1}
    \includegraphics[width=1.0\columnwidth]{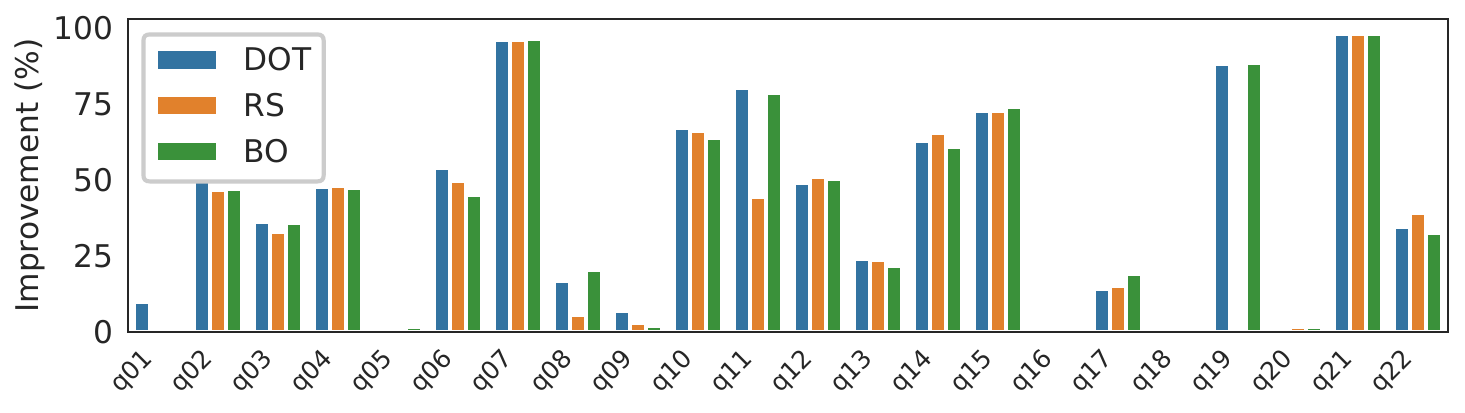}}
\subfigure[\textbf{TPC-DS}, sf = 1]{ \label{fig:improve:singlequery:tpcds1}
    \includegraphics[width=1.0\columnwidth]{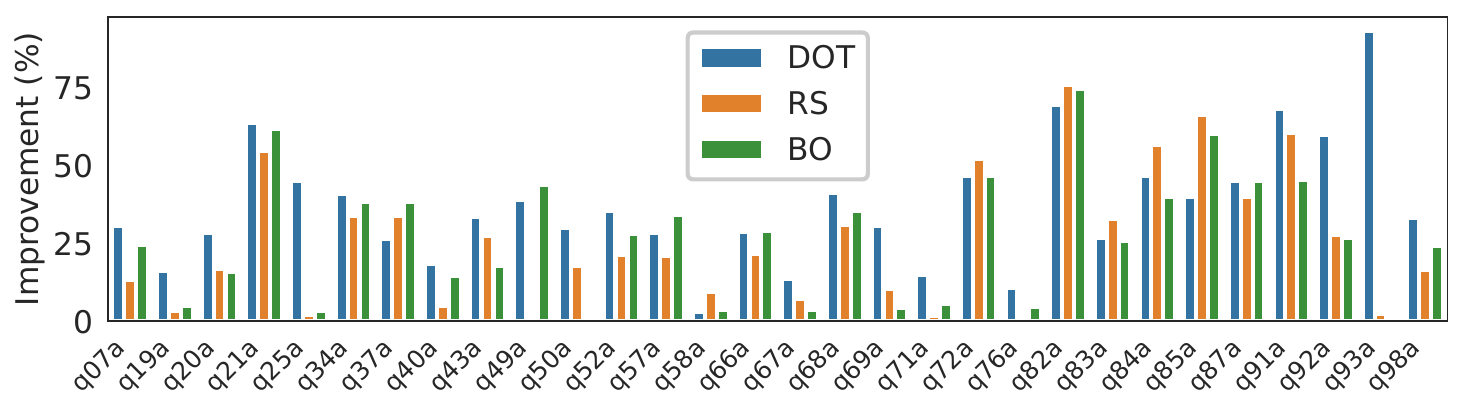}}
\subfigure[\textbf{TPC-DS}, sf = 10]{ \label{fig:improve:singlequery:tpcds10}
    \includegraphics[width=1.0\columnwidth]{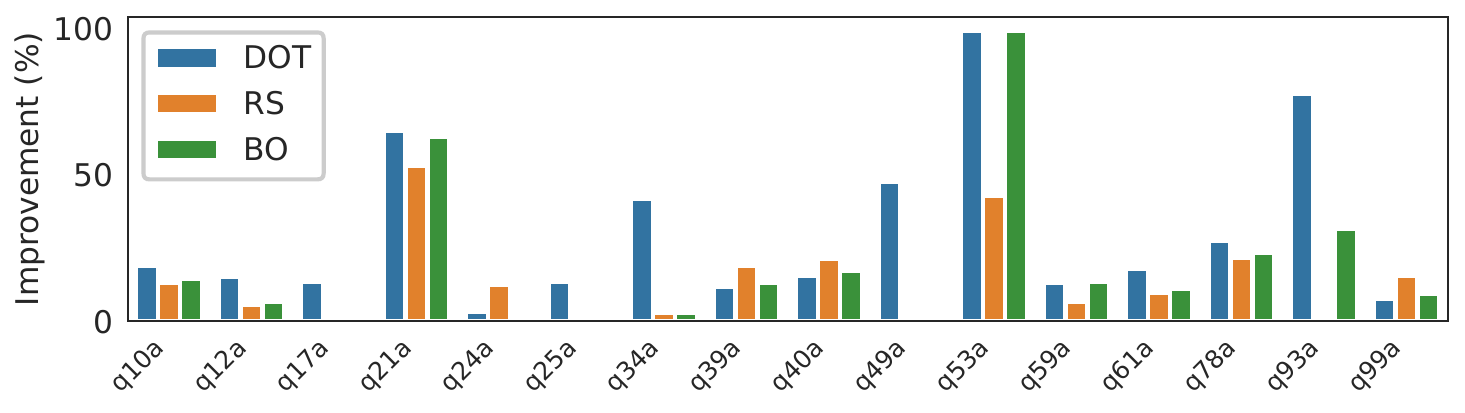}}
\subfigure[\textbf{JOB}, Part I]{ \label{fig:improve:singlequery:job:part-1}
    \includegraphics[width=1.0\columnwidth]{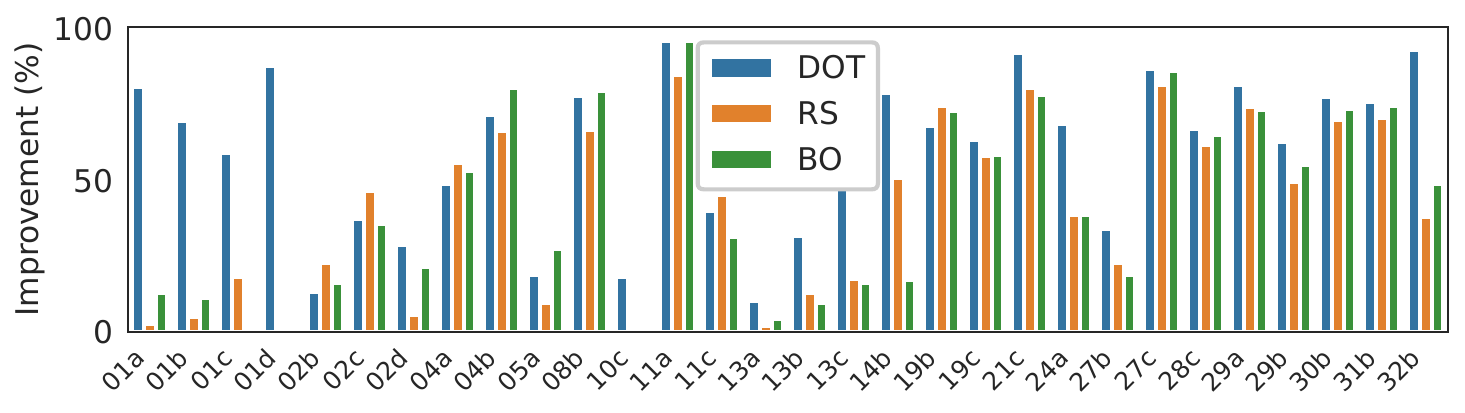}}
\subfigure[\textbf{JOB}, Part II]{ \label{fig:improve:singlequery:job:part-2}
    \includegraphics[width=1.0\columnwidth]{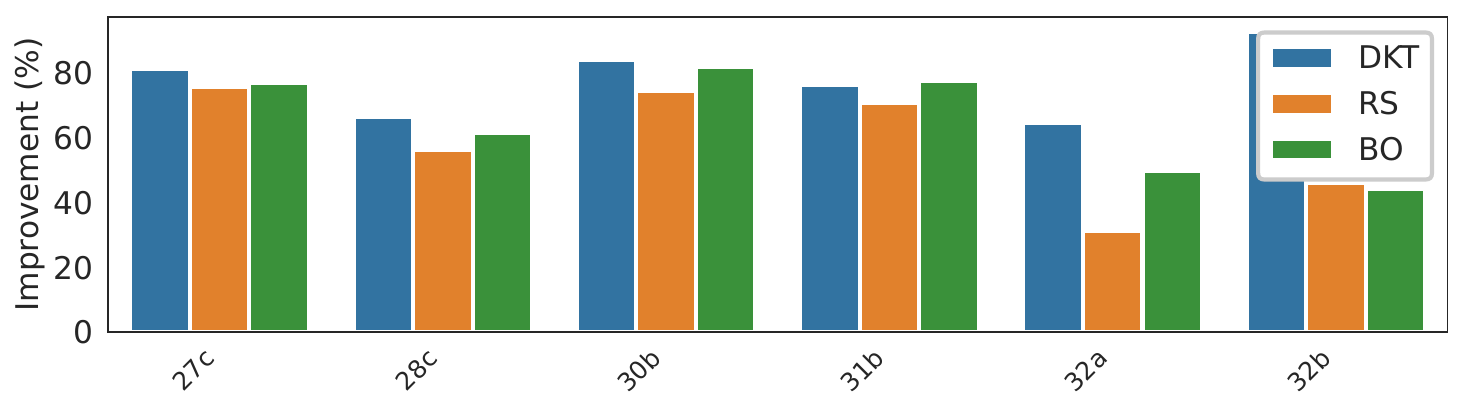}}
\subfigure[\textbf{STATS}, Part I]{ \label{fig:improve:singlequery:stats:part-1}
    \includegraphics[width=1.0\columnwidth]{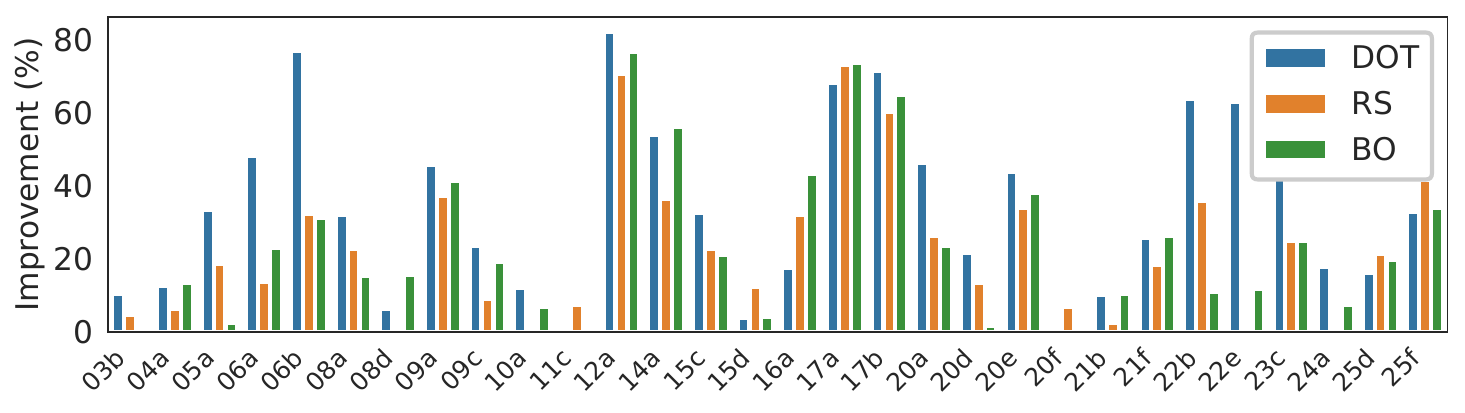}}
\subfigure[\textbf{STATS}, Part II]{ \label{fig:improve:singlequery:stats:part-2}
    \includegraphics[width=1.0\columnwidth]{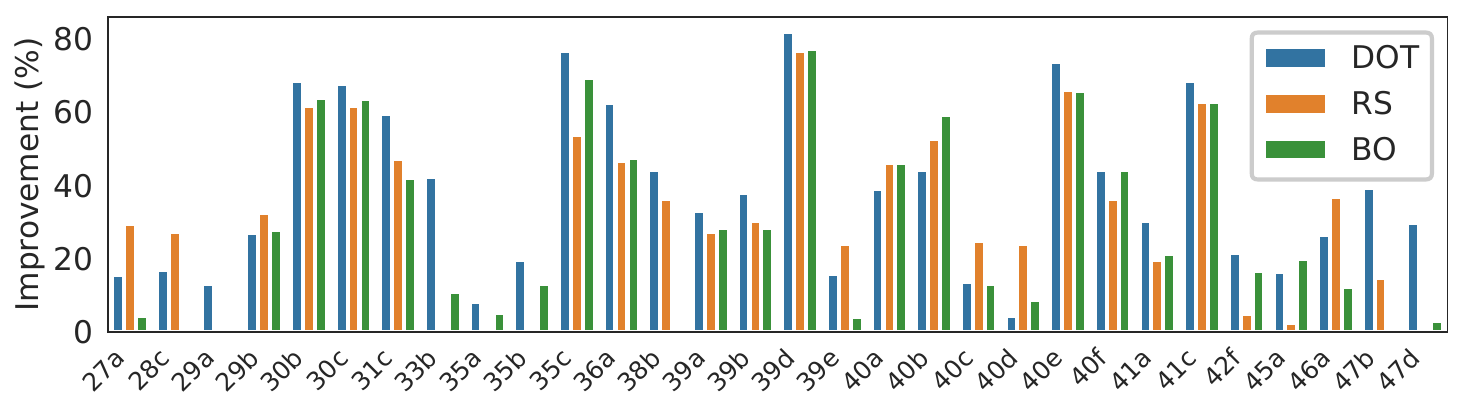}}
\vspace{-1.5em}
\caption{\rthree{Improvements over \textbf{Default} by \textbf{RS}, \textbf{BO}, and \textbf{DOT} for individual queries from each workload.}}
\label{fig:improve:singlequery:rest}
\end{figure}

\rthree{Figure~\ref{fig:improve:singlequery:rest} presents the improvements of individual queries for the remaining workloads.
To avoid clutter, we only show queries with improvements obtained by \textbf{RS} and \textbf{DOT} differing by at least 5\% for the workloads \textbf{TPC-DS}, \textbf{JOB}, and \textbf{STATS}.}

\vspace{-0.5em}
\subsection{\rthree{Analysis of Search Space in Section~\ref{sec:eval:exploration:plan-config}}}
\label{appendix:results:analysis-search-space}

\begin{figure*}
\centering
\subfigure[\textbf{TPC-H}, sf = 1]{ \label{fig:space:plans:tpch1}
    \includegraphics[width=0.28\textwidth]{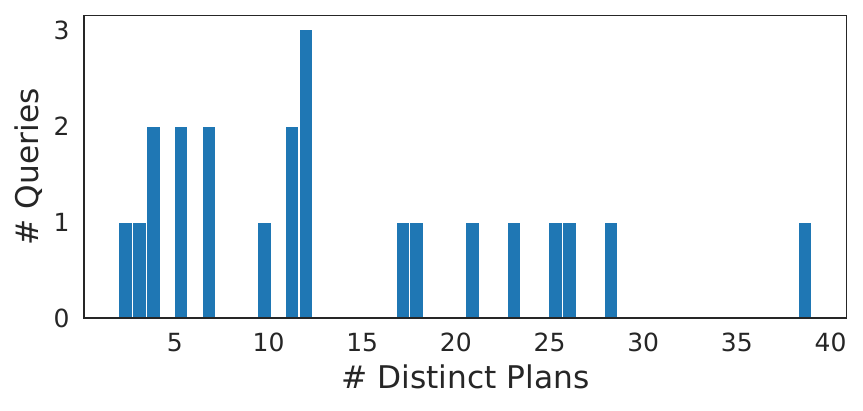}}
\subfigure[\textbf{TPC-DS}, sf = 1]{ \label{fig:space:plans:tpcds1}
    \includegraphics[width=0.28\textwidth]{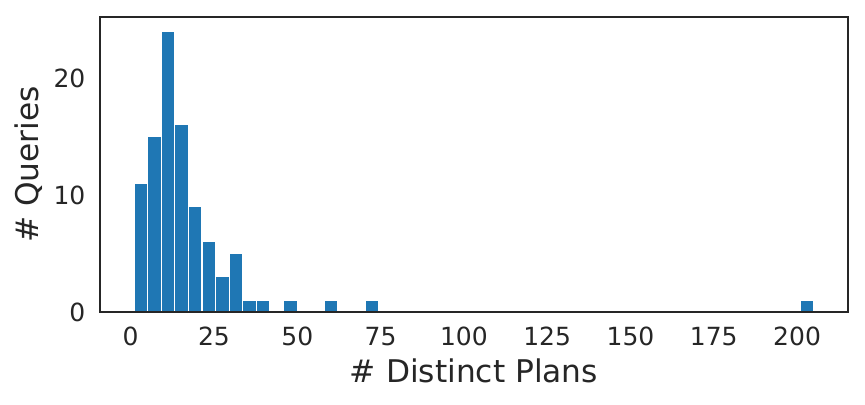}}
\subfigure[\textbf{STATS}]{ \label{fig:space:plans:stats}
    \includegraphics[width=0.28\textwidth]{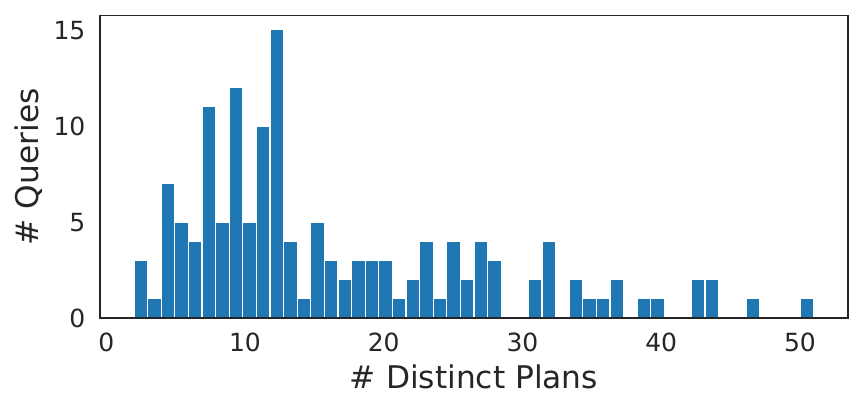}}
\vspace{-1.5em}
\caption{\rthree{Distribution of the number of distinct query plans found during plan space exploration by \textbf{DOT}.}}
\label{fig:space:plans:rest}
\end{figure*}

\begin{figure*}
\centering
\subfigure[\textbf{TPC-H}, sf = 1]{ \label{fig:space:configs:tpch1}
    \includegraphics[width=0.28\textwidth]{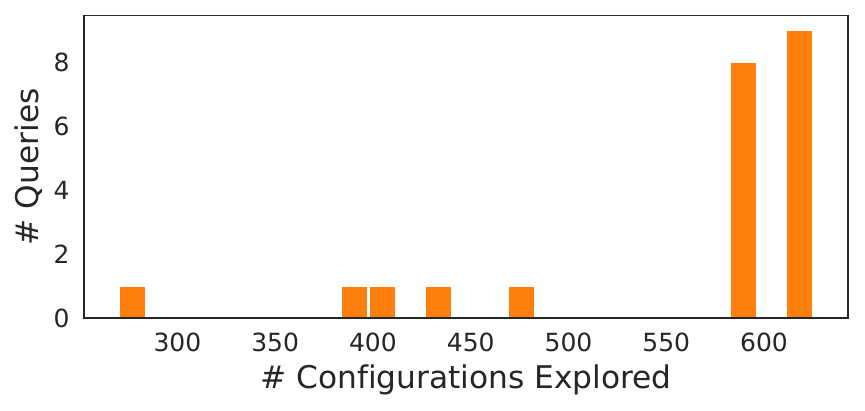}}
\subfigure[\textbf{TPC-DS}, sf = 1]{ \label{fig:space:configs:tpcds1}
    \includegraphics[width=0.28\textwidth]{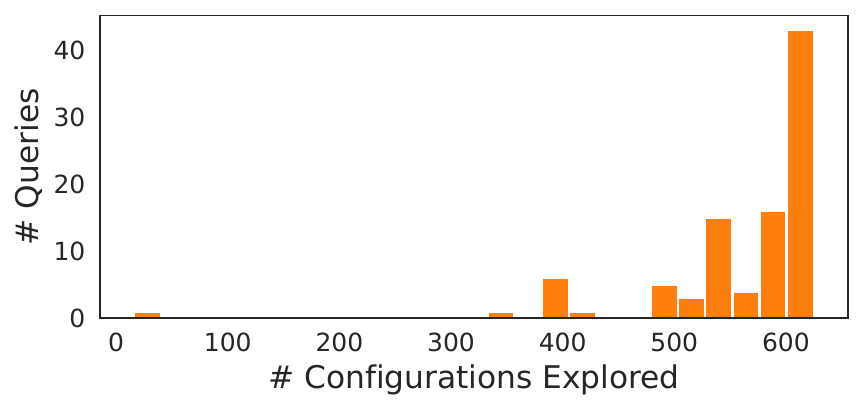}}
\subfigure[\textbf{STATS}]{ \label{fig:space:configs:stats}
    \includegraphics[width=0.28\textwidth]{figures/expr/space_num_configs_stats-eps-converted-to.pdf}}
\vspace{-1.5em}
\caption{\rthree{Distribution of the number of CMP configurations explored (i.e., the number of query optimizer calls) by \textbf{DOT}.}}
\label{fig:space:configs:rest}
\end{figure*}

\begin{figure*}
\centering
\subfigure[\textbf{TPC-H}, sf = 1]{ \label{fig:plan-config-per-query:tpch1}
    \includegraphics[width=0.29\textwidth]{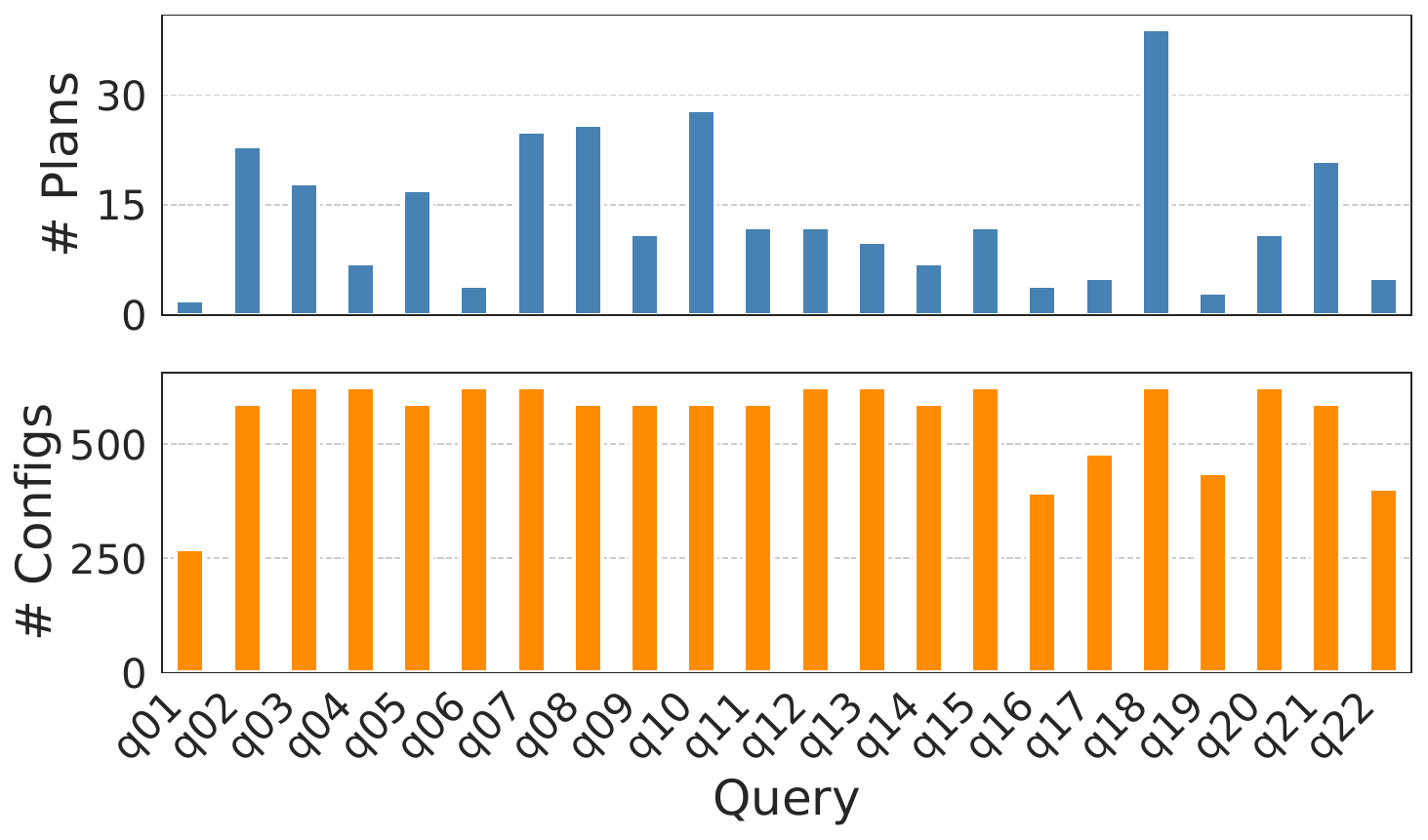}}
\subfigure[\textbf{TPC-DS}, sf = 1, Part I]{ \label{fig:plan-config-per-query:tpcds1:1}
    \includegraphics[width=0.29\textwidth]{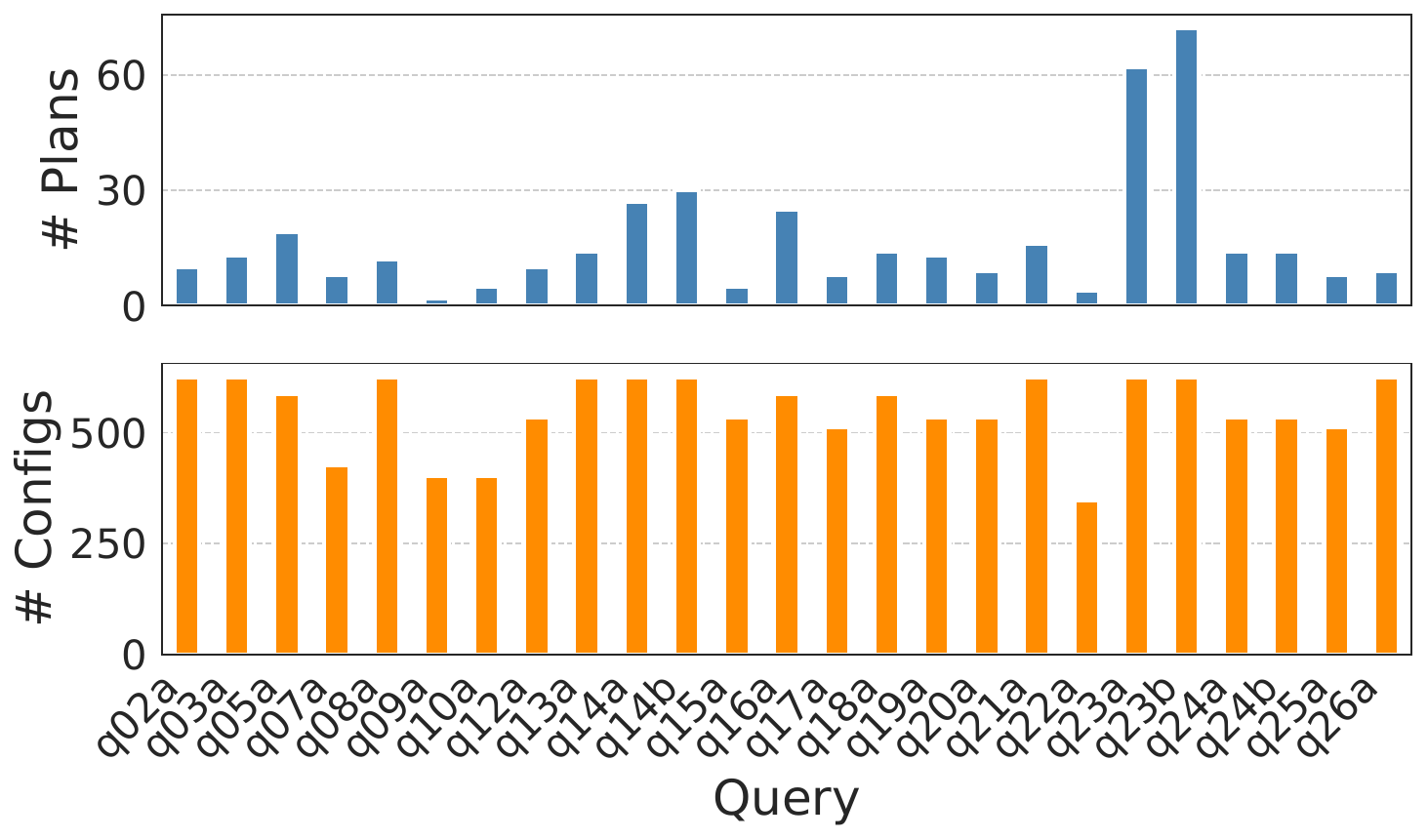}}
\subfigure[\textbf{TPC-DS}, sf = 1, Part II]{ \label{fig:plan-config-per-query:tpcds1:2}
    \includegraphics[width=0.29\textwidth]{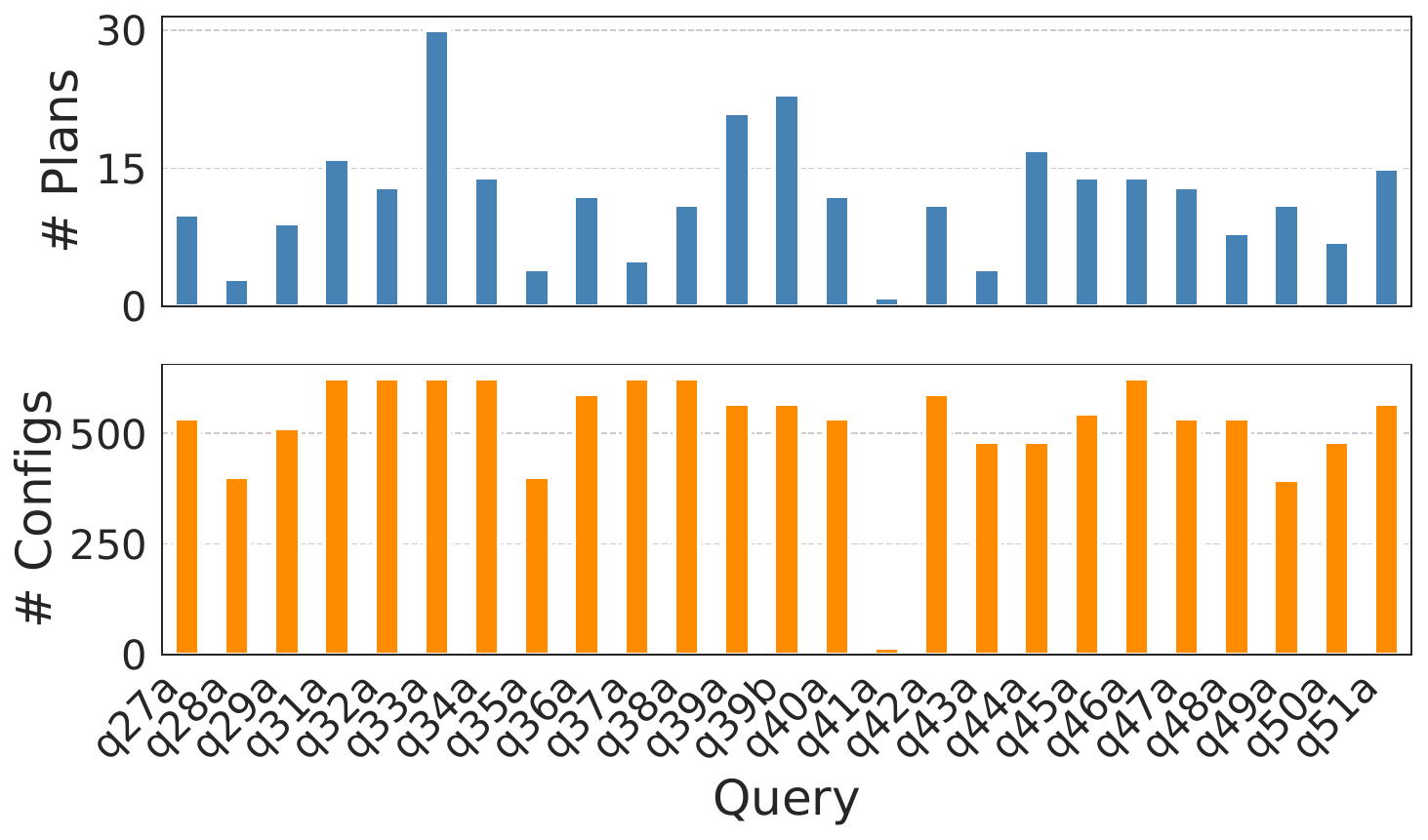}}
\subfigure[\textbf{TPC-DS}, sf = 1, Part III]{ \label{fig:plan-config-per-query:tpcds1:3}
    \includegraphics[width=0.29\textwidth]{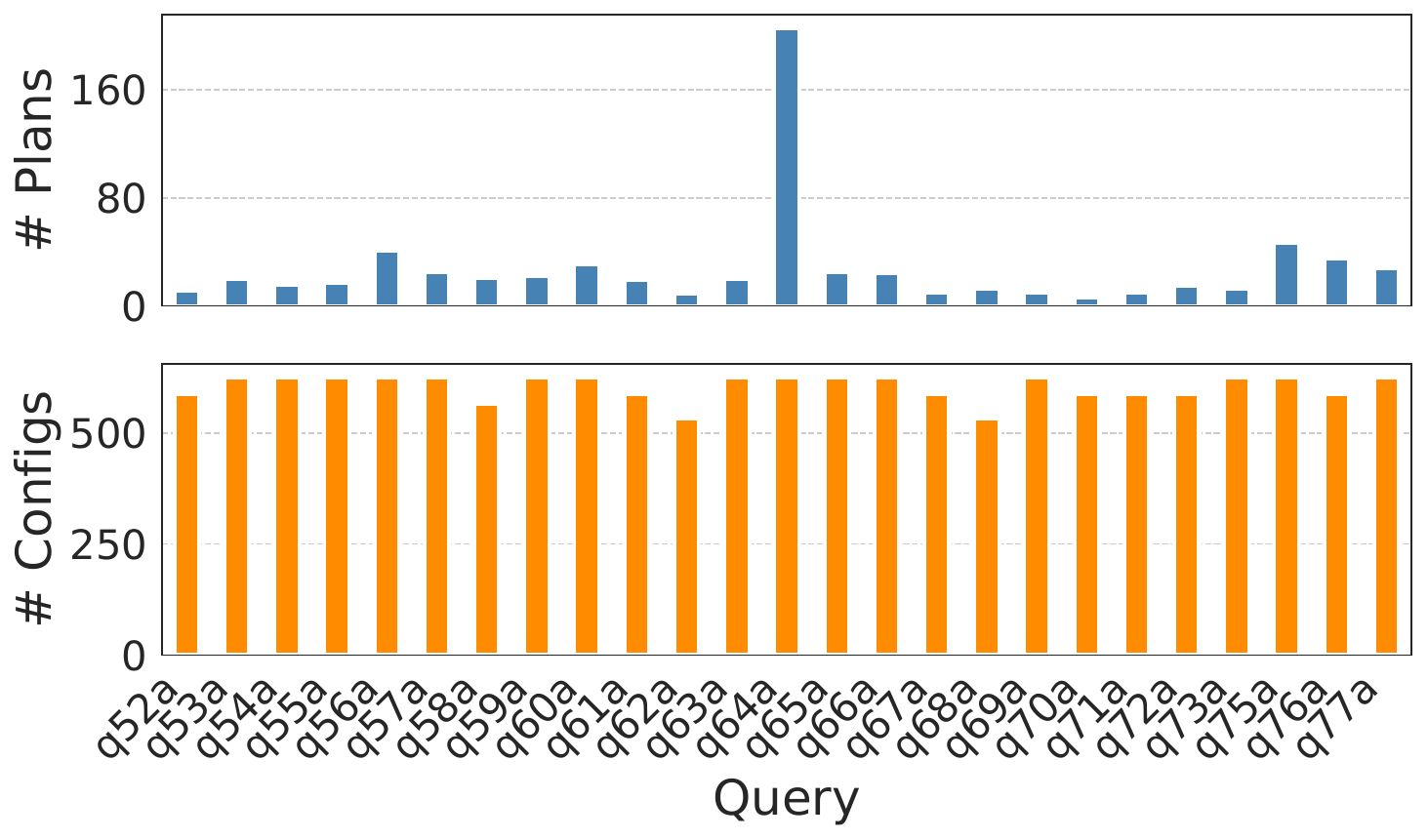}}
\subfigure[\textbf{TPC-DS}, sf = 1, Part IV]{ \label{fig:plan-config-per-query:tpcds1:4}
    \includegraphics[width=0.29\textwidth]{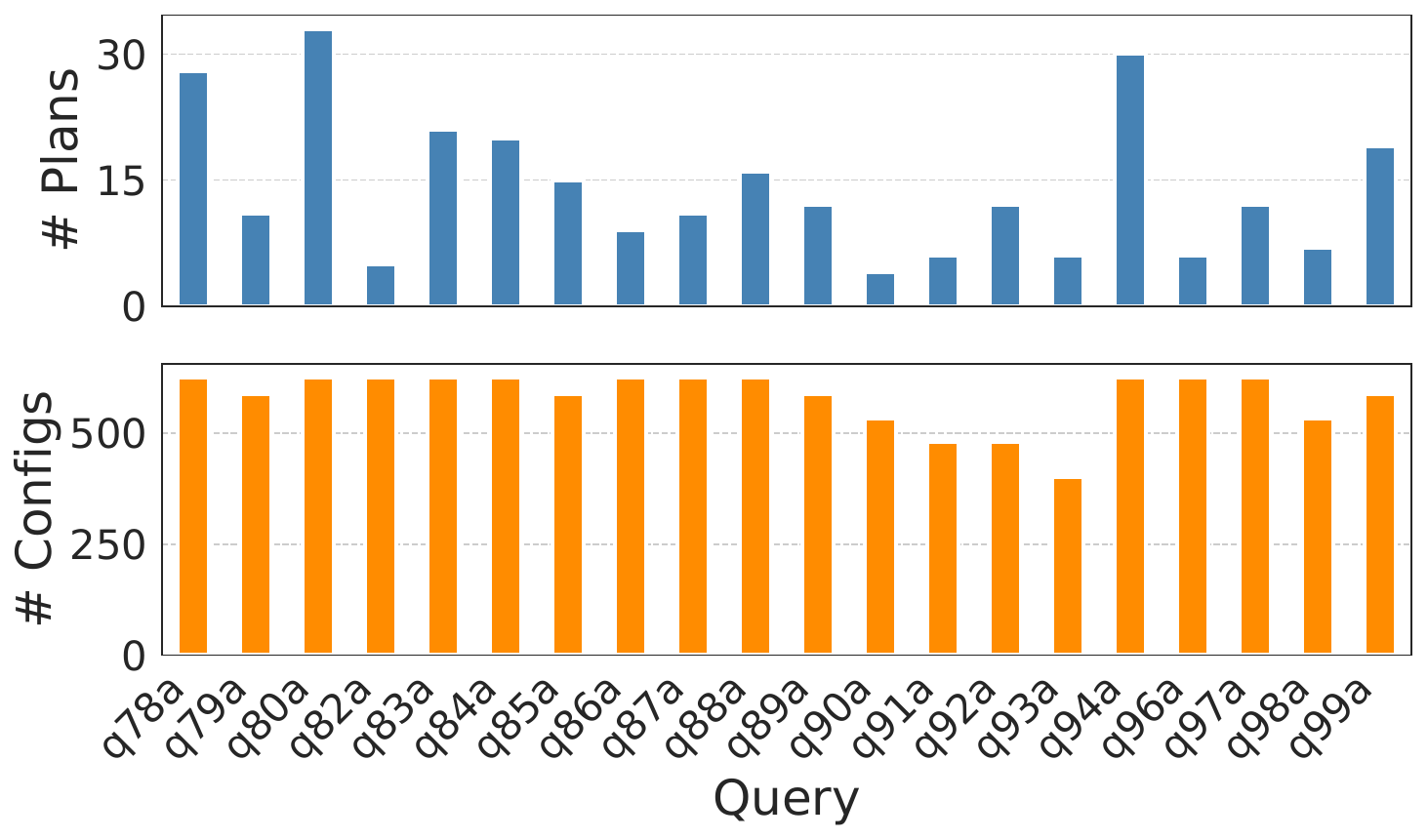}}
\subfigure[\textbf{TPC-DS}, sf = 10, Part II]{ \label{fig:plan-config-per-query:tpcds10:2}
    \includegraphics[width=0.29\textwidth]{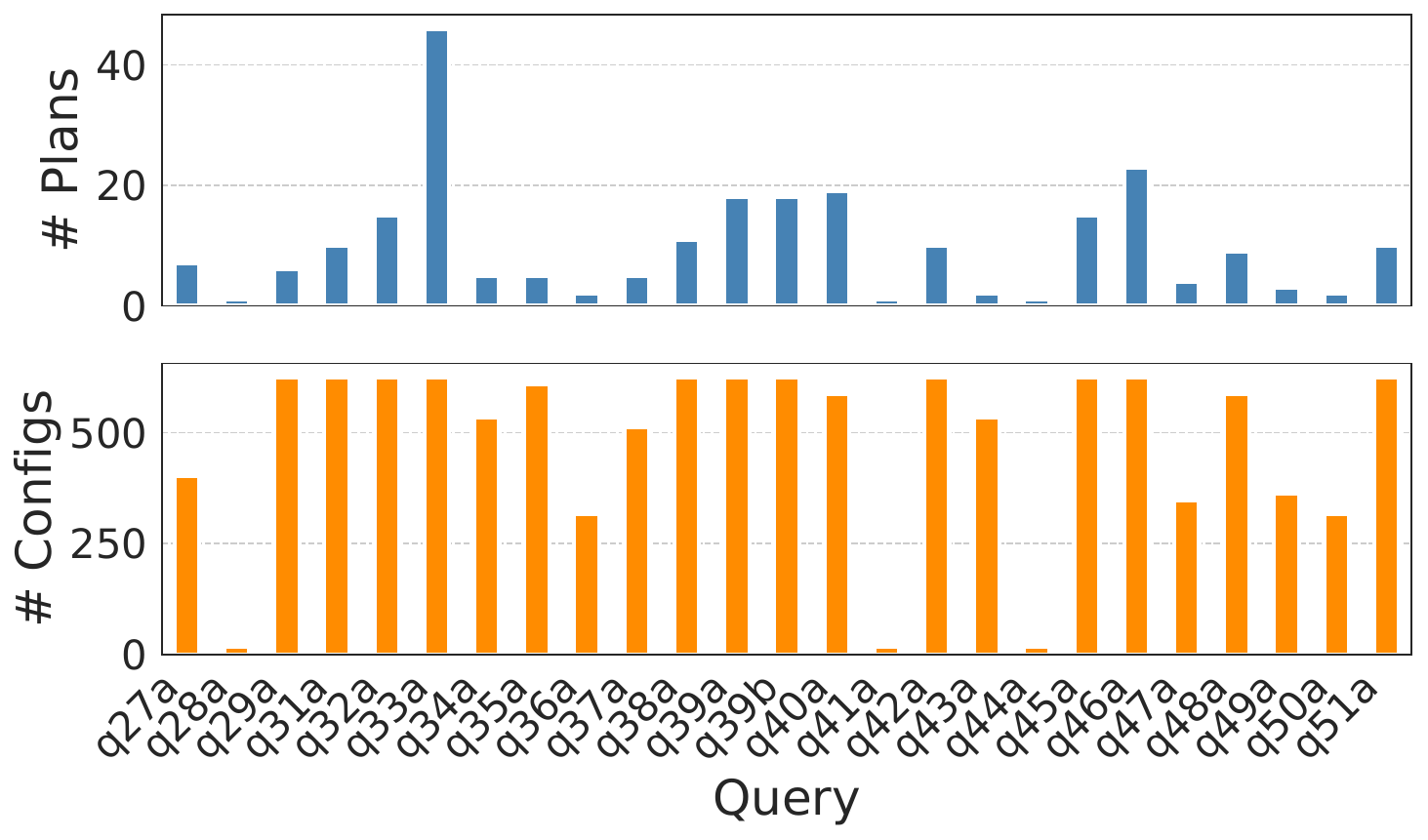}}
\vspace{-1.5em}
\caption{\rthree{(Part I) The number of distinct query plans found and the number of CMP configurations explored (i.e., the number of query optimizer calls) during plan space exploration by \textbf{DOT} for individual queries.}}
\label{fig:plan-config-per-query:all:part-1}
\end{figure*}

\begin{figure*}
\centering
\subfigure[\textbf{TPC-DS}, sf = 10, Part III]{ \label{fig:plan-config-per-query:tpcds10:3}
    \includegraphics[width=0.29\textwidth]{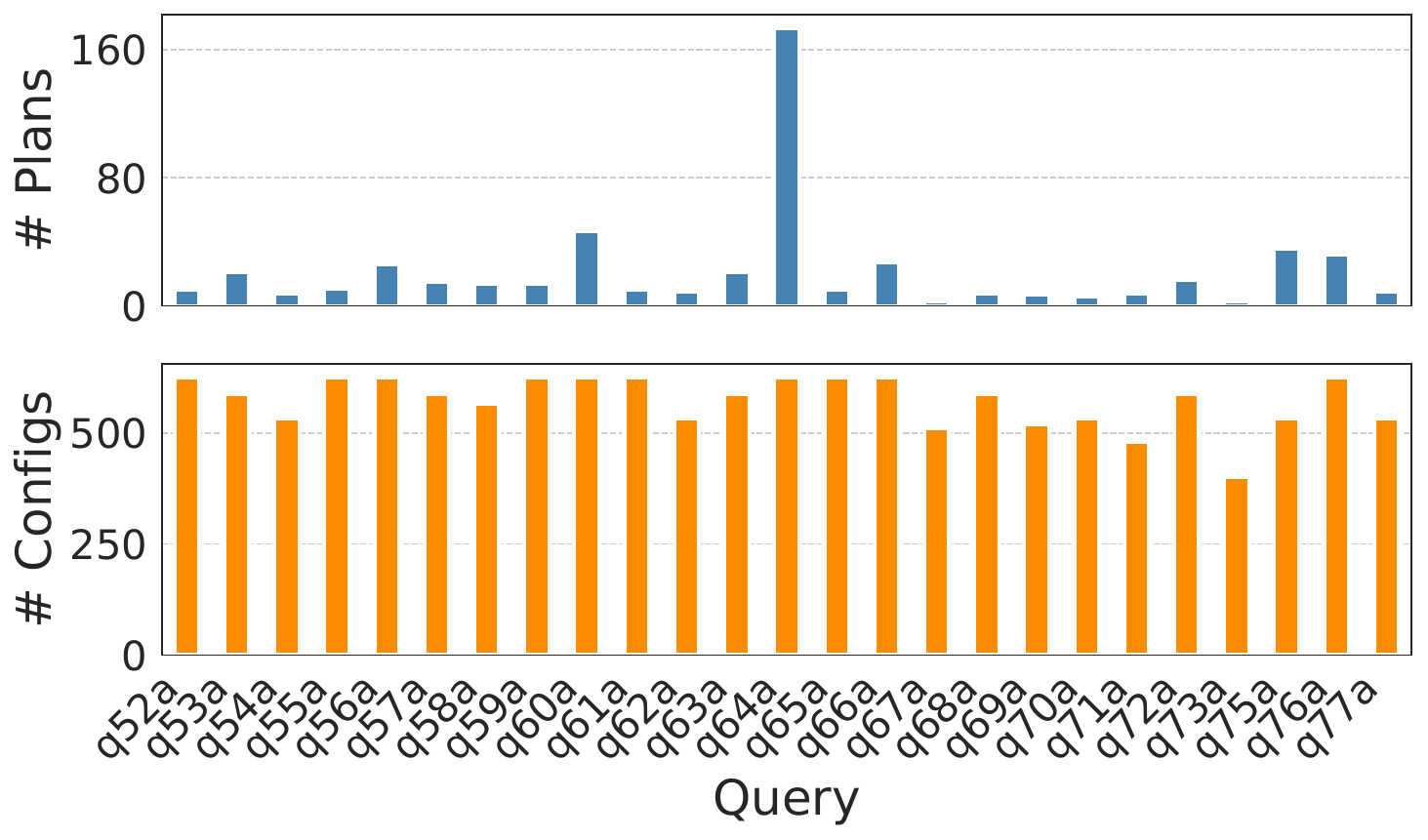}}
\subfigure[\textbf{TPC-DS}, sf = 10, Part IV]{ \label{fig:plan-config-per-query:tpcds10:4}
    \includegraphics[width=0.29\textwidth]{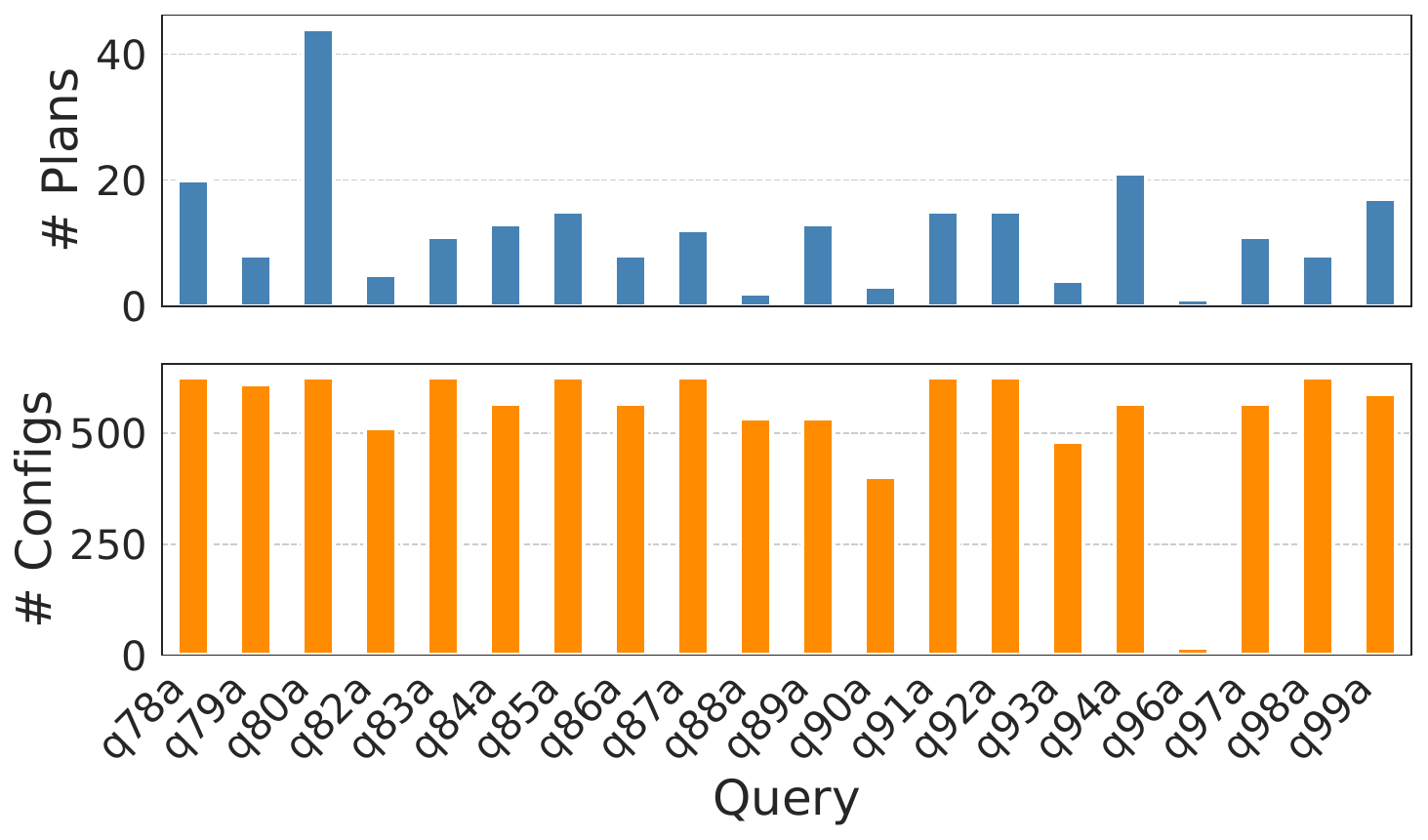}}
\subfigure[\textbf{JOB}, Part II]{ \label{fig:plan-config-per-query:job:2}
    \includegraphics[width=0.29\textwidth]{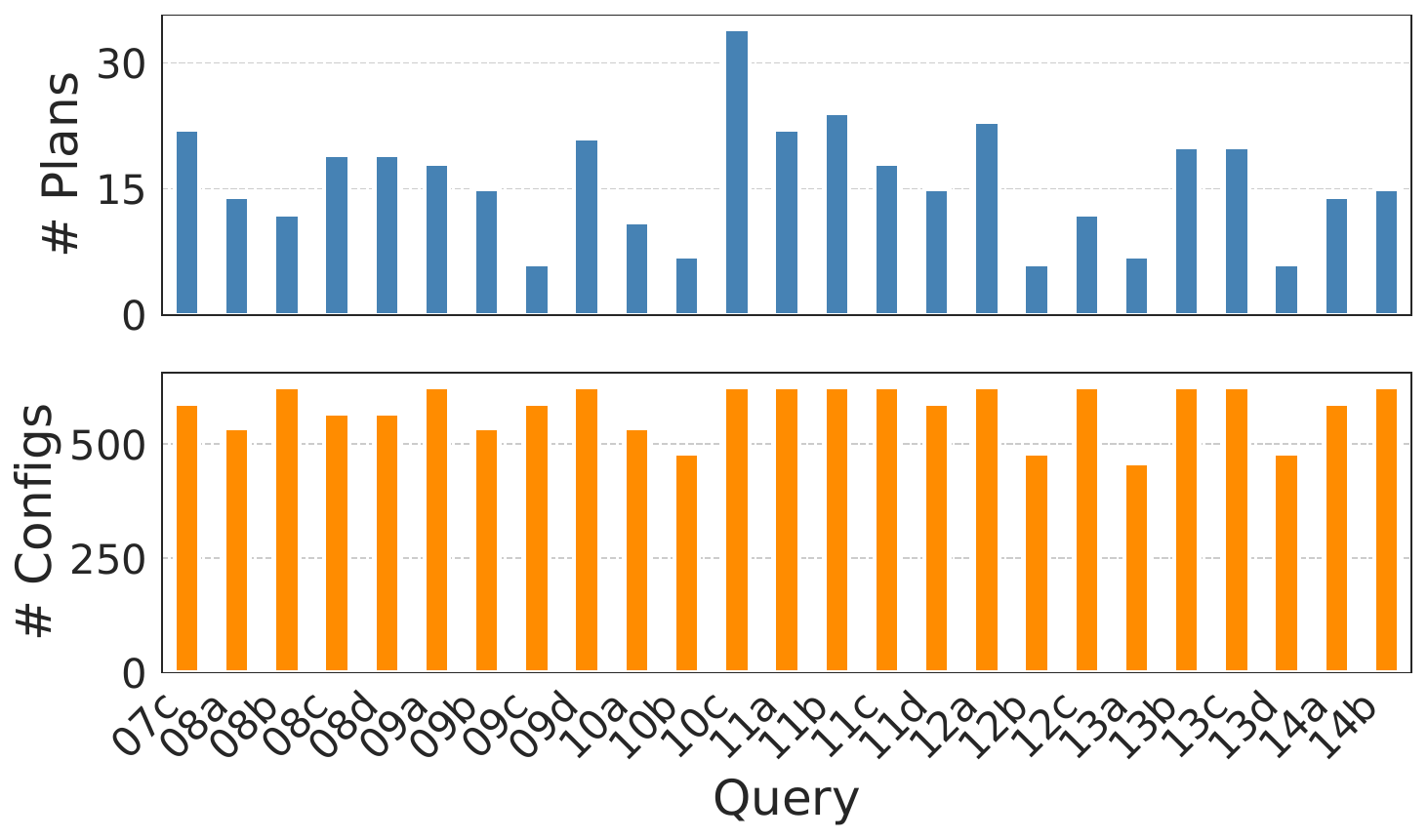}}
\subfigure[\textbf{JOB}, Part III]{ \label{fig:plan-config-per-query:job:3}
    \includegraphics[width=0.29\textwidth]{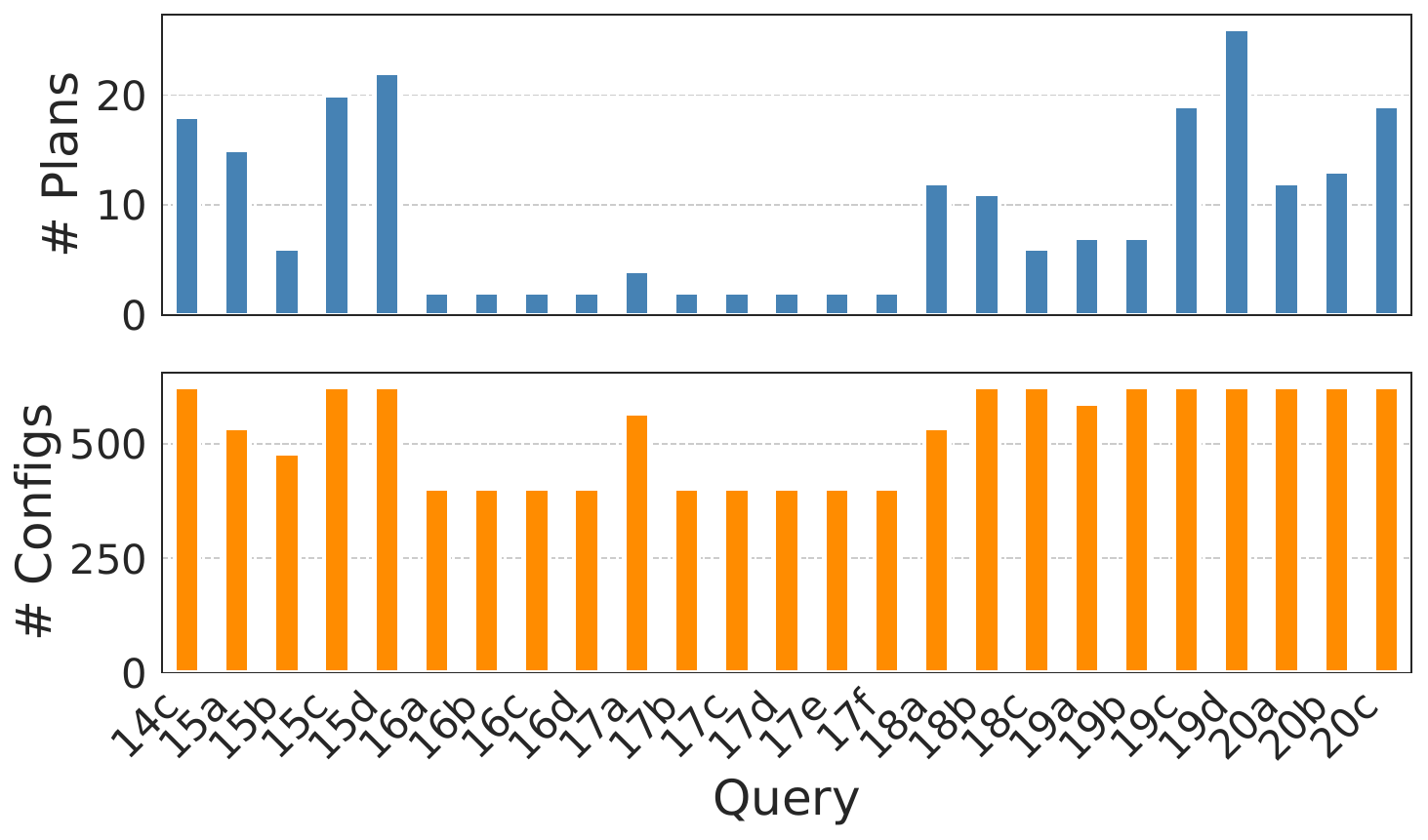}}
\subfigure[\textbf{JOB}, Part IV]{ \label{fig:plan-config-per-query:job:4}
    \includegraphics[width=0.29\textwidth]{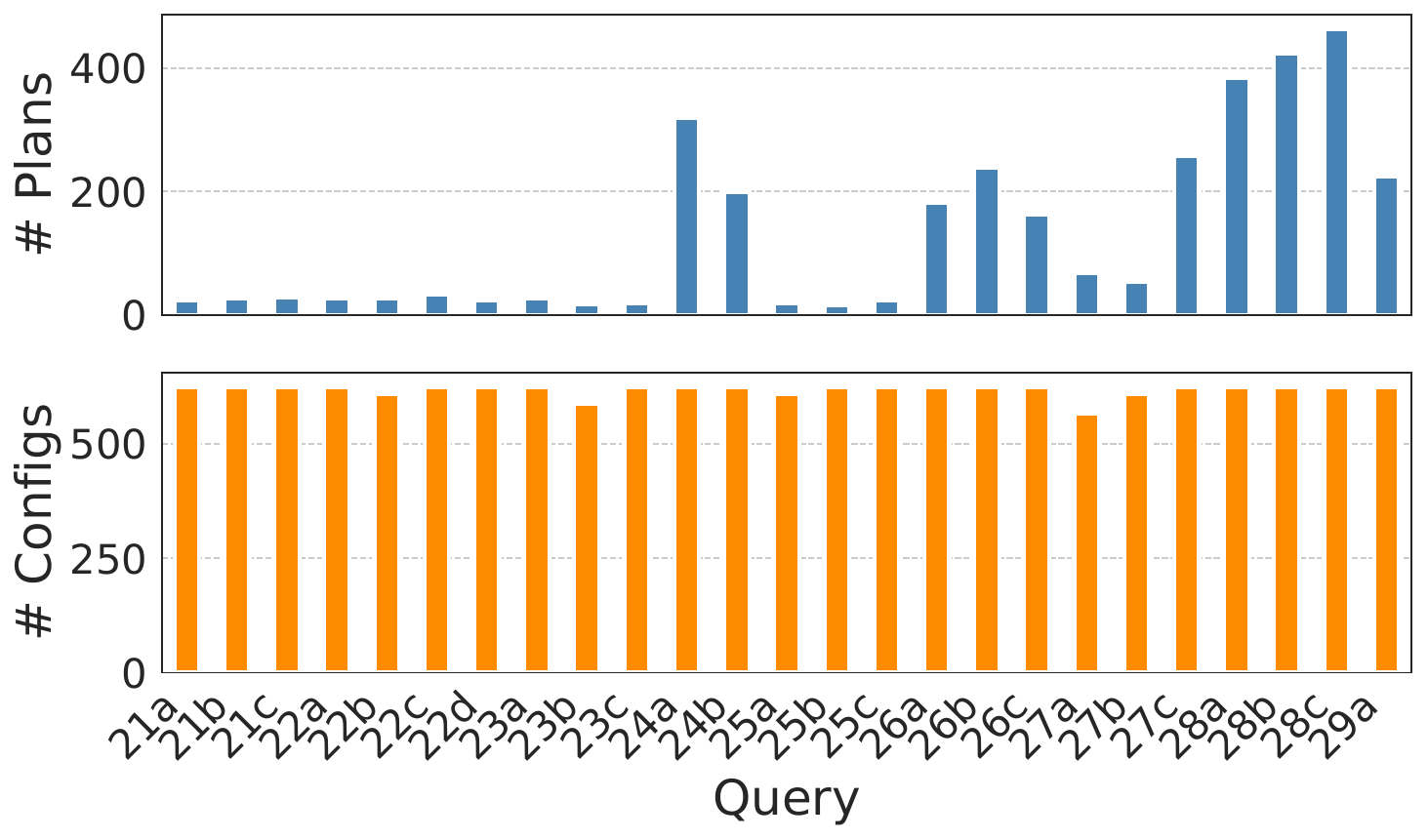}}
\subfigure[\textbf{JOB}, Part V]{ \label{fig:plan-config-per-query:job:5}
    \includegraphics[width=0.29\textwidth]{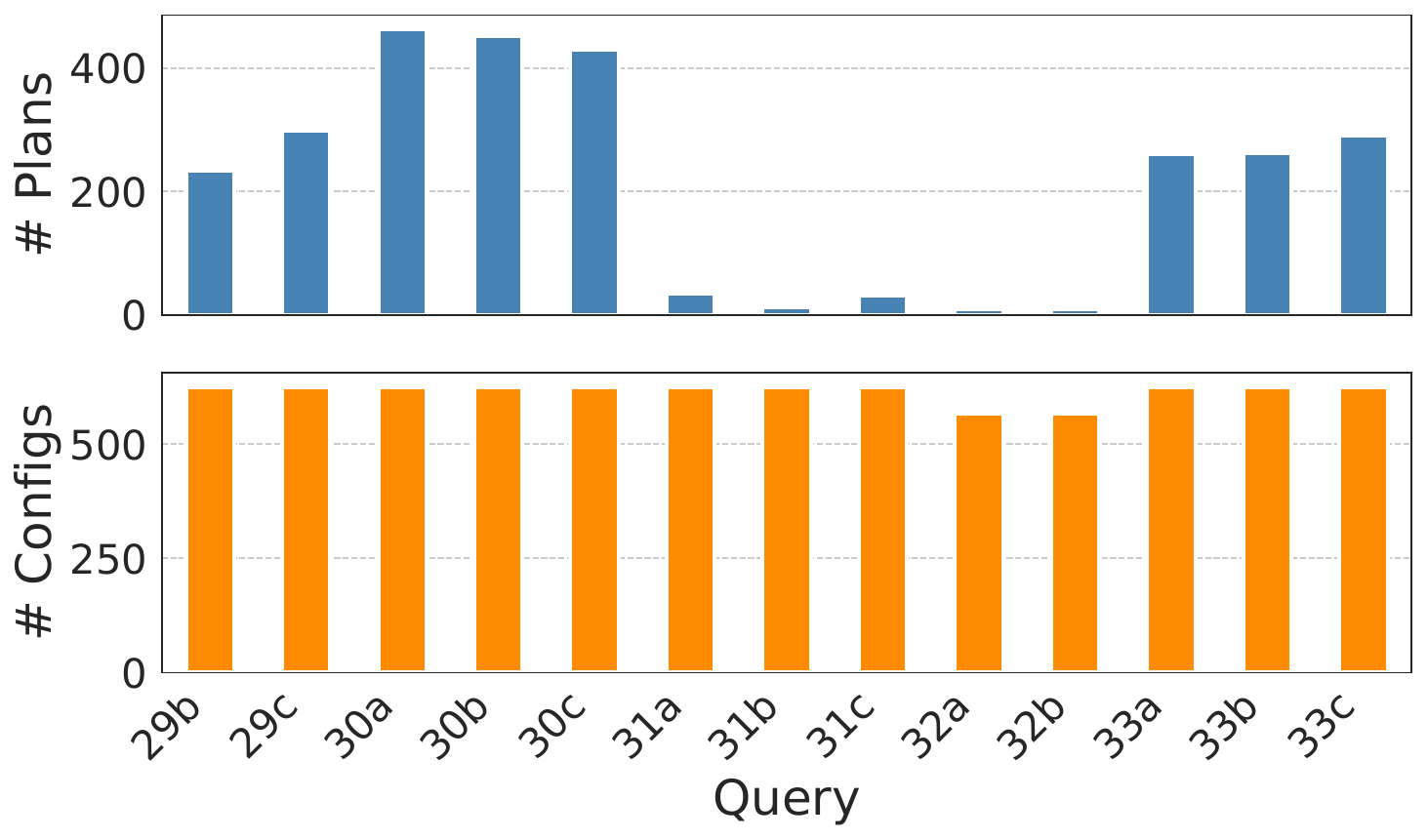}}
\vspace{-1.5em}
\caption{\rthree{(Part II) The number of distinct query plans found and the number of CMP configurations explored (i.e., the number of query optimizer calls) during plan space exploration by \textbf{DOT} for individual queries.}}
\label{fig:plan-config-per-query:all:part-2}
\end{figure*}

\begin{figure*}
\centering
\subfigure[\textbf{STATS}, Part I]{ \label{fig:plan-config-per-query:stats:1}
    \includegraphics[width=0.29\textwidth]{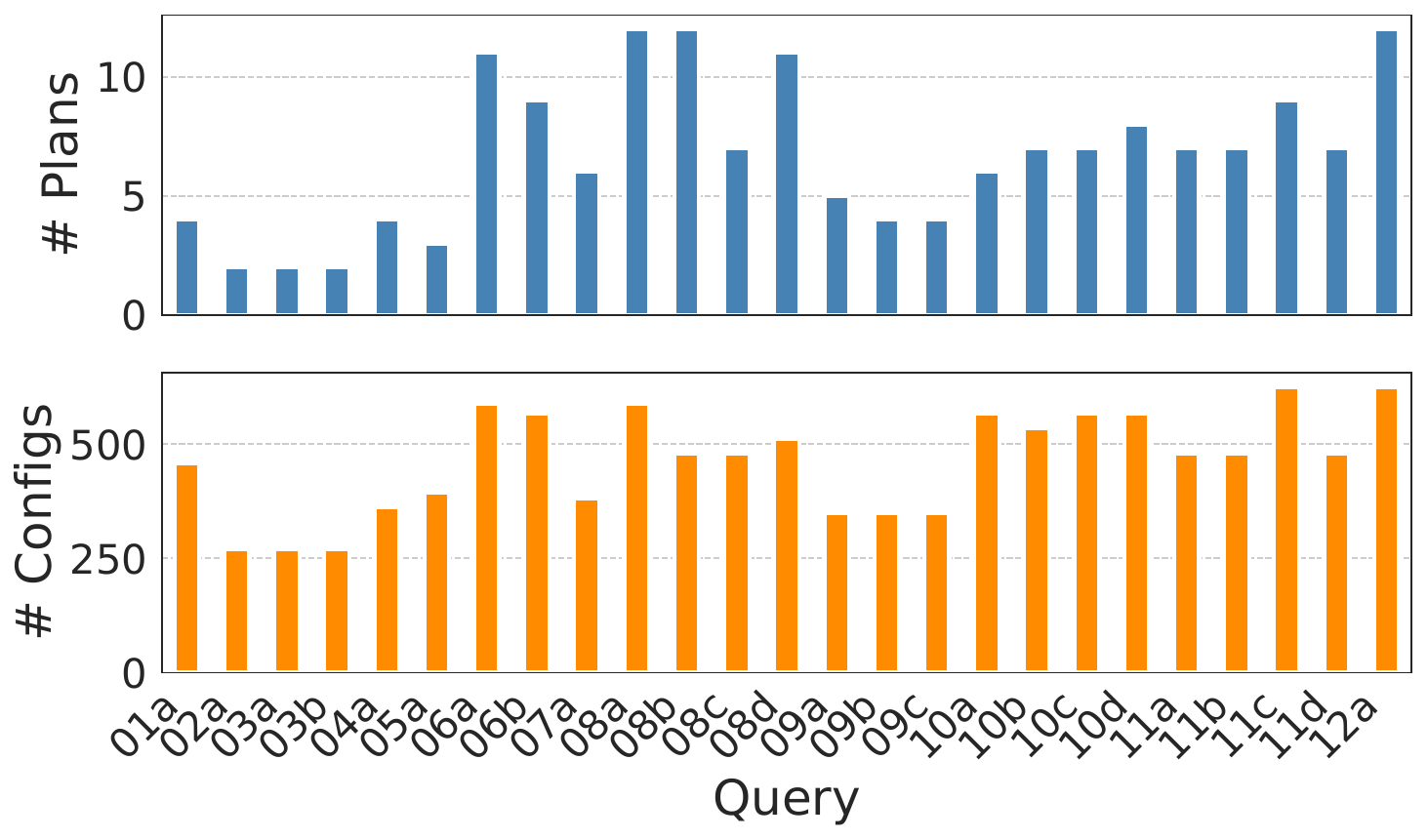}}
\subfigure[\textbf{STATS}, Part II]{ \label{fig:plan-config-per-query:stats:2}
    \includegraphics[width=0.29\textwidth]{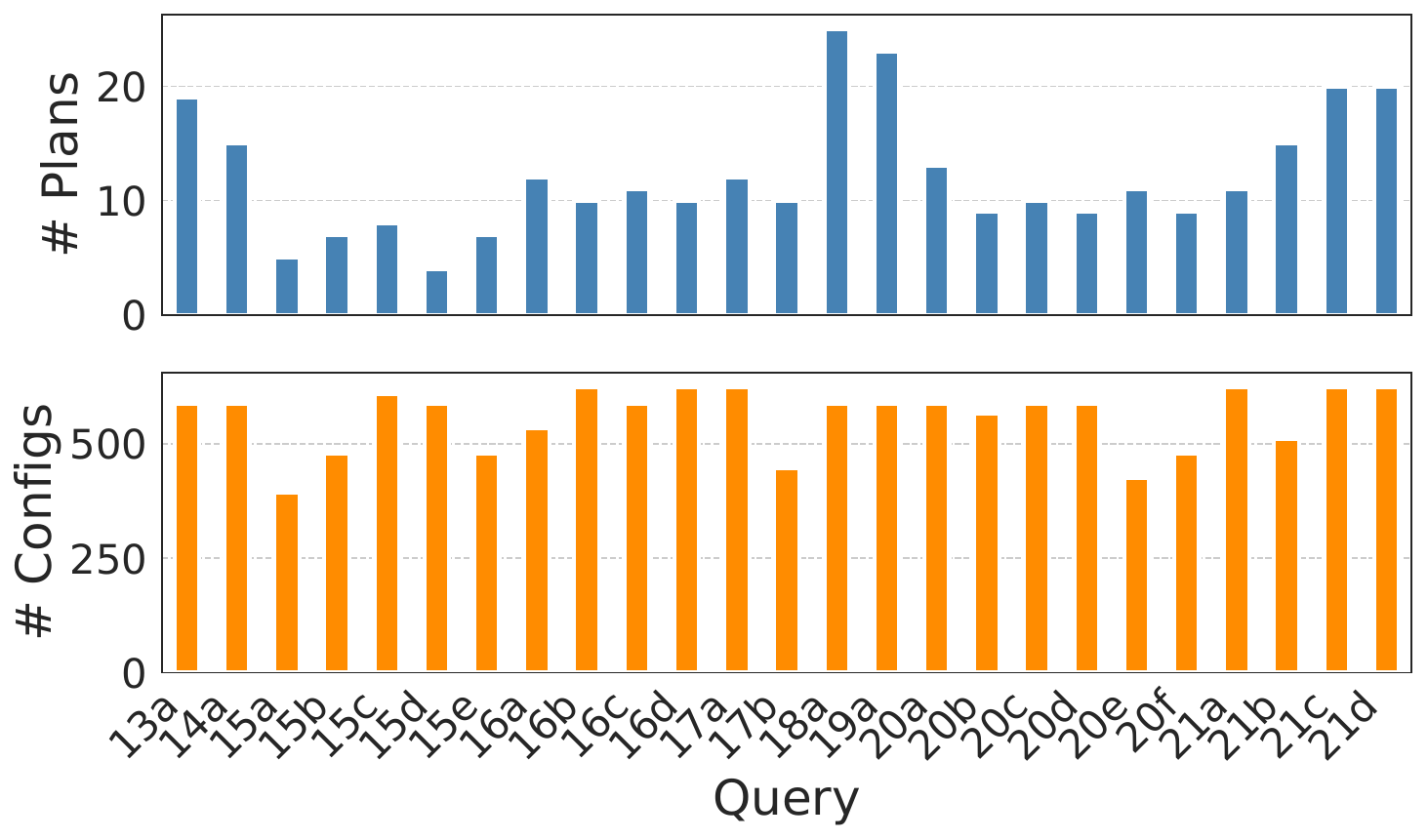}}
\subfigure[\textbf{STATS}, Part III]{ \label{fig:plan-config-per-query:stats:3}
    \includegraphics[width=0.29\textwidth]{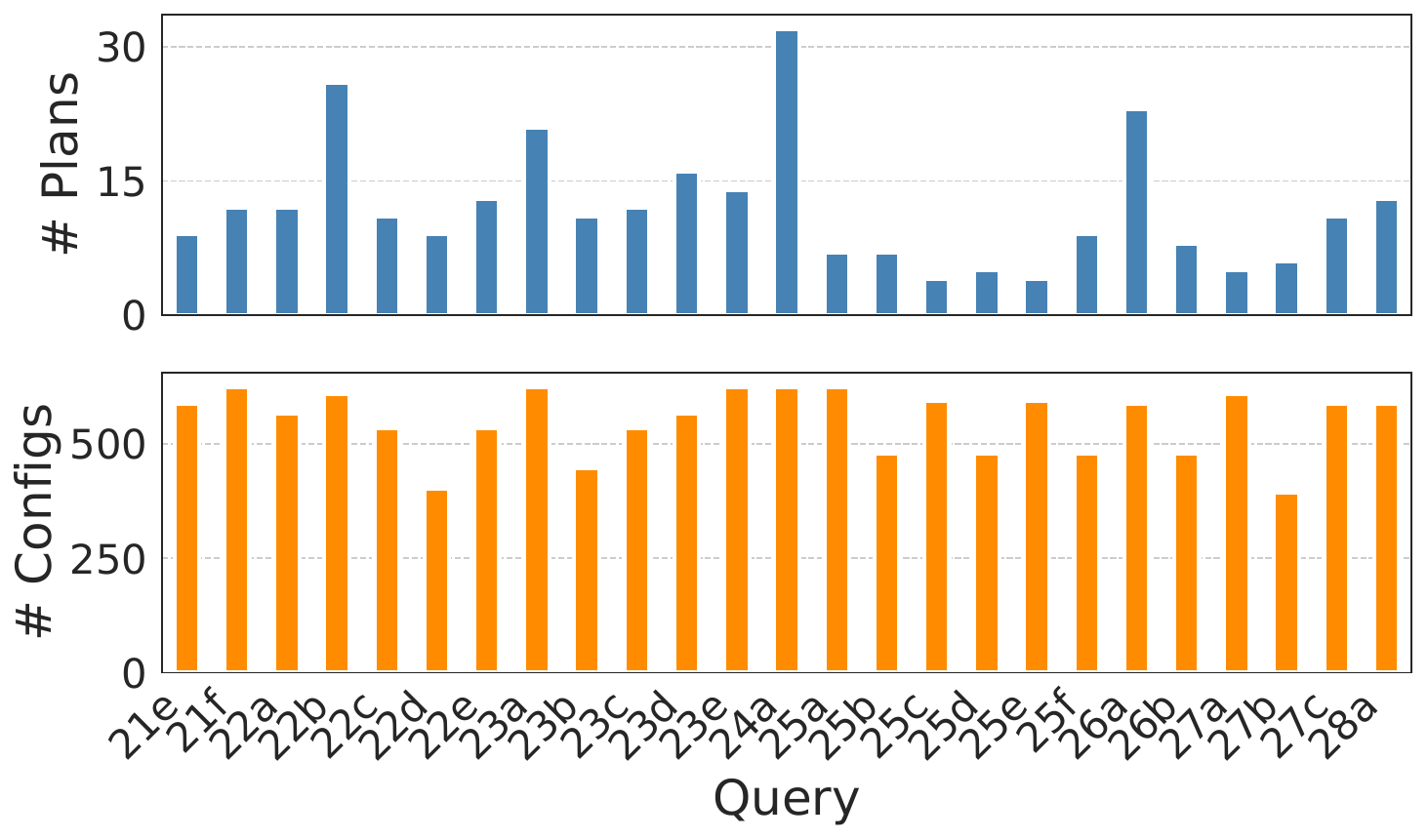}}
\subfigure[\textbf{STATS}, Part IV]{ \label{fig:plan-config-per-query:stats:4}
    \includegraphics[width=0.29\textwidth]{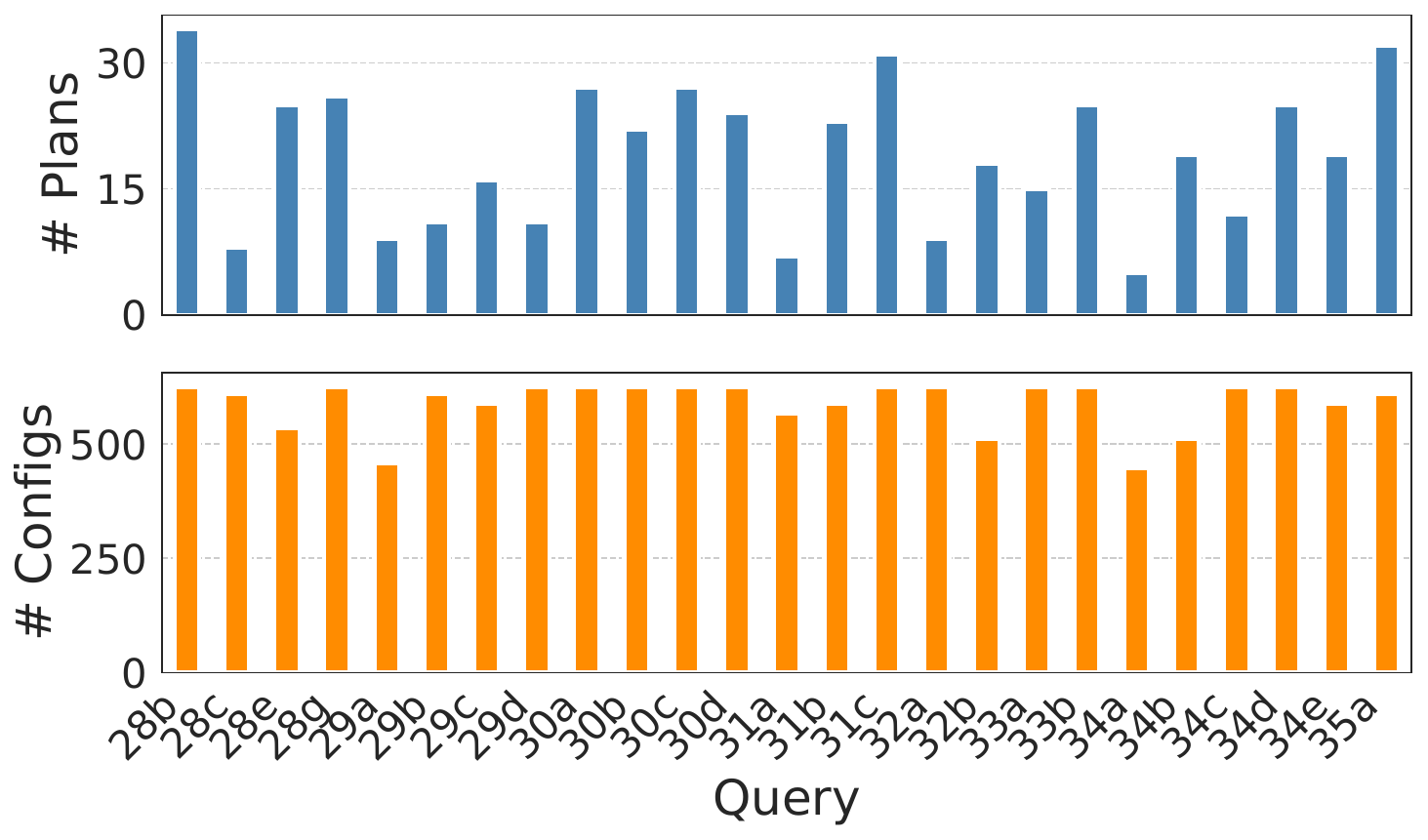}}
\subfigure[\textbf{STATS}, Part V]{ \label{fig:plan-config-per-query:stats:5}
   \includegraphics[width=0.29\textwidth]{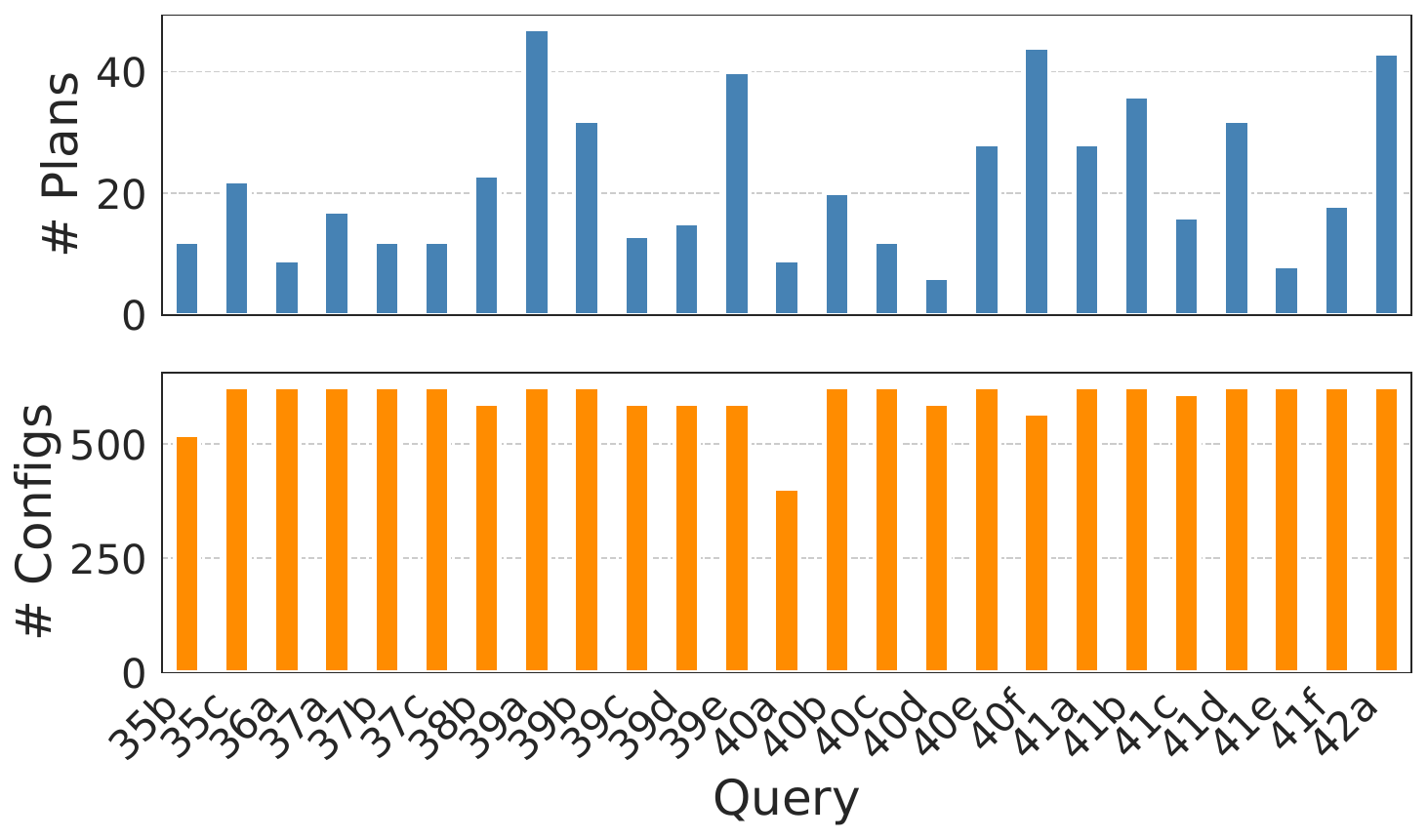}}
\subfigure[\textbf{STATS}, Part VI]{ \label{fig:plan-config-per-query:stats:6}
   \includegraphics[width=0.29\textwidth]{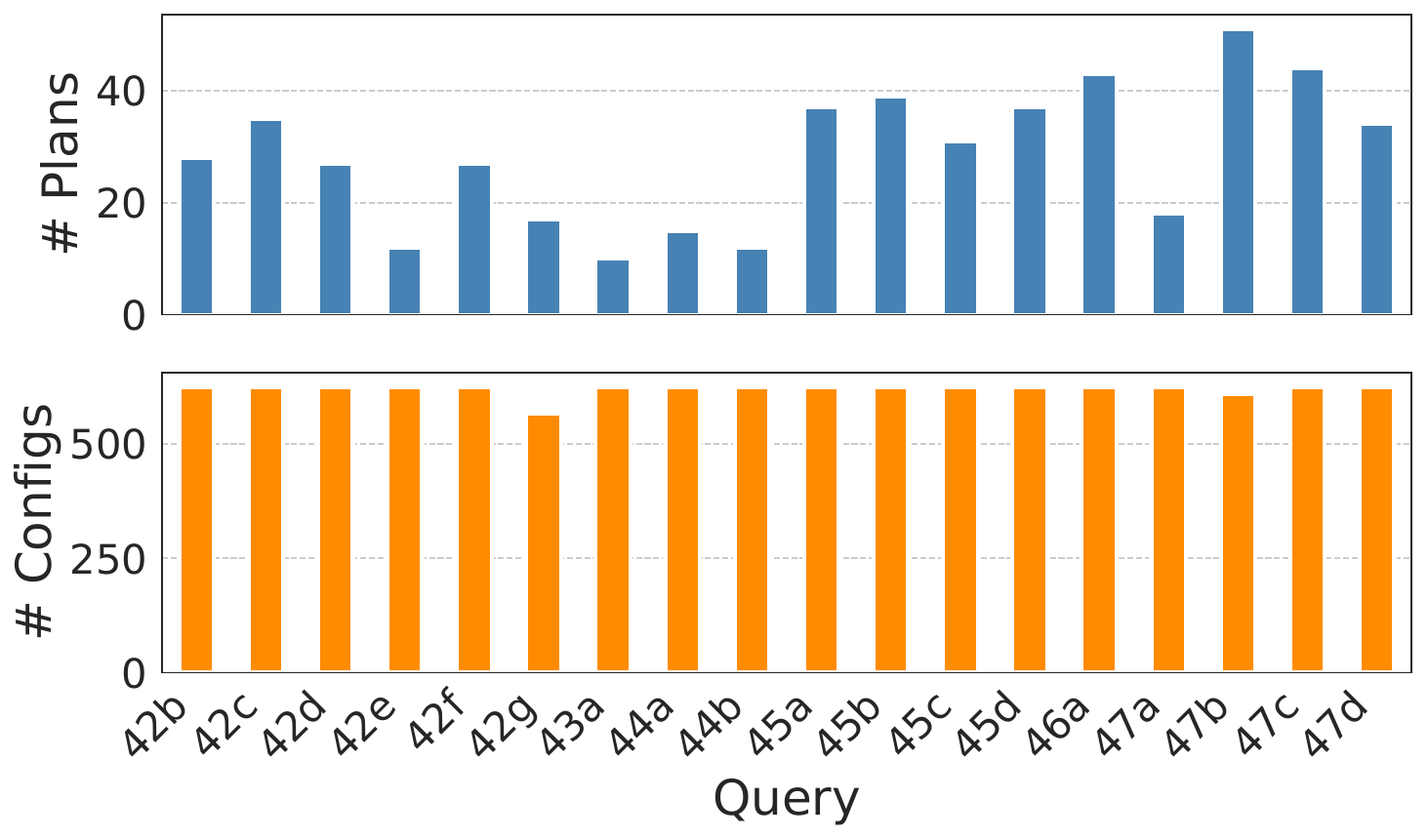}}
\vspace{-1.5em}
\caption{\rthree{(Part III) The number of distinct query plans found and the number of CMP configurations explored (i.e., the number of query optimizer calls) during plan space exploration by \textbf{DOT} for individual queries.}}
\label{fig:plan-config-per-query:all:part-3}
\vspace{-1em}
\end{figure*}

\rthree{Figures~\ref{fig:space:plans:rest} and~\ref{fig:space:configs:rest} present distributions of the number of distinct plans found and the number of configurations explored for the remaining workloads. Figures~\ref{fig:plan-config-per-query:all:part-1},~\ref{fig:plan-config-per-query:all:part-2}, and~\ref{fig:plan-config-per-query:all:part-3} further present the number of distinct plans found and the number of configurations explored for individual queries of those workloads.}

\vspace{-0.5em}
\subsection{\rthree{Analysis of CMP Diversity in Section~\ref{sec:eval:exploration:cmp-diversity}}}
\label{appendix:results:analysis-cmpt-diversity}

\begin{figure*}
\centering
\subfigure[\textbf{TPC-H}, sf = 1]{ \label{fig:cmp-diversity:tpch1}
    \includegraphics[width=0.28\textwidth]{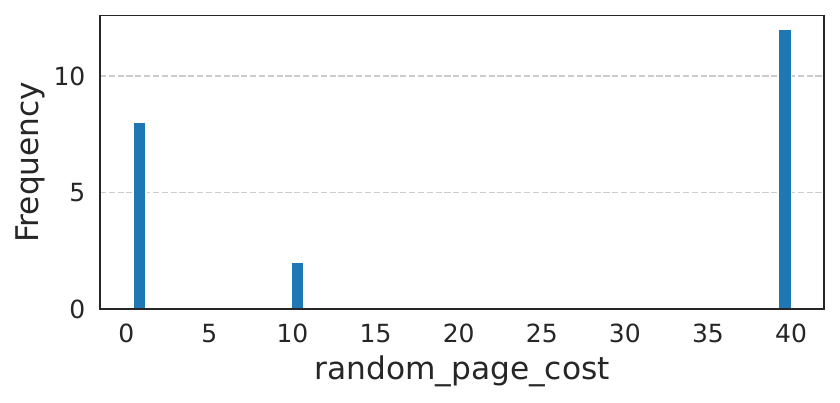}}
\subfigure[\textbf{TPC-DS}, sf = 1]{ \label{fig:cmp-diversity:tpcds1}
    \includegraphics[width=0.28\textwidth]{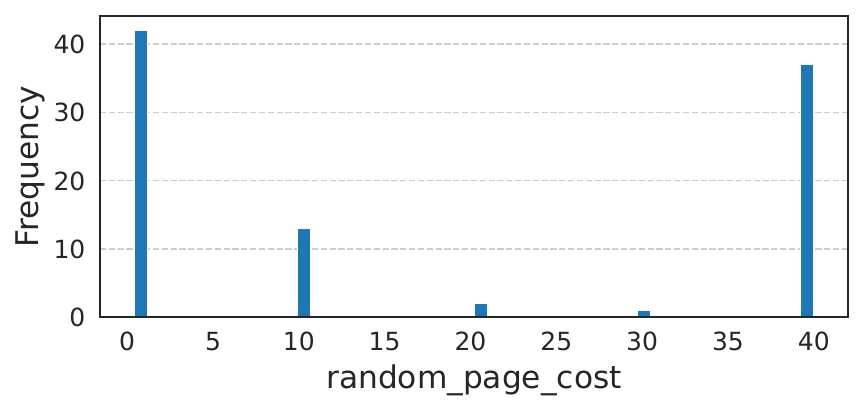}}
\subfigure[\textbf{STATS}]{ \label{fig:cmp-diversity:stats}
    \includegraphics[width=0.28\textwidth]{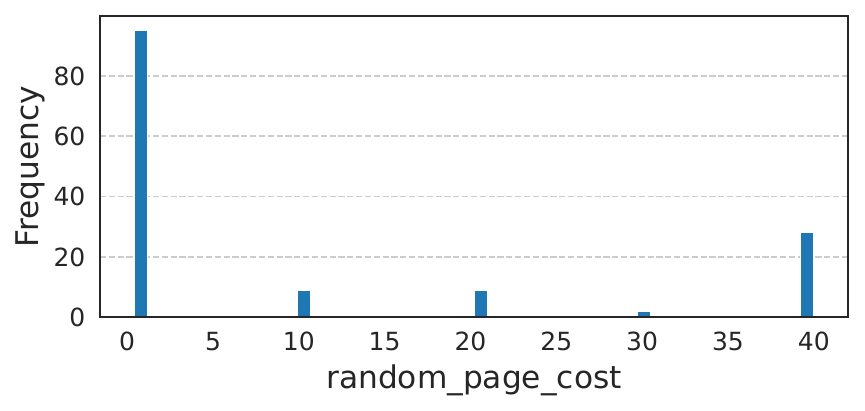}}
\vspace{-1.5em}
\caption{\rtwo{Distribution of values of the CMP \texttt{random\_page\_cost} used by the best query plans found by \textbf{DOT}.}}
\label{fig:cmp-diversity:rest}
\vspace{-1em}
\end{figure*}

\rthree{Figure~\ref{fig:cmp-diversity:rest} presents the distribution of \texttt{random\_page\_cost} values used by the best query plans found by \textbf{DOT} for the other workloads.}

\subsection{Queries in the Case Study of Section~\ref{sec:eval:exploration:case-study}}
\label{appendix:results:case-study-queries}

\subsubsection{\textbf{TPC-H} (sf = 10) Query 11}
\begin{verbatim}
select
    ps_partkey,
    sum(ps_supplycost * ps_availqty) as value
from
    partsupp,
    supplier,
    nation
where
    ps_suppkey = s_suppkey
    and s_nationkey = n_nationkey
    and n_name = 'GERMANY'
group by
    ps_partkey 
having sum(ps_supplycost * ps_availqty) > (
			select
				sum(ps_supplycost * ps_availqty) * 0.0000100000
			from
				partsupp,
				supplier,
				nation
			where
				ps_suppkey = s_suppkey
				and s_nationkey = n_nationkey
				and n_name = 'GERMANY'
		)
order by
    value desc
\end{verbatim}

\subsubsection{\textbf{Stats} Query 47b}

\begin{verbatim}
SELECT COUNT(*) 
FROM 
    comments as c, 
    posts as p, 
    postLinks as pl, 
    postHistory as ph, 
    votes as v, 
    badges as b, 
    users as u 
WHERE 
    p.Id = pl.RelatedPostId 
    AND b.UserId = u.Id 
    AND c.UserId = u.Id 
    AND p.Id = v.PostId 
    AND p.Id = c.PostId 
    AND p.Id = ph.PostId 
    AND p.Score <= 40 
    AND p.CommentCount >= 0 
    AND p.CreationDate >=
        '2010-07-28 17:40:56'::timestamp 
    AND p.CreationDate <=
        '2014-09-11 04:22:44'::timestamp 
    AND pl.LinkTypeId = 1;
\end{verbatim}

\subsection{Reduction of Tuning Time in Section~\ref{sec:eval:tuning-time}}
\label{appendix:results:reduction-tuning-time}

\begin{figure}
\centering
\subfigure[\textbf{TPC-H}, sf = 10]{ \label{fig:tune:reduction:tpch10}
    \includegraphics[width=0.48\columnwidth]{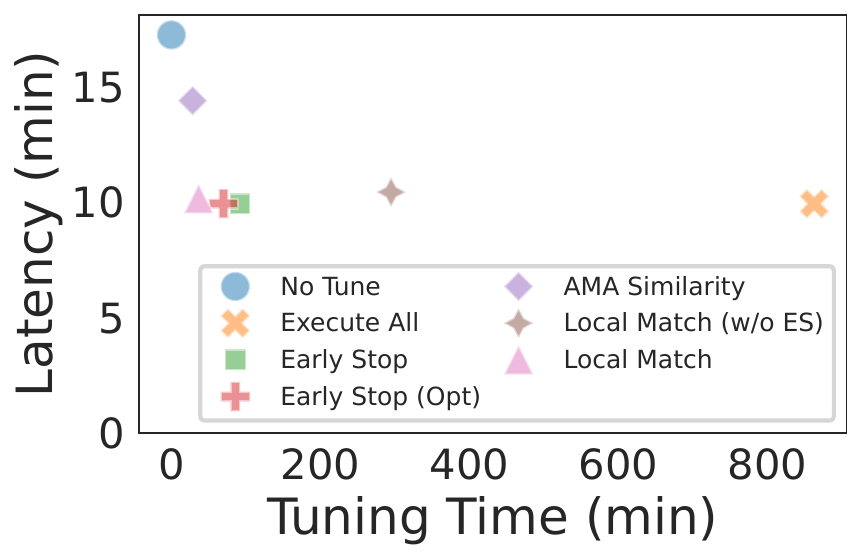}}
\subfigure[\textbf{TPC-DS}, sf = 10]{ \label{fig:tune:reduction:tpcds10}
    \includegraphics[width=0.48\columnwidth]{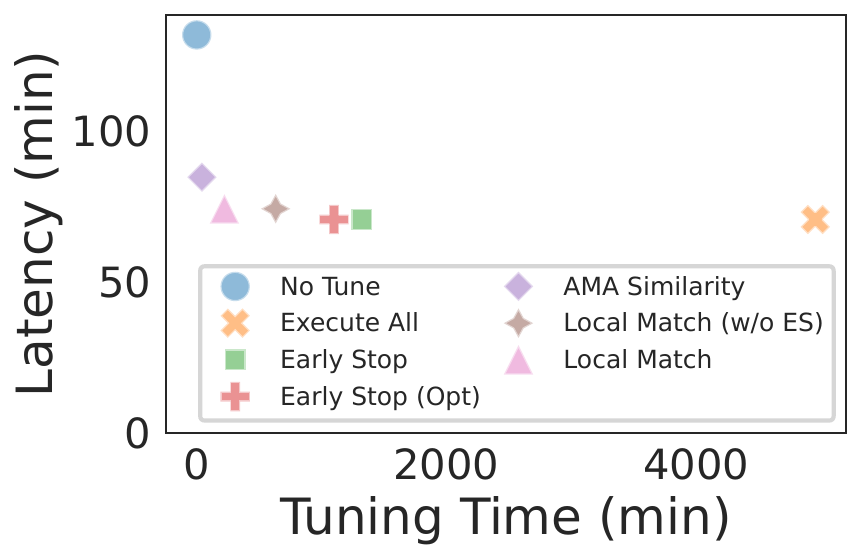}}
\vspace{-1.5em}
\caption{\shepherd{Comparison of variant approaches for reducing tuning time as well as baseline approaches, using TPC-H and TPC-DS with scaling factor 10.}}
\vspace{0.5em}
\label{fig:tune:reduction:more}
\end{figure}

\subsubsection{Comparison of the Variants}
Figure~\ref{fig:tune:reduction:more} presents the results of the variant approaches for reducing tuning time, as well as the baseline approaches, that have been detailed in Section~\ref{sec:eval:tuning-time}, by using \textbf{TPC-H} and \textbf{TPC-DS} with scaling factor 10.
The observations are similar to those from \textbf{TPC-H} and \textbf{TPC-DS} with scaling factor 1.

\vspace{-0.5em}
\subsubsection{Robustness of Similarity-based Mechanisms}

\begin{figure}
\centering
\subfigure[\textbf{TPC-H}, sf = 1]{ \label{fig:tune:sensitive:tpch1}
    \includegraphics[width=0.48\columnwidth]{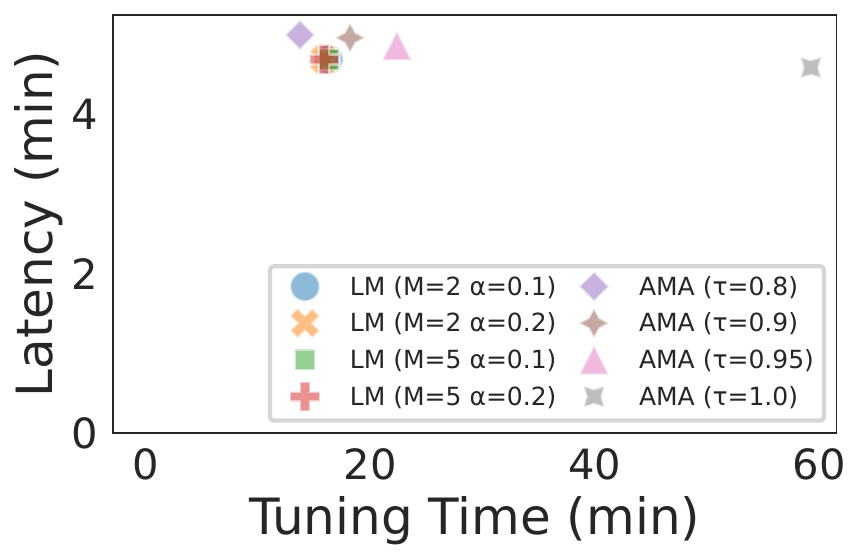}}
\subfigure[\textbf{TPC-H}, sf = 10]{ \label{fig:tune:sensitive:tpch10}
    \includegraphics[width=0.48\columnwidth]{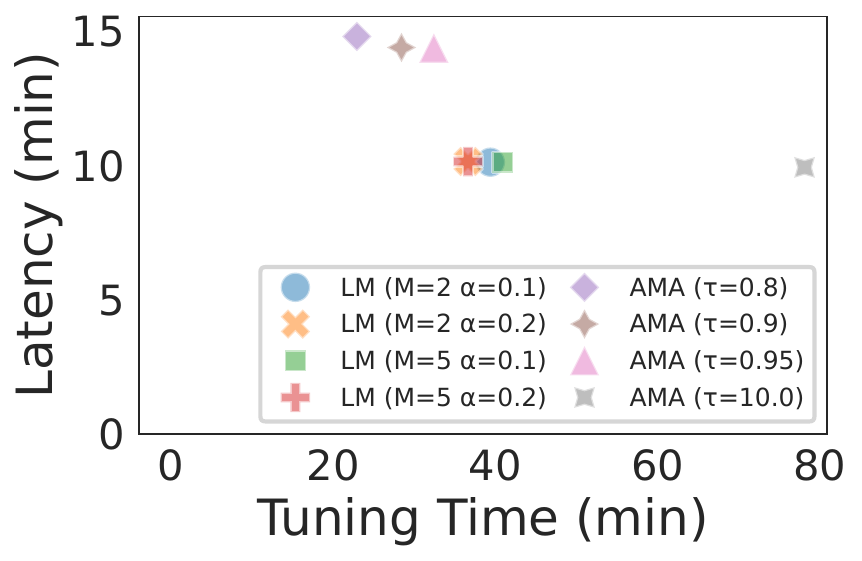}}
\subfigure[\textbf{TPC-DS}, sf = 1]{ \label{fig:tune:sensitive:tpcds1}
    \includegraphics[width=0.48\columnwidth]{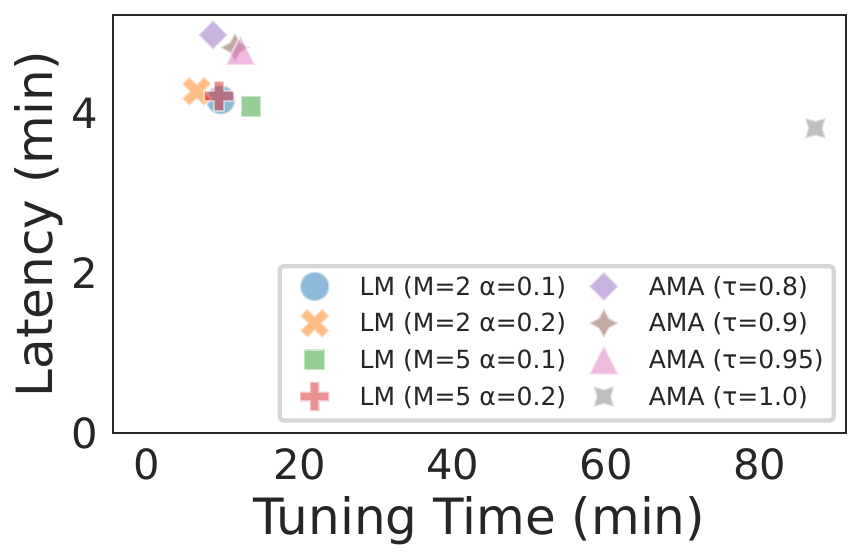}}
\subfigure[\textbf{TPC-DS}, sf = 10]{ \label{fig:tune:sensitive:tpcds10}
    \includegraphics[width=0.48\columnwidth]{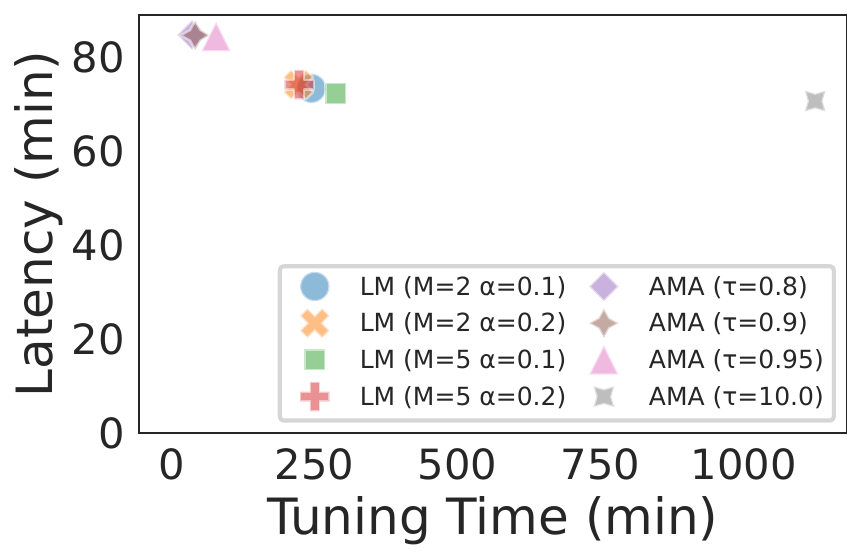}}
\subfigure[\textbf{JOB}]{ \label{fig:tune:sensitive:job}
    \includegraphics[width=0.48\columnwidth]{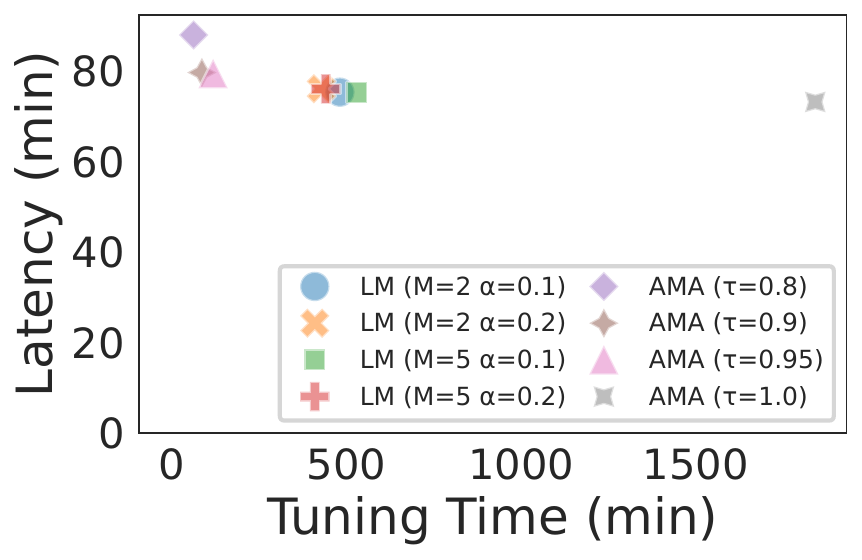}}
\subfigure[\textbf{STATS}]{ \label{fig:tune:sensitive:stats}
    \includegraphics[width=0.48\columnwidth]{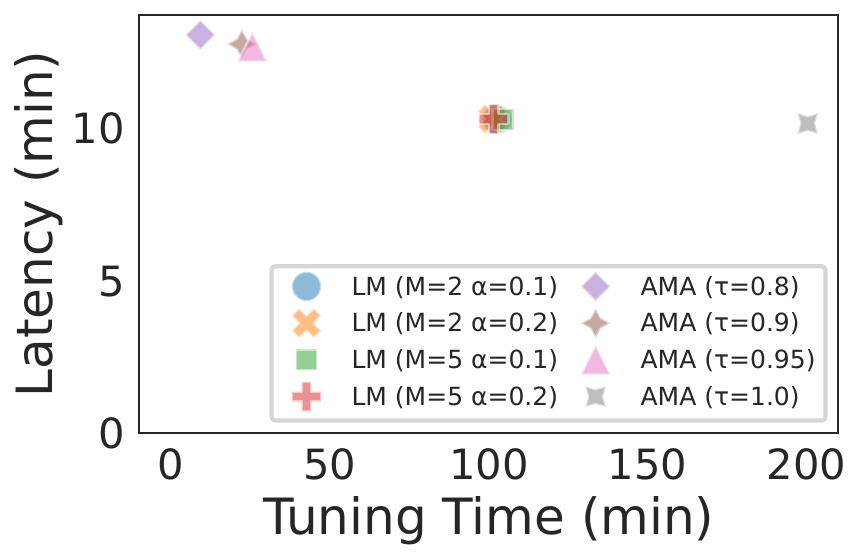}}
\vspace{-1.5em}
\caption{\shepherd{Robustness of the similarity-based plan-filtering mechanisms w.r.t. the input threshold parameters. ``AMA'' means ``\emph{AMA Similarity}'' and ``LM'' means ``\emph{Local Match}.''}}
\label{fig:tune:sensitive}
\end{figure}

Both \emph{AMA Similarity} and \emph{Local Match} require input threshold parameters, whose settings can affect their efficacy.
To study the robustness of both approaches, we further vary their input parameters and observe their sensitivity.
Specifically, we vary the plan similarity threshold $\tau$ (Algorithm~\ref{alg:optimization:embedding-sim}) used by \emph{AMA Similarity} as $\tau\in\{0.8, 0.9, 0.95, 1\}$.
We vary the two input parameters, $\alpha$ to determine bottleneck operators and $M$ to select the top bottleneck operators (Algorithm~\ref{alg:optimization:bottleneck-sim}), used by \emph{Local Match} as $\alpha\in\{0.1, 0.2\}$ and $M\in\{2,5\}$.
Figure~\ref{fig:tune:sensitive} presents the evaluation results.
We find that \emph{AMA Similarity} is more sensitive to its input parameters compared to \emph{Local Match} in terms of workload execution time.
For example, the workload execution time increases dramatically for both \textbf{JOB} when decreasing the threshold $\tau$ employed by \emph{AMA Similarity} from 0.9 to 0.8, as shown by Figure~\ref{fig:tune:sensitive:job}.
The workload execution time by using \emph{Local Match} remains stable for all workloads tested when varying its input parameters $\alpha$ and $M$.
On the other hand, both \emph{AMA Similarity} and \emph{Local Match} are insensitive to its input parameters in terms of tuning time, with the exception of setting the similarity threshold $\tau=1$ for \emph{AMA Similarity} which represents a degenerate case where no candidate plans can be filtered out.



\nobalance

\subsection{Extended PG Evaluation in Section~\ref{sec:eval:pg-extended}}
\label{appendix:results:pg-eval-extended}

\subsubsection{\rtwo{Detailed Results on Tuned Database}}

\begin{figure}[t]
\centering
\subfigure[\textbf{TPC-H}, sf = 10]{ \label{fig:improve:tuned:singlequery:tpch10}
    \includegraphics[width=1.0\columnwidth]{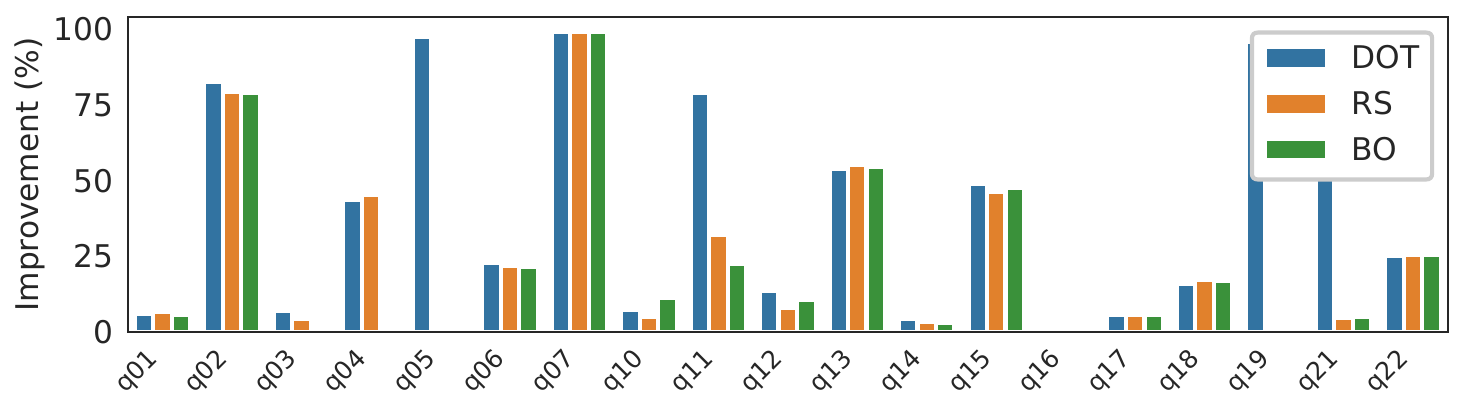}}
\subfigure[\textbf{TPC-DS}, sf = 10]{ \label{fig:improve:tuned:singlequery:tpcds10}
    \includegraphics[width=1.0\columnwidth]{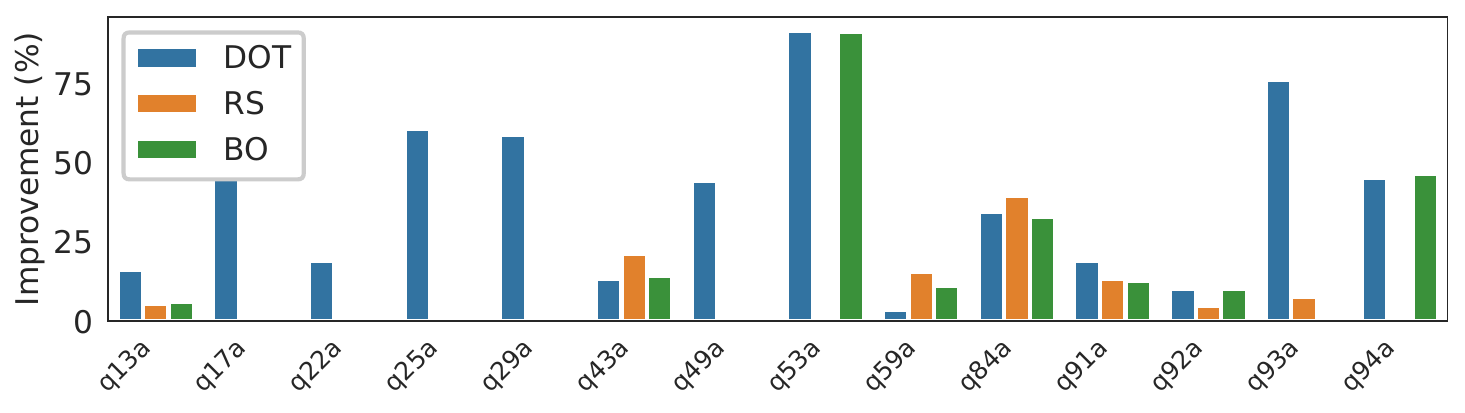}}
\subfigure[\textbf{STATS}, Part I]{ \label{fig:improve:tuned:singlequery:stats:part-1}
    \includegraphics[width=1.0\columnwidth]{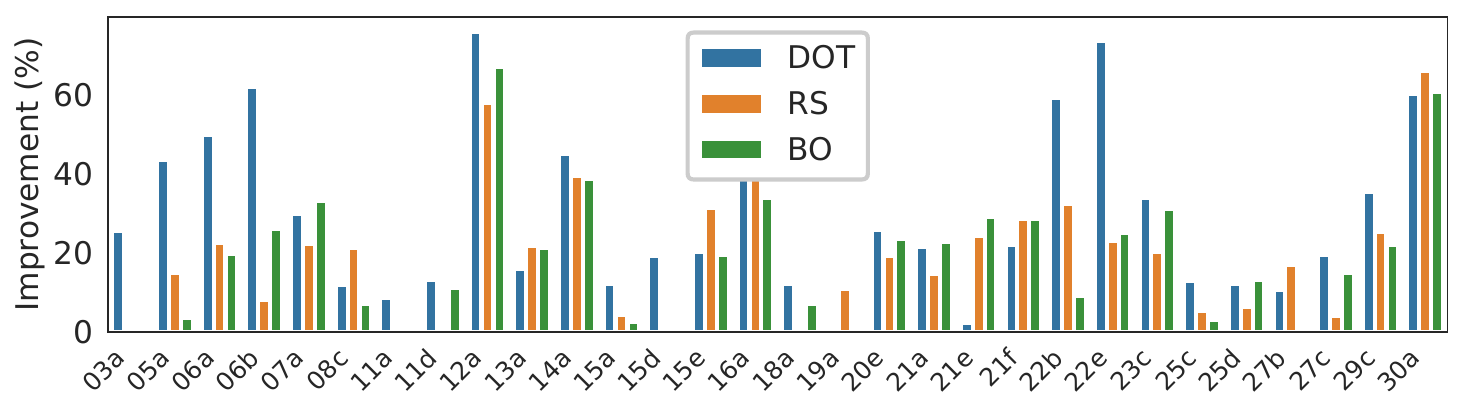}}
\subfigure[\textbf{STATS}, Part II]{ \label{fig:improve:tuned:singlequery:stats:part-2}
    \includegraphics[width=1.0\columnwidth]{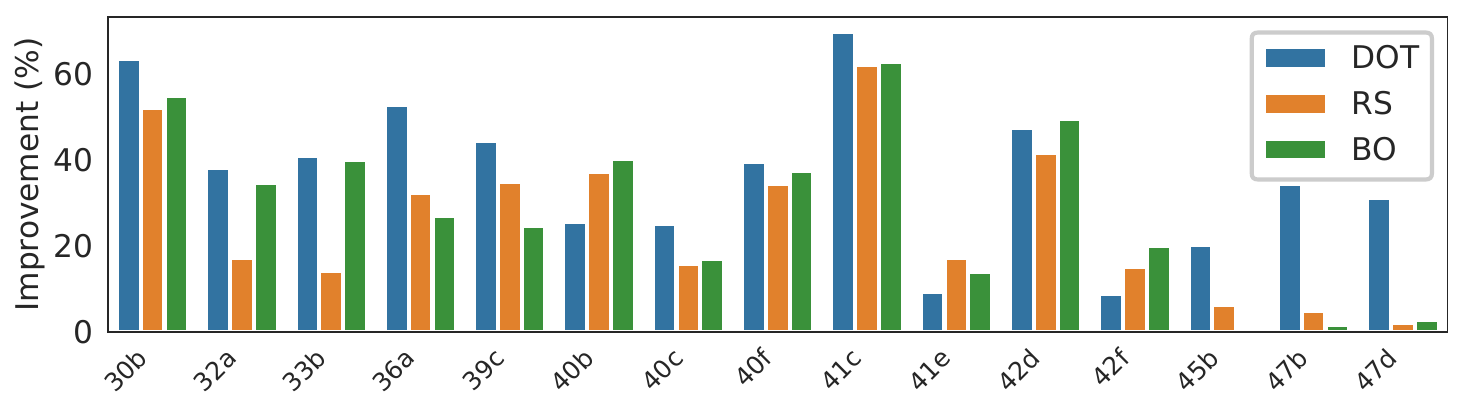}}
\vspace{-1.5em}
\caption{\rv{Improvements over \textbf{Default} by \textbf{RS}, \textbf{BO}, and \textbf{DOT} for individual queries on a tuned PG system.}}
\label{fig:improve:tuned:singlequery}
\end{figure}

\rtwo{Figure~\ref{fig:improve:tuned:singlequery} presents the details of the (percentage) improvements of individual queries for the tuned PG databases.}

\subsubsection{\rtwo{Detailed Results with ``Warm Cache''}}

\begin{figure}[t]
\centering
\subfigure[\textbf{TPC-H}, sf = 10]{ \label{fig:improve:warm-cache:singlequery:tpch10}
    \includegraphics[width=1.0\columnwidth]{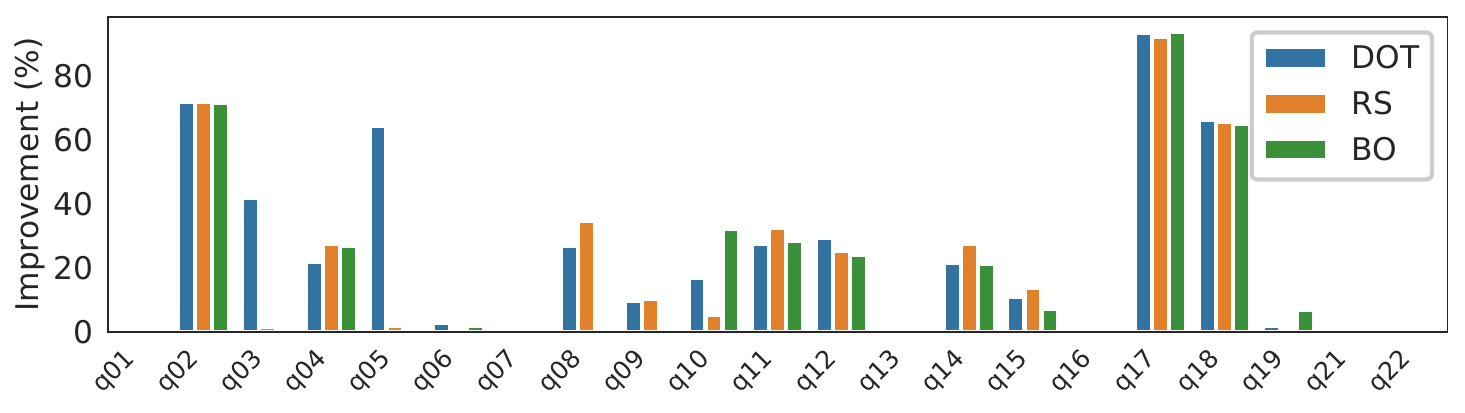}}
\subfigure[\textbf{TPC-DS}, sf = 10]{ \label{fig:improve:warm-cache:singlequery:tpcds10}
    \includegraphics[width=1.0\columnwidth]{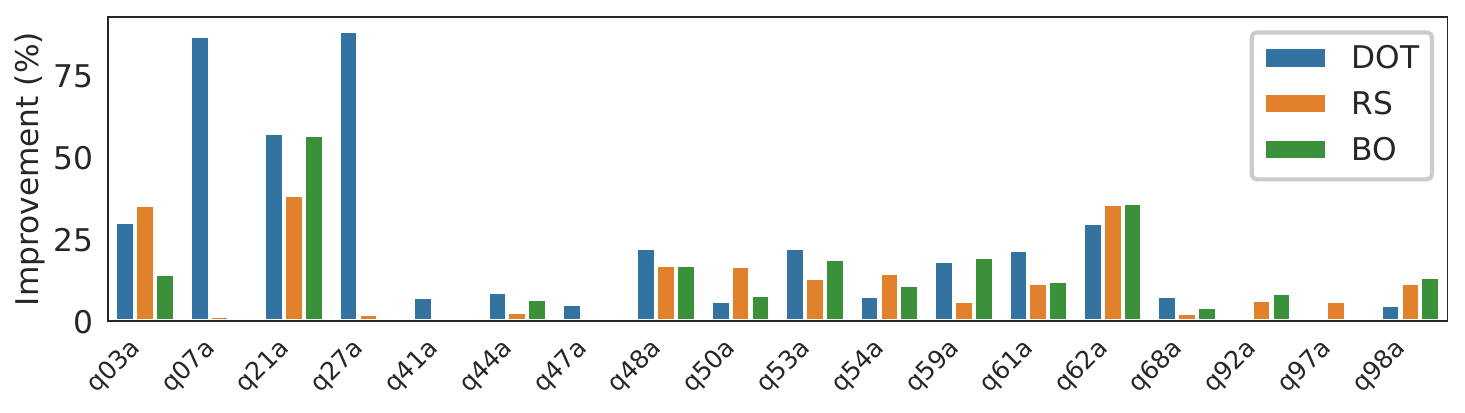}}
\subfigure[\textbf{STATS}]{ \label{fig:improve:warm-cache:singlequery:stats}
    \includegraphics[width=1.0\columnwidth]{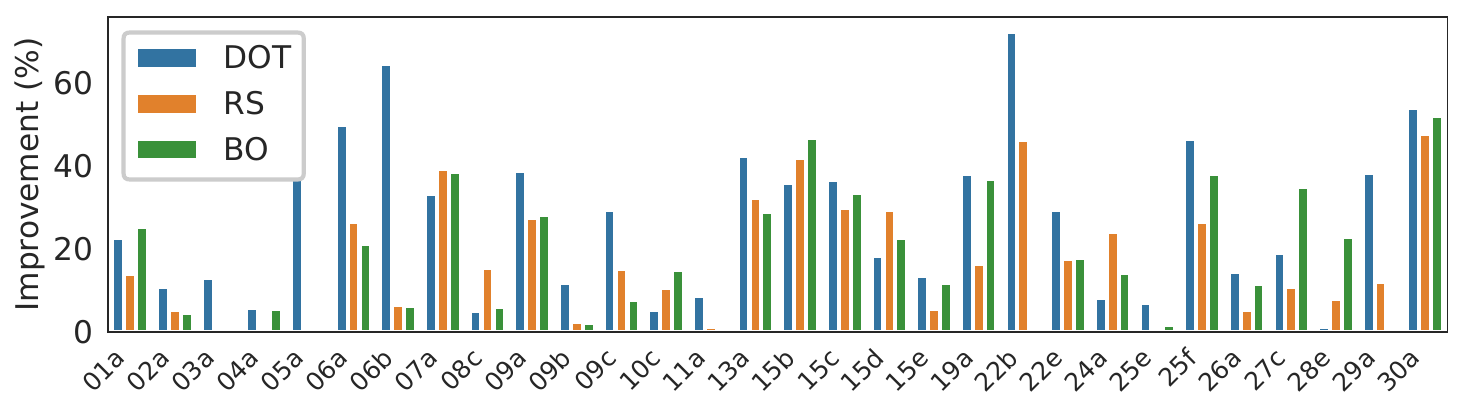}}
\vspace{-1.5em}
\caption{\rv{Improvements over \textbf{Default} by \textbf{RS}, \textbf{BO}, and \textbf{DOT} for individual queries with the ``warm cache'' scenario.}}
\label{fig:improve:warm-cache:singlequery}
\vspace{0.5em}
\end{figure}

\rtwo{Figure~\ref{fig:improve:warm-cache:singlequery} presents the details of the (percentage) improvements of individual queries with the ``warm cache'' scenario where the database resides in memory.}

\subsubsection{\rthree{Detailed Results with Larger Databases}}

\begin{figure}[t]
\centering
\subfigure[\textbf{TPC-H}, sf = 25]{ \label{fig:improve:scale:singlequery:tpch25}
    \includegraphics[width=1.0\columnwidth]{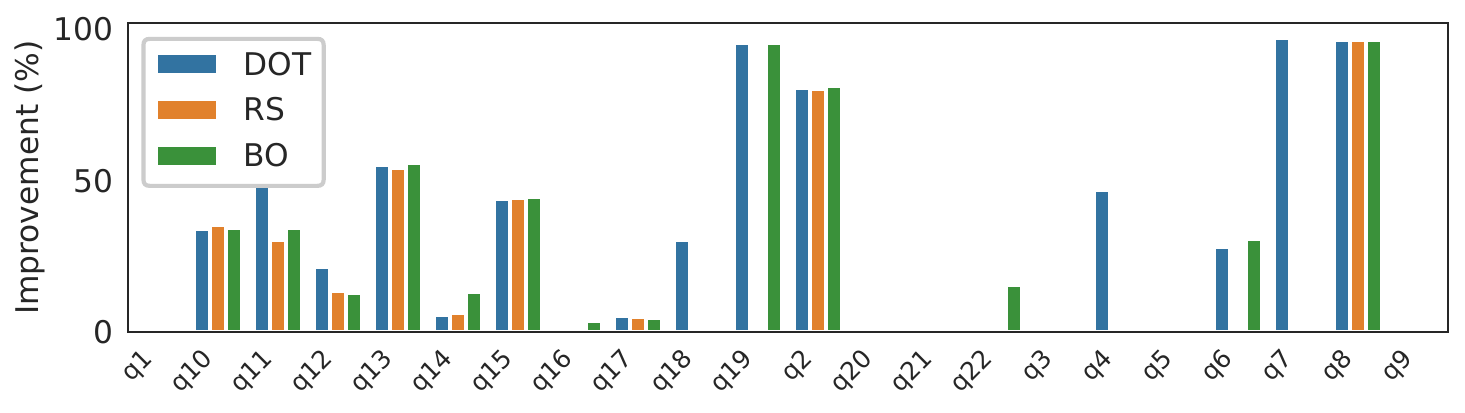}}
\subfigure[\textbf{TPC-DS}, sf = 25]{ \label{fig:improve:scale:singlequery:tpcds25}
    \includegraphics[width=1.0\columnwidth]{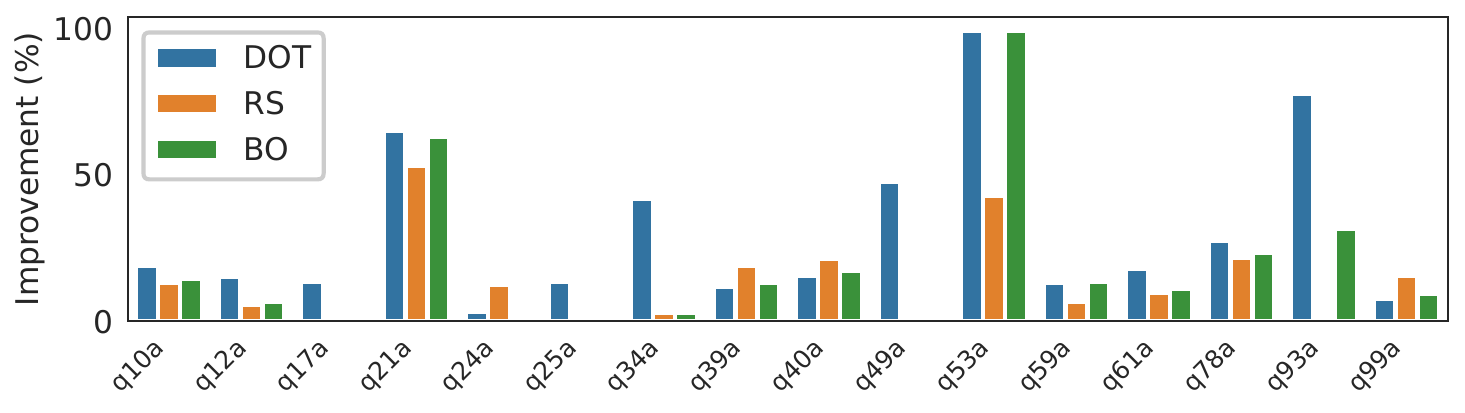}}
\subfigure[\textbf{TPC-H}, sf = 100]{ \label{fig:improve:scale:singlequery:tpch100}
    \includegraphics[width=1.0\columnwidth]{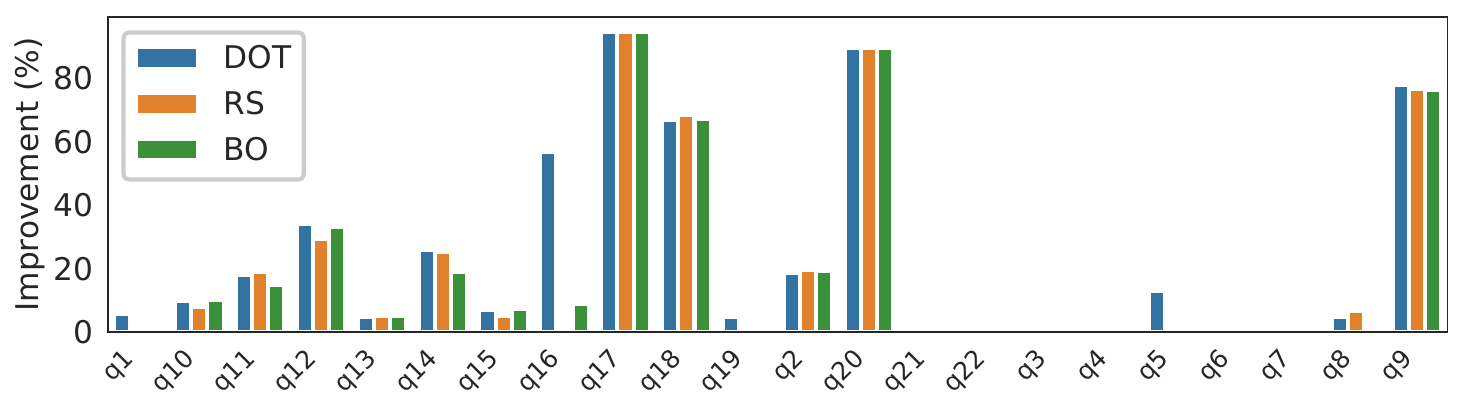}}
\vspace{-1.5em}
\caption{\rthree{Improvements over \textbf{Default} by \textbf{RS}, \textbf{BO}, and \textbf{DOT} for individual queries with larger scale databases.}}
\label{fig:improve:scale:singlequery}
\end{figure}

\rthree{Figure~\ref{fig:improve:scale:singlequery} presents the details of the (percentage) improvements of individual queries with larger scale databases.}

\subsection{\rthree{Detailed Results with Microsoft SQL Server}}

\begin{figure}[t]
\centering
\subfigure[\textbf{TPC-H}, sf = 10]{ \label{fig:improve:mssql:singlequery:tpch10}
    \includegraphics[width=1.0\columnwidth]{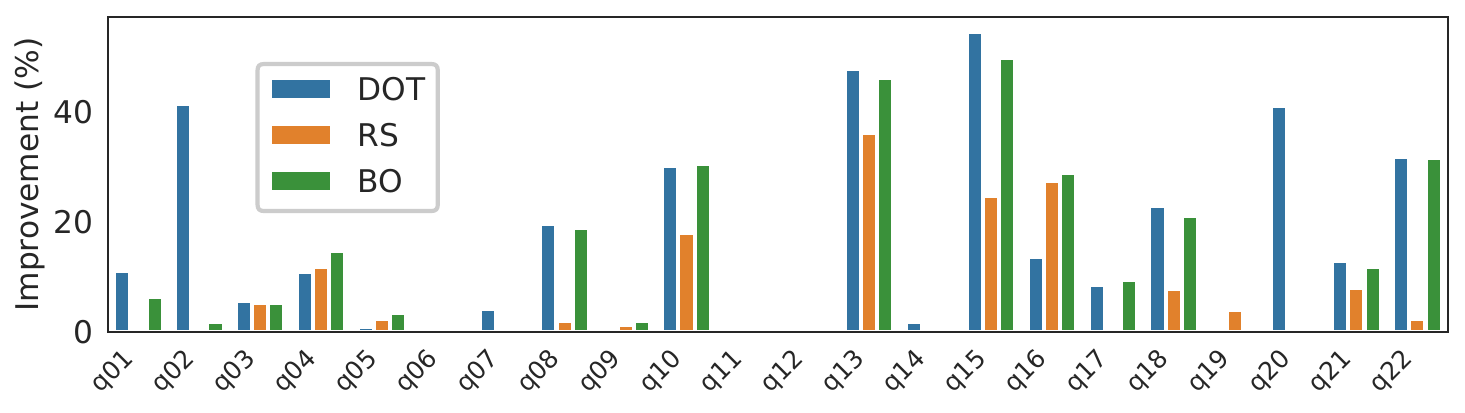}}
\subfigure[\textbf{TPC-DS}, sf = 10, Part I]{ \label{fig:improve:mssql:singlequery:tpcds10:part-1}
    \includegraphics[width=1.0\columnwidth]{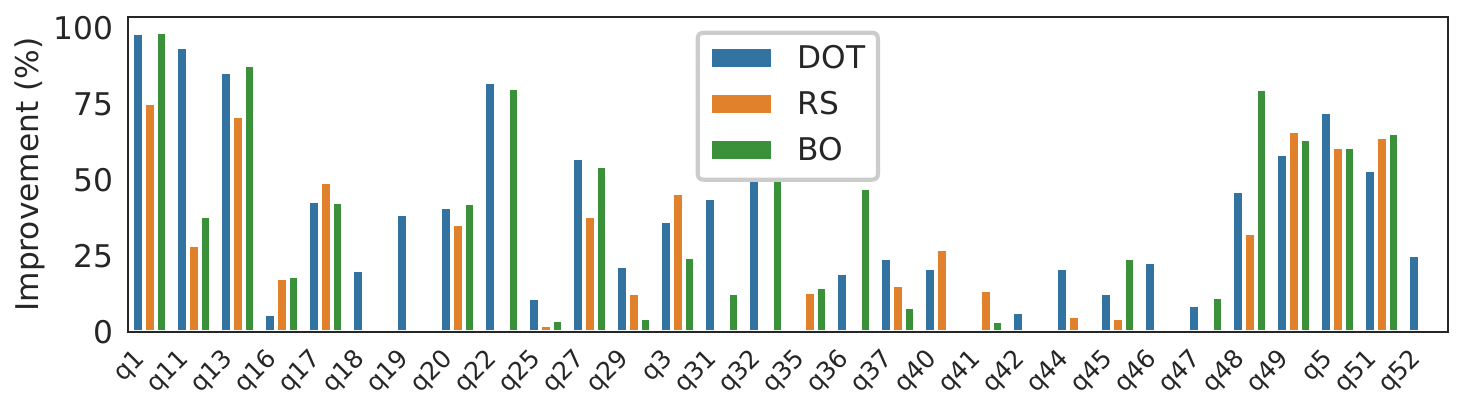}}
\subfigure[\textbf{TPC-DS}, sf = 10, Part II]{ \label{fig:improve:mssql:singlequery:tpcds10:part-2}
    \includegraphics[width=1.0\columnwidth]{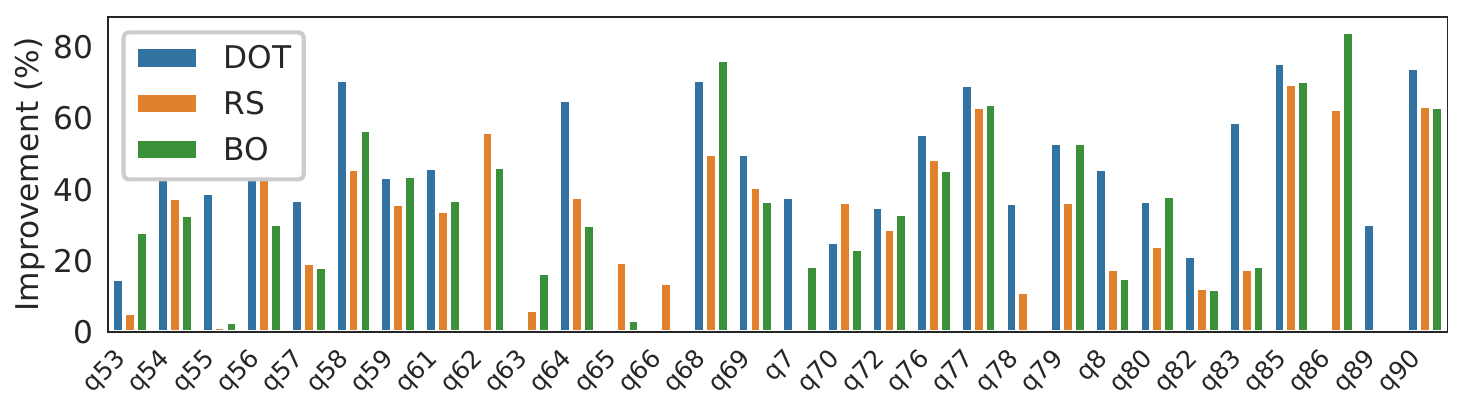}}
\subfigure[\textbf{TPC-DS}, sf = 10, Part III]{ \label{fig:improve:mssql:singlequery:tpcds10:part-3}
    \includegraphics[width=1.0\columnwidth]{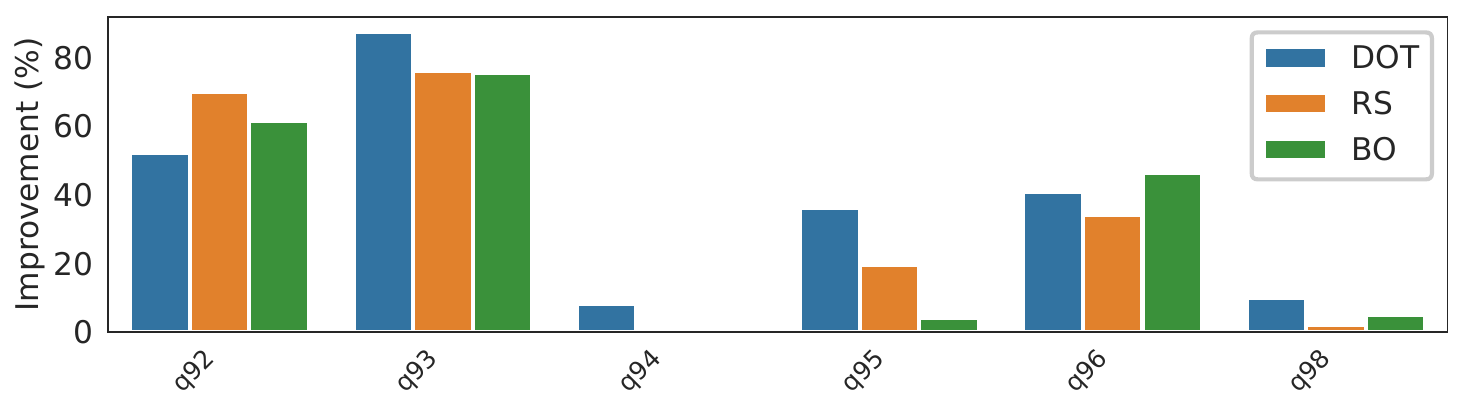}}
\subfigure[\textbf{JOB}, Part I]{ \label{fig:improve:mssql:singlequery:job:part-1}
    \includegraphics[width=1.0\columnwidth]{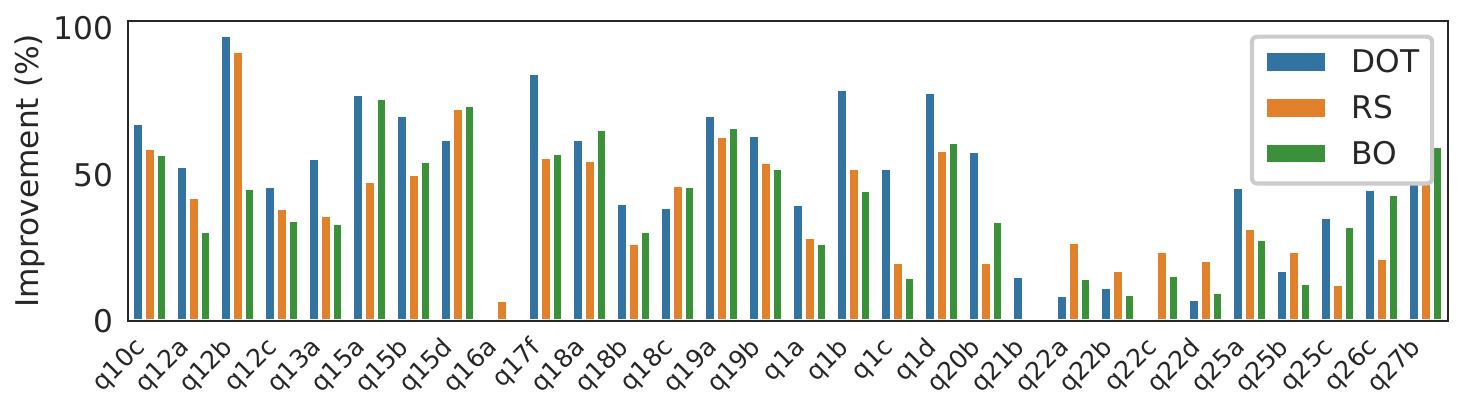}}
\subfigure[\textbf{JOB}, Part II]{ \label{fig:improve:mssql:singlequery:job:part-2}
    \includegraphics[width=1.0\columnwidth]{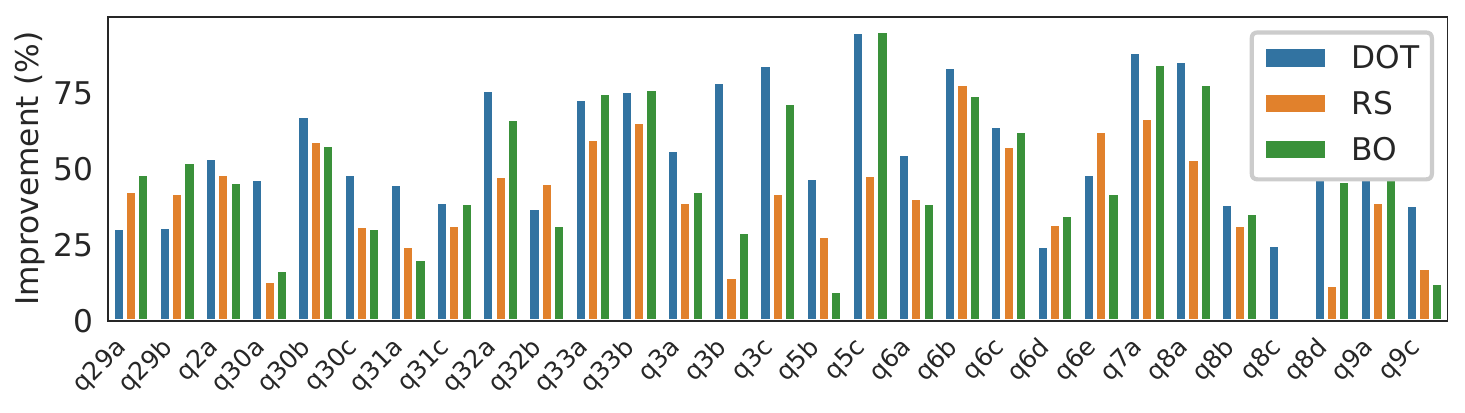}}
\vspace{-1.5em}
\caption{\rthree{Improvements over \textbf{Default} by \textbf{RS}, \textbf{BO}, and \textbf{DOT} for individual queries with Microsoft SQL Server.}}
\label{fig:improve:mssql:singlequery}
\vspace{0.5em}
\end{figure}

\rthree{Figure~\ref{fig:improve:mssql:singlequery} presents the details of the (percentage) improvements of individual queries when evaluated on top of Microsoft SQL Server.} 

\end{document}